\documentclass[journal,twocolumn,nofootinbib]{IEEEtran}
\pdfoutput=1

\usepackage{times}
\usepackage[T1]{fontenc}
\usepackage{microtype}
\usepackage[utf8]{inputenc}
\usepackage{amsmath,amsthm,amsfonts,amsbsy,amssymb}
\usepackage{newtxmath,mathrsfs,bbold,bbm,dsfont}
\usepackage{enumitem}
\usepackage{graphicx,float}
\usepackage[caption=false]{subfig}
\usepackage{array,longtable,booktabs,makecell}
\usepackage{tikz,xcolor}
\usepackage{verbatim,algpseudocode,listings}
\usepackage{csquotes}
\usepackage{hyperref}
\usepackage{lipsum}

\newtheorem{theorem}{Theorem}
\newtheorem{corollary}[theorem]{Corollary}

\newtheorem{lemma}[theorem]{Lemma}
\newtheorem{proposition}[theorem]{Proposition}

\theoremstyle{definition}
\newtheorem{remark}[theorem]{Remark}
\newtheorem{example}[theorem]{Example}

\newcommand{\1}{I}
\newcommand{\de}{\mathrm{d}}
\newcommand{\aff}{\mathrm{aff}}
\newcommand{\itr}{\mathrm{int}}
\newcommand{\supp}{\mathrm{supp}}
\newcommand{\tr}{\mathrm{Tr}}
\DeclareMathOperator*{\argmax}{argmax}
\DeclareMathOperator*{\argmin}{argmin}

\newcommand{\abb}[1]{{\textnormal{#1}}} 
\newcommand{\en}[1]{{\mathscr{#1}}} 
\newcommand{\ch}[1]{{\mathcal{#1}}} 
\newcommand{\cq}[1]{{\hat{#1}}} 
\newcommand{\g}[1]{{\widehat{#1}}} 
\newcommand{\gd}[1]{{\boldsymbol{#1}}} 
\newcommand{\s}[1]{{\mathsf{#1}}} 
\newcommand{\sw}[1]{{\widetilde{#1}}} 
\newcommand{\spa}[1]{{\mathds{#1}}} 

\newcommand{\bra}[1]{\langle#1\rvert}
\newcommand{\ket}[1]{\lvert#1\rangle}
\newcommand{\ip}[2]{\langle#1|#2\rangle}
\newcommand{\op}[2]{\ket{#1}\!\bra{#2}}

\newcommand{\mclose}{\mathclose{}}
\newcommand{\fleft}{\mathopen{}\left}
\newcommand{\fright}{\aftergroup\mclose\right}

\definecolor{darkblue}{rgb}{0,0,0.6}
\definecolor{darkgreen}{rgb}{0,0.45,0.1}
\definecolor{darkred}{rgb}{0.5,0,0}
\hypersetup{
	colorlinks = true,
	citecolor = darkgreen, 
	linkcolor = darkblue, 
	urlcolor  = darkblue, 
	filecolor = darkred 
}

\usetikzlibrary{quantikz2}
\allowdisplaybreaks

\begin{document}


\title{Barycentric Bounds on the Error Exponents of Quantum Hypothesis Exclusion

\author{Kaiyuan Ji, Hemant K. Mishra, Milán Mosonyi, and Mark M. Wilde}
\thanks{Kaiyuan Ji is with the School of Electrical and Computer Engineering, Cornell University, Ithaca, New York 14850, USA (email: kj264@cornell.edu).}
\thanks{Hemant K. Mishra was with the School of Electrical and Computer Engineering, Cornell University, Ithaca, New York 14850, USA.  He is now with the Department of Mathematics and Computing, Indian Institute of Technology (Indian School of Mines) Dhanbad, Jharkhand 826004, India (email: hemantmishra1124@iitism.ac.in).}
\thanks{Milán Mosonyi is with the HUN-REN Alfréd Rényi Institute of Mathematics, Reáltanoda Street 13-15, H-1053, Budapest, Hungary, and also with the Department of Analysis and Operations Research, Institute of Mathematics, Budapest University of Technology and Economics, Műegyetem rkp.\ 3., H-1111 Budapest, Hungary (email: milan.mosonyi@gmail.com).}
\thanks{Mark M. Wilde is with the School of Electrical and Computer Engineering, Cornell University, Ithaca, New York 14850, USA (email: wilde@cornell.edu).}
}

\maketitle



\begin{abstract}
Quantum state exclusion is an operational task with application to ontological interpretations of quantum states.  In such a task, one is given a system whose state is randomly selected from a finite set, and the goal is to identify a state from the set that is \emph{not} the true state of the system.  An error occurs if and only if the state identified is the true state.  In this paper, we study the optimal error probability of quantum state exclusion and its error exponent from an information-theoretic perspective.  Our main finding is a single-letter upper bound on the error exponent of state exclusion given by the multivariate log-Euclidean Chernoff divergence, and we prove that this improves upon the best previously known upper bound.  We also extend our analysis to quantum channel exclusion, and we establish a single-letter and efficiently computable upper bound on its error exponent, admitting the use of adaptive strategies.  We derive both upper bounds, for state and channel exclusion, based on one-shot analysis and formulate them as a type of multivariate divergence measure called a barycentric Chernoff divergence.  Moreover, our result on channel exclusion has implications in two important special cases.  First, when there are two hypotheses, our result provides the first known efficiently computable upper bound on the error exponent of symmetric binary channel discrimination.  Second, when all channels are classical, we show that our upper bound is achievable by a parallel strategy, thus solving the exact error exponent of classical channel exclusion.
\end{abstract}




\tableofcontents

\section{Introduction}
\label{sec:introduction}

Hypothesis testing has a fundamental and ubiquitous role in quantum information theory~\cite{wilde2017QuantumInformationTheory,hayashi2017QuantumInformationTheory,watrous2018TheoryQuantumInformation,khatri2024PrinciplesQuantumCommunication}.  In most contexts, it takes the form of a \emph{discrimination task}, in which an experimenter tries to figure out the true hypothesis among a number of candidates.  However, in some other contexts, it may be of interest to consider variations of hypothesis testing in which one is not expected to figure out the true hypothesis completely but is only required to take a meaningful step towards determining the truth.  The most lenient such variation is an \emph{exclusion task} (which may also be referred to as ``hypothesis exclusion''), where the experimenter is merely asked to choose one false hypothesis to rule out.\footnote{Throughout this paper, by ``hypothesis exclusion'' we specifically mean $1$-hypothesis exclusion.  More generally, one can conceive of a task of $m$-hypothesis exclusion, where the experimenter is asked to choose $m$ false hypotheses (instead of just one) to rule out~\cite{bandyopadhyay2014ConclusiveExclusionQuantum}.}  For instance, in quantum state exclusion, the experimenter is given a quantum system and is asked to choose a state from a given tuple $(\rho_1,\rho_2,\dots,\rho_r)$ such that the chosen state is not the true state of the system.  The given set of states is said to be ``perfectly antidistinguishable'' whenever perfect exclusion can be achieved, that is, whenever there exists a measurement such that on observing any of the possible measurement outcomes, the experimenter is able to rule out some state with certainty.  Mathematically, this amounts to the existence of a positive operator-valued measure (POVM) $(\Lambda_1,\Lambda_2,\dots,\Lambda_r)$ such that $\tr[\Lambda_x\rho_x]=0$ for all $x\in\{1,2,\dots,r\}$.

\subsection{Literature review}
\label{sec:literature}

Early treatments of quantum state exclusion (and state antidistinguishability) can be traced back to the study of compatibility of quantum-state assignments~\cite{brun2002HowMuchState,caves2002ConditionsCompatibilityQuantumstate}.  To elaborate, consider a set of quantum states, each representing the belief of a party about a common system.  If the set of states is perfectly antidistinguishable, then we know for sure that at least one party's belief is inconsistent with the reality, regardless of the true state of the system.  This is because, once the measurement that achieves perfect exclusion is applied to the system, by definition, the measurement outcome would always contradict some party's belief in the sense that this party would not believe such an outcome is possible.  This reveals that state antidistinguishability indicates ``incompatibility'' of beliefs.  Understood as a test of belief incompatibility, state exclusion then found applications in the ontological interpretation of quantum states~\cite{pusey2012RealityQuantumState,leifer2014QuantumStateReal,barrett2014NoPsepistemicModel}.  Specifically, it serves as a crucial ingredient in the proof of the seminal Pusey--Barrett--Rudolph theorem~\cite{pusey2012RealityQuantumState}, which, under fairly minimal assumptions, disproves $\psi$-epistemic models of quantum theory and essentially concludes that distinct pure quantum states correspond to distinct objective realities rather than different beliefs about a common reality.  

A remarkable fact revealed by the aforementioned studies about state exclusion is that, in contrast to perfect distinguishability, perfect antidistinguishability does not require orthogonality in any sense as a precondition, even for pure states~\cite{caves2002ConditionsCompatibilityQuantumstate,pusey2012RealityQuantumState}.  This stimulated a number of works focusing on characterizing the conditions for perfect exclusion between pure states~\cite{caves2002ConditionsCompatibilityQuantumstate,bandyopadhyay2014ConclusiveExclusionQuantum,heinosaari2018AntidistinguishabilityPureQuantum,crickmore2020UnambiguousQuantumState,russo2023InnerProductsPure,johnston2025TightBoundsAntidistinguishability}.  It also inspired a spectrum of applications of state exclusion in constructing noncontextuality inequalities~\cite{leifer2020NoncontextualityInequalitiesAntidistinguishability}, proving advantage in communication complexity~\cite{havlicek2020SimpleCommunicationComplexity}, and devising new cryptographic schemes~\cite{collins2014RealizationQuantumDigital,amiri2021Imperfect1outof2Quantum}.  Exclusion between mixed states with finite errors was much less studied in comparison, but in this setting the task is useful in providing operational interpretations for a type of so-called ``weight-based'' resource measures in quantum resource theories~\cite{ducuara2020OperationalInterpretationWeightbased,uola2020AllQuantumResources}.  

The information-theoretic treatment of state exclusion, especially the study of its asymptotic features, was initiated recently~\cite{mishra2024OptimalErrorExponents}.  The authors of Ref.~\cite{mishra2024OptimalErrorExponents} proved that the asymptotic error exponent of classical state exclusion is exactly characterized by the multivariate classical Chernoff divergence and established various lower and upper bounds on the error exponent of general quantum state exclusion, including a single-letter upper bound that is computable via a semidefinite program (SDP).  This related the task of state exclusion with multivariate divergence measures and, in particular, with a recently proposed methodology for constructing such measures, known as a ``barycentric Rényi divergence''~\cite{mosonyi2024GeometricRelativeEntropies}.

\subsection{Main results}
\label{sec:main}

In this paper, we make a number of contributions to the information-theoretic study of imperfect exclusion between general quantum states, as well as exclusion between quantum channels.
\begin{itemize}
	\item For quantum state exclusion, our main finding is a single-letter upper bound on the asymptotic error exponent, given by the multivariate log-Euclidean Chernoff divergence (Theorem~\ref{thm:asymptotic-state}):
	\begin{align}
		\label{eq:introduction-1}
		&\limsup_{n\to\infty}-\frac{1}{n}\ln P_\abb{err}\fleft(\en{E}^n\fright) \notag\\
		&\leq C^\flat\fleft(\rho_{[r]}\fright)=\sup_{s_{[r]}\in\s{P}_r}\inf_{\tau\in\s{D}_A}\sum_{x\in[r]}s_xD\fleft(\tau\middle\|\rho_x\fright),
	\end{align}
	where $\en{E}^n$ denotes an $n$-fold ensemble of a tuple of states $\rho_{[r]}\equiv(\rho_1,\rho_2,\dots,\rho_r)$.  The quantity on the right-hand side of \eqref{eq:introduction-1} is termed the barycentric Chernoff divergence of $\rho_{[r]}$ based on the Umegaki divergence $D$, with the supremum over probability distributions and the infimum over states.  This upper bound is tight for classical states, and we further show that it improves upon the aforementioned SDP upper bound in Ref.~\cite{mishra2024OptimalErrorExponents} for general quantum states.
	\item For quantum channel exclusion, we establish a single-letter upper bound on the asymptotic error exponent, even assuming that adaptive strategies are allowed (Theorem~\ref{thm:asymptotic-channel}):
	\begin{align}
		\label{eq:introduction-2}
		&\limsup_{n\to\infty}-\frac{1}{n}\ln P_\abb{err}\fleft(n;\en{N}\fright) \notag\\
		&\leq\sup_{s_{[r]}\in\s{P}_r}\inf_{\ch{T}\in\s{C}_{A\to B}}\sum_{x\in[r]}s_x\g{D}\fleft(\ch{T}\middle\|\ch{N}_x\fright),
	\end{align}
	where $\en{N}$ denotes an ensemble of a tuple of channels $\ch{N}_{[r]}\equiv(\ch{N}_1,\ch{N}_2,\dots,\ch{N}_r)$.  The quantity on the right-hand side of \eqref{eq:introduction-2} is termed the barycentric Chernoff channel divergence of $\ch{N}_{[r]}$ based on the Belavkin--Staszewski divergence $\g{D}$, with the infimum over channels.
	\item We show that for classical channels, the upper bound in \eqref{eq:introduction-2} is achievable by a parallel strategy and thus characterizes the exact asymptotic error exponent of classical channel exclusion (Theorem~\ref{thm:asymptotic-channel-classical}):
	\begin{align}
		\label{eq:introduction-3}
		\lim_{n\to\infty}-\frac{1}{n}\ln P_\abb{err}\fleft(n;\en{P}\fright)&=\max_{y\in[d_Y]}C\fleft(\varrho_{y,[r]}\fright),
	\end{align}
	where $\en{P}$ denotes an ensemble of a tuple of classical channels, $\varrho_{y,[r]}\equiv(\varrho_{y,1},\varrho_{y,2},\dots,\varrho_{y,r})$ denotes the tuple of classical output states of the channels on the classical input $y$, and $C$ is the multivariate classical Chernoff divergence.
	\item Given that channel exclusion encompasses symmetric binary channel discrimination~\cite{harrow2010AdaptiveNonadaptiveStrategies,wilde2020AmortizedChannelDivergence} as a special case, our upper bound in \eqref{eq:introduction-2} for general quantum channels also provides the first known efficiently computable upper bound for symmetric binary channel discrimination.  Furthermore, our characterization of the error exponent of classical channel exclusion in \eqref{eq:introduction-3} generalizes Hayashi's result on symmetric binary classical channel discrimination~\cite{hayashi2009DiscriminationTwoChannels}.
\end{itemize}
Our upper bounds in \eqref{eq:introduction-1} and \eqref{eq:introduction-2} are both derived based on one-shot analysis, and they both have the desirable features of being single-letter and efficiently computable.

A distinctive aspect of our approach is the substantial use of an \emph{extended} notion of divergence measures, as considered previously in the context of entanglement theory \cite{wang2020AlogarithmicNegativity}.  For an extended divergence, we do not require its first argument to be a quantum state but allow it more generally to be an affine combination of states or even a general Hermitian operator.  The use of extended divergences is essential to our derivation of the new bounds, as well as a comparison with old ones, in both the nonasymptotic and asymptotic regimes of quantum state exclusion.  Since the need for extended divergences arises only in the quantum setting (and is absent in the classical setting), it marks a key distinction between the classical and quantum theories of information.

\begin{table*}[t]
\centering
\caption{Summary of notation and corresponding definitions}
\label{tab:notation}
\renewcommand{\arraystretch}{1.2}
\begin{tabular}
[c]{c|l|l}
\toprule
Symbol & Meaning & Definition \\\midrule
$[r]$ & set of $r$ smallest, distinct positive integers & $\{1,2,\dots,r\}$ \\
$\gamma_{[r]}$ & a tuple of entities with indices from $[r]$ & $(\gamma_1,\gamma_2,\dots,\gamma_r)$ \\
$\s{S}^{[r]}$ & set of tuples whose entities belong to a set $\s{S}$ & $\{\gamma_{[r]}\colon\gamma_x\in\s{S}\;\forall x\in[r]\}$ \\\midrule
$\spa{H}_A$ & Hilbert space associated with a quantum system $A$ & \\
$d_A$ & dimension of $A$ & $\dim(\spa{H}_A)$ \\
$\spa{B}_A$ & space of bounded operators acting on $\spa{H}_A$ & \\
$\s{Herm}_A$ & space of Hermitian operators in $\spa{B}_A$ & $\{\gamma\in\spa{B}_A\colon\gamma=\gamma^\dagger\}$ \\
$\s{PSD}_A$ & space of positive semidefinite operators in $\spa{B}_A$ & $\{\gamma\in\s{Herm}_A\colon\bra{\psi}\gamma\ket{\psi}\in[0,\infty)\;\forall\ket{\psi}\in\spa{H}_A\}$ \\
$\s{D}_A$ & set of quantum states of $A$ & $\{\rho\in\s{PSD}_A\colon\tr[\rho]=1\}$ \\
$\aff(\s{D}_A)$ & set of unit-trace Hermitian operators in $\spa{B}_A$ & $\{\tau\in\s{Herm}_A\colon\tr[\tau]=1\}$ \\
$\s{CP}_{A\to B}$ & set of completely positive maps from $\spa{B}_A$ to $\spa{B}_B$ & $\{\ch{M}\colon\spa{B}_A\to\spa{B}_B,\;\ch{M}\text{ is completely positive}\}$ \\
$\s{C}_{A\to B}$ & set of quantum channels from $A$ to $B$ & $\{\ch{N}\in\s{CP}_{A\to B},\;\tr\circ\ch{N}=\tr\}$ \\
$\s{M}_{A,r}$ & set of POVMs on $A$ with $r$ outcomes & $\{\Lambda_{[r]}\colon\Lambda_x\in\s{PSD}_A\;\forall x\in[r],\;\sum_{x\in[r]}\Lambda_x=\1\}$ \\
$\s{P}_r$ & set of probability distributions over $[r]$ & $\{s_{[r]}\colon s_x\in[0,1]\;\forall x\in[r],\;\sum_{x\in[r]}s_x=1\}$ \\
$\itr(\s{P}_r)$ & set of probability distributions whose support is $[r]$ & $\{p_{[r]}\in\s{P}_r\colon p_x\in(0,1)\;\forall x\in[r]\}$ 
\\\bottomrule
\end{tabular}
\end{table*}

\subsection{Outline of the paper}
\label{sec:outline}

The rest of our paper is outlined as follows.  In Section~\ref{sec:preliminaries}, we begin by establishing notation (Section~\ref{sec:notation}), followed by a review of the definitions and properties of various divergence measures for states (Section~\ref{sec:divergence-state}) and channels (Section~\ref{sec:divergence-channel}).  We also discuss a class of multivariate divergence measures known as left divergence radii in Section~\ref{sec:radius}.  In Section~\ref{sec:extension}, we systematically investigate the mathematical properties of extensions of the sandwiched Rényi divergence and the hypothesis-testing divergence, as a technical preparation for applications in subsequent sections.  Our study of quantum state exclusion is presented in Section~\ref{sec:exclusion-state}.  Section~\ref{sec:setting-state} explains the setting of the task in detail and formally defines its error probability and error exponent.  Sections~\ref{sec:oneshot-state} and \ref{sec:asymptotic-state} contain the one-shot and asymptotic analyses, respectively, which lead to our log-Euclidean upper bound on the error exponent of state discrimination.  In Section~\ref{sec:exclusion-channel}, we extend our study to  quantum channel exclusion.  The setting of the task is explained in Section~\ref{sec:setting-channel}, and the pertaining results are presented in Sections~\ref{sec:nonasymptotic-channel} and \ref{sec:asymptotic-channel}.  In Section~\ref{sec:conclusion}, we summarize our findings (Section~\ref{sec:summary}) and discuss avenues of future research (Section~\ref{sec:future}).  Appendices~\ref{app:divergence}--\ref{app:indefinite} complement the main text with proofs of technical lemmas and additional discussions.

\section{Preliminaries}
\label{sec:preliminaries}

\subsection{Notation}
\label{sec:notation}

We provide a summary of notation and corresponding mathematical definitions in Table~\ref{tab:notation}.

Let $[r]\equiv\{1,2,\dots,r\}$ denote the set of the $r$ smallest distinct positive integers.  Throughout the paper, we use $\gamma_{[r]}\equiv(\gamma_1,\gamma_2,\dots,\gamma_r)$ to denote a tuple of entities with indices from $[r]$, regardless of the nature of these entities.  We also use $\s{S}^{[r]}\equiv\{\gamma_{[r]}\colon\gamma_x\in\s{S}\;\forall x\in[r]\}$ to denote the set of tuples each of whose entities belongs to a set $\s{S}$.

Let $\spa{H}_A$ denote the (finite-dimensional) Hilbert space associated with a quantum system $A$, and let $d_A\equiv\dim(\spa{H}_A)$ denote the dimension of $A$.  Let $\spa{B}_A$ denote the space of bounded operators acting on $\spa{H}_A$.  Let $\s{Herm}_A$ and $\s{PSD}_A$ denote the set of Hermitian operators and the set of positive semidefinite operators in $\spa{B}_A$, respectively.  An operator inequality $\gamma\geq\sigma$ regarding $\gamma,\sigma\in\s{Herm}_A$ should be understood as $\gamma-\sigma\in\s{PSD}_A$.  For a Hermitian operator $\gamma\in\s{Herm}_A$, let $\supp(\gamma)$ denote its support, and let $\gamma^0\in\s{PSD}_A$ denote the projector onto $\supp(\gamma)$.  For a tuple of Hermitian operators $\gamma_{[r]}\in\s{Herm}_A^{[r]}$, let $\bigwedge_{x\in[r]}\gamma_x^0\in\s{PSD}_A$ denote the projector onto $\bigcap_{x\in[r]}\supp(\gamma_x)$.  For a positive semidefinite operator $\sigma\in\s{PSD}_A$, the negative power $\sigma^{-\alpha}$ for $\alpha\in(0,\infty)$ and the logarithm $\ln\sigma$ should both be understood as taken on $\supp(\sigma)$.  

A quantum state is a unit-trace positive semidefinite operator.  Let $\s{D}_A$ denote the set of quantum states of a system $A$.  Let $\aff(\s{D}_A)$ denote the set of unit-trace Hermitian operators in $\spa{B}_A$, which is also the affine hull of $\s{D}_A$.

A quantum channel is a completely positive trace-preserving (CPTP) linear map.  Let $\s{CP}_{A\to B}$ and $\s{C}_{A\to B}$ denote the set of completely positive linear maps from $\spa{B}_A$ to $\spa{B}_B$ and the set of quantum channels from $A$ to $B$, respectively.  Every quantum channel realizes a transformation between quantum states; that is, for every $\ch{N}\in\s{C}_{A\to B}$ and $\rho\in\s{D}_{RA}$, we always have that $\ch{N}_{A\to B}[\rho_{RA}]\in\s{D}_{RB}$, where $RA$ is the composite system identified by the Hilbert space $\spa{H}_{RA}\equiv\spa{H}_R\otimes\spa{H}_A$ and similarly for $RB$.  Let $J_\ch{M}\equiv\sum_{i,j\in[d_A]}\op{i}{j}_R\otimes\ch{M}_{A\to B}[\op{i}{j}_A]\in\s{PSD}_{RB}$ denote the Choi operator of a completely positive map $\ch{M}\in\s{CP}_{A\to B}$, where $R$ is a system such that $d_R=d_A$.

A POVM is a tuple of positive semidefinite operators summing to the identity operator.  Let $\s{M}_{A,r}$ denote the set of POVMs on a system $A$ with $r$ possible outcomes.  According to the Born rule, when applying a POVM $\Lambda_{[r]}\in\s{M}_{A,r}$ to a state $\rho\in\s{D}_A$, the probability of obtaining the $x$th outcome is given by $\tr[\Lambda_x\rho]$ for each $x\in[r]$.  Every POVM $\Lambda_{[r]}$ can be understood as representing a quantum-to-classical channel:
\begin{align}
	\ch{M}\in\s{C}_{A\to X}&\colon\rho\mapsto\sum_{x\in[r]}\tr\fleft[\Lambda_x\rho\fright]\op{x}{x}.
\end{align}

Let $\s{P}_r$ denote the set of probability distributions over $[r]$.  Let $\itr(\s{P}_r)$ denote the set of probability distributions whose support is $[r]$, which is also the interior of $\s{P}_r$.

\subsection{Bivariate divergence measures for states}
\label{sec:divergence-state}

We provide a glossary of divergence measures in Table~\ref{tab:divergence}.

\begin{table*}[t]
\centering
\caption{Glossary of divergence measures}
\label{tab:divergence}
\renewcommand{\arraystretch}{1.2}
\begin{tabular}
[c]{c|l|l}
\toprule
Symbol & Meaning & Definition \\\midrule
$\gd{D}(\rho\|\sigma)$ & generalized divergence & Eqs.~\eqref{eq:divergence} and \eqref{eq:DPI} for $\rho\in\s{D}_A$, $\sigma\in\s{PSD}_A$ \\
$\sw{D}_\alpha(\rho\|\sigma)$ & sandwiched Rényi divergence & Eq.~\eqref{eq:sandwiched} for $\alpha\in(1,\infty)$, $\rho\in\s{D}_A$, $\sigma\in\s{PSD}_A$ \\
$D(\rho\|\sigma)$ & Umegaki divergence & Eq.~\eqref{eq:umegaki} for $\rho\in\s{D}_A$, $\sigma\in\s{PSD}_A$ \\
$\g{D}_\alpha(\rho\|\sigma)$ & geometric Rényi divergence & Eq.~\eqref{eq:geometric} for $\alpha\in(1,2]$, $\rho\in\s{D}_A$, $\sigma\in\s{PSD}_A$ \\
$\g{D}(\rho\|\sigma)$ & Belavkin--Staszewski divergence & Eq.~\eqref{eq:belavkin} for $\rho\in\s{D}_A$, $\sigma\in\s{PSD}_A$ \\
$\sw{D}_\alpha(\gamma\|\sigma)$ & extended sandwiched Rényi divergence & Eq.~\eqref{eq:sandwiched-extended} for $\alpha\in(1,\infty)$, $\gamma\in\s{Herm}_A$, $\sigma\in\s{PSD}_A$ \\
$D_{\max}(\gamma\|\sigma)$ & extended max-divergence & Eq.~\eqref{eq:max-extended} for $\gamma\in\s{Herm}_A$, $\sigma\in\s{PSD}_A$ \\
$D_\abb{H}^\varepsilon(\tau\|\sigma)$ & extended hypothesis-testing divergence & Eq.~\eqref{eq:hypothesis-extended} for $\varepsilon\in[0,1]$, $\tau\in\aff(\s{D}_A)$, $\sigma\in\s{PSD}_A$ \\
$D_\alpha(\rho\|\sigma)$ & Petz--Rényi (pseudo-)divergence & Eq.~\eqref{eq:petz} for $\alpha\in(0,1)\cup(1,\infty)$, $\rho\in\s{D}_A$, $\sigma\in\s{PSD}_A$ \\
$\gd{D}(\ch{N}\|\ch{M})$ & channel $\gd{D}$-divergence & Eq.~\eqref{eq:divergence-channel} for $\ch{N}\in\s{C}_{A\to B}$, $\ch{M}\in\s{CP}_{A\to B}$ \\
$\g{D}_\alpha(\ch{N}\|\ch{M})$ & geometric Rényi channel divergence & Eq.~\eqref{eq:geometric-channel} for $\alpha\in(1,2]$, $\ch{N}\in\s{C}_{A\to B}$, $\ch{M}\in\s{CP}_{A\to B}$ \\
$\g{D}(\ch{N}\|\ch{M})$ & Belavkin--Staszewski channel divergence & Eq.~\eqref{eq:belavkin-channel} for $\ch{N}\in\s{C}_{A\to B}$, $\ch{M}\in\s{CP}_{A\to B}$ \\
$R^\gd{D}(\rho_{[r]})$ & left $\gd{D}$-radius & Eq.~\eqref{eq:radius-1} for $\rho_{[r]}\in\s{PSD}_A^{[r]}$ \\
$C(\varrho_{[r]})$ & multivariate classical Chernoff divergence & Eq.~\eqref{eq:chernoff-classical} for $\varrho_{[r]}\in\s{PSD}_Y^{[r]}$ \\
$C^\flat(\rho_{[r]})$ & multivariate log-Euclidean Chernoff divergence & Eq.~\eqref{eq:euclidean} for $\rho_{[r]}\in\s{PSD}_A^{[r]}$ \\
$R^\gd{D}(\ch{N}_{[r]})$ & left channel $\gd{D}$-radius & Eq.~\eqref{eq:radius-channel-1} for $\ch{N}_{[r]}\in\s{CP}_{A\to B}^{[r]}$
\\\bottomrule
\end{tabular}
\end{table*}

Divergence measures are functions that quantify the dissimilarity of states, or more generally, positive semidefinite operators acting on the same Hilbert space.  Consider a bivariate function
\begin{align}
	\label{eq:divergence}
	\gd{D}\colon\bigcup_{A}\left(\s{D}_A\times\s{PSD}_A\right)\to\spa{R}\cup\{\infty\}
\end{align}
on a pair consisting of a state and a positive semidefinite operator, where the union is over all quantum systems.  The function $\gd{D}$ is called a \emph{(bivariate) generalized divergence} whenever it obeys the \emph{data-processing inequality (DPI)}~\cite{polyanskiy2010ArimotoChannelCoding,sharma2012StrongConversesQuantum,wilde2014StrongConverseClassical}:
\begin{align}
	\label{eq:DPI}
	&\gd{D}\fleft(\ch{N}\fleft[\rho\fright]\middle\|\ch{N}\fleft[\sigma\fright]\fright)\leq\gd{D}\fleft(\rho\middle\|\sigma\fright) \notag\\
	&\qquad\forall\ch{N}\in\s{C}_{A\to B},\;\rho\in\s{D}_A,\;\sigma\in\s{PSD}_A,
\end{align}
that is, whenever its value does not increase as both of its arguments are acted upon by the same channel.  A generalized divergence $\gd{D}$ is said to be \emph{additive} whenever
\begin{align}
	\label{eq:additivity}
	&\gd{D}\fleft(\rho\otimes\rho'\middle\|\sigma\otimes\sigma'\fright)=\gd{D}\fleft(\rho\middle\|\sigma\fright)+\gd{D}\fleft(\rho'\middle\|\sigma'\fright) \notag\\
	&\qquad\forall\rho\in\s{D}_A,\;\rho'\in\s{D}_B,\;\sigma\in\s{PSD}_A,\;\sigma'\in\s{PSD}_B.
\end{align}
It is said to be \emph{jointly quasiconvex} whenever
\begin{align}
	\label{eq:quasiconvexity}
	&\gd{D}\fleft(\sum_{x\in[r]}p_x\rho_x\middle\|\sum_{x\in[r]}p_x\sigma_x\fright)\leq\max_{x\in[r]}\gd{D}\fleft(\rho_x\middle\|\sigma_x\fright) \notag\\
	&\qquad\forall\rho_{[r]}\in\s{D}_A^{[r]},\;\sigma_{[r]}\in\s{PSD}_A^{[r]},\;p_{[r]}\in\s{P}_r.
\end{align}
It is said to be \emph{monotonically nonincreasing in its second argument} whenever
\begin{align}
	\label{eq:antimonotonicity}
	\gd{D}\fleft(\rho\middle\|\sigma+\sigma'\fright)&\leq\gd{D}\fleft(\rho\middle\|\sigma\fright)\quad\forall\rho\in\s{D}_A,\;\sigma,\sigma'\in\s{PSD}_A.
\end{align}
It is said to be \emph{regular in its second argument} whenever
\begin{align}
	\label{eq:regularity}
	\gd{D}\fleft(\rho\middle\|\sigma\fright)&=\lim_{\varepsilon\searrow0}\gd{D}\fleft(\rho\middle\|\sigma+\varepsilon\1\fright)\quad\forall\rho\in\s{D}_A,\;\sigma\in\s{PSD}_A.
\end{align}

In what follows, we review the definitions and properties of several generalized divergences that are relevant to this work.  Let $\rho\in\s{D}_A$ be a state, and let $\sigma\in\s{PSD}_A$ be a positive semidefinite operator.  For $\alpha\in(1,\infty)$, the \emph{sandwiched Rényi divergence} is defined as~\cite{muller-lennert2013QuantumRenyiEntropies,wilde2014StrongConverseClassical}
\begin{align}
	\label{eq:sandwiched}
	\sw{D}_\alpha\fleft(\rho\middle\|\sigma\fright)&\coloneq\begin{cases}
		\frac{1}{\alpha-1}\ln\left\lVert\sigma^\frac{1-\alpha}{2\alpha}\rho\sigma^\frac{1-\alpha}{2\alpha}\right\rVert_\alpha^\alpha&\text{if }\rho^0\leq\sigma^0, \\
		\infty&\text{otherwise},
	\end{cases}
\end{align}
where
\begin{align}
	\label{eq:norm}
	\left\lVert\gamma\right\rVert_\alpha&\coloneq\left(\tr\fleft[\left(\gamma^\dagger\gamma\right)^\frac{\alpha}{2}\fright]\right)^\frac{1}{\alpha}
\end{align}
denotes the $\alpha$-norm of an operator $\gamma\in\spa{B}_A$ for $\alpha\in[1,\infty)$.  Apart from obeying the DPI, the sandwiched Rényi divergence is additive, jointly quasiconvex, monotonically nonincreasing and regular in its second argument, and monotonically nondecreasing in $\alpha$~\cite{muller-lennert2013QuantumRenyiEntropies,wilde2014StrongConverseClassical,beigi2013SandwichedRenyiDivergence,frank2013MonotonicityRelativeRenyi,mosonyi2015QuantumHypothesisTesting} (see Refs.~\cite{tomamichel2016QuantumInformationProcessing,khatri2024PrinciplesQuantumCommunication} for textbooks on the topic).  We defer the discussion of these properties to Theorem~\ref{thm:sandwiched-extended} of Section~\ref{sec:extension}, where we introduce an extended definition of the sandwiched Rényi divergence whose first argument can be a general Hermitian operator.  As shown in Refs.~\cite{muller-lennert2013QuantumRenyiEntropies,wilde2014StrongConverseClassical}, the limit of the sandwiched Rényi divergence as $\alpha\searrow1$ is given by the \emph{Umegaki divergence} (also known as the \emph{quantum relative entropy})~\cite{umegaki1962ConditionalExpectationOperator}, which is defined as
\begin{align}
	D\fleft(\rho\middle\|\sigma\fright)&\coloneq\begin{cases}
		\tr\fleft[\rho\left(\ln\rho-\ln\sigma\right)\fright]&\text{if }\rho^0\leq\sigma^0, \\
		\infty&\text{otherwise},
	\end{cases} \label{eq:umegaki}\\
	&=\lim_{\alpha\searrow1}\sw{D}_\alpha\fleft(\rho\middle\|\sigma\fright). \label{eq:umegaki-limit}
\end{align}
Another useful generalized divergence is the \emph{geometric Rényi divergence}~\cite{matsumoto2018NewQuantumVersion}, which for $\alpha\in(1,2]$ is defined as
\begin{align}
	\label{eq:geometric}
	\g{D}_\alpha\fleft(\rho\middle\|\sigma\fright)&\coloneq\begin{cases}
		\frac{1}{\alpha-1}\ln\tr\fleft[\sigma\left(\sigma^{-\frac{1}{2}}\rho\sigma^{-\frac{1}{2}}\right)^\alpha\fright]&\text{if }\rho^0\leq\sigma^0, \\
		\infty&\text{otherwise}.
	\end{cases}
\end{align}
Apart from obeying the DPI, the geometric Rényi divergence is additive, jointly quasiconvex, monotonically nonincreasing and regular in its second argument~\cite{petz1998ContractionGeneralizedRelative,matsumoto2018NewQuantumVersion,fang2021GeometricRenyiDivergence,katariya2021GeometricDistinguishabilityMeasures,khatri2024PrinciplesQuantumCommunication}.  Moreover, it is monotonically nondecreasing in $\alpha$~\cite[Proposition~72.2]{katariya2021GeometricDistinguishabilityMeasures}: for all $\alpha,\beta\in(1,2]$ such that $\alpha\leq\beta$,
\begin{align}
	\label{eq:monotonicity-geometric}
	\g{D}_\alpha\fleft(\rho\middle\|\sigma\fright)&\leq\g{D}_{\beta}\fleft(\rho\|\sigma\fright)\quad\forall\rho\in\s{D}_A,\;\sigma\in\s{PSD}_A.
\end{align}
The geometric Rényi divergence provides an upper bound on the sandwiched Rényi divergence~\cite{tomamichel2016QuantumInformationProcessing,matsumoto2018NewQuantumVersion}: for all $\alpha\in(1,2]$,
\begin{align}
	\label{eq:order}
	\sw{D}_\alpha\fleft(\rho\middle\|\sigma\fright)&\leq\g{D}_\alpha\fleft(\rho\middle\|\sigma\fright)\quad\forall\rho\in\s{D}_A,\;\sigma\in\s{PSD}_A.
\end{align}
When defined on classical states, the above inequality becomes an equality, and both sides reduce to the classical Rényi divergence.  As shown in Ref.~\cite[Proposition~79]{katariya2021GeometricDistinguishabilityMeasures}, the limit of the geometric Rényi divergence as $\alpha\searrow1$ is given by the \emph{Belavkin--Staszewski divergence}~\cite{belavkin1982$C^$algebraicGeneralizationRelative}, which is defined as
\begin{align}
	\label{eq:belavkin}
	\g{D}\fleft(\rho\middle\|\sigma\fright)&\coloneq\begin{cases}
		\tr\fleft[\rho\ln\left(\rho^\frac{1}{2}\sigma^{-1}\rho^\frac{1}{2}\right)\fright]&\text{if }\rho^0\leq\sigma^0, \\
		\infty&\text{otherwise},
	\end{cases} \\
	&=\lim_{\alpha\searrow1}\g{D}_\alpha\fleft(\rho\middle\|\sigma\fright).
\end{align}
When defined on classical states, the Belavkin--Staszewski divergence and the Umegaki divergence are equal, both reducing to the Kullback--Leibler divergence (also known as the classical relative entropy).

\subsection{Bivariate divergence measures for channels}
\label{sec:divergence-channel}

Divergence measures can also be used to quantify the dissimilarity of channels, based on their definitions for states.  Let $\ch{N}\in\s{C}_{A\to B}$ be a channel, and let $\ch{M}\in\s{CP}_{A\to B}$ be a completely positive map.  For a (bivariate) generalized divergence $\gd{D}$, the \emph{channel $\gd{D}$-divergence} is defined as~\cite{leditzky2018ApproachesApproximateAdditivity}
\begin{align}
	\label{eq:divergence-channel}
	\gd{D}\fleft(\ch{N}\middle\|\ch{M}\fright)&\coloneq\sup_{\rho\in\s{D}_{RA}}\gd{D}\fleft(\ch{N}_{A\to B}\fleft[\rho_{RA}\fright]\middle\|\ch{M}_{A\to B}\fleft[\rho_{RA}\fright]\fright),
\end{align}
where $R$ is a system with $d_R$ allowed to be arbitrarily large.  As shown in Ref.~\cite{leditzky2018ApproachesApproximateAdditivity}, the supremum in \eqref{eq:divergence-channel} can equivalently be taken over pure states such that $d_R=d_A$.  A channel $\gd{D}$-divergence is said to be \emph{monotonically nonincreasing in its second argument} whenever
\begin{align}
	\label{eq:antimonotonicity-channel}
	&\gd{D}\fleft(\ch{N}\middle\|\ch{M}+\ch{M}'\fright)\leq\gd{D}\fleft(\ch{N}\middle\|\ch{M}\fright) \notag\\
	&\qquad\forall\ch{N}\in\s{C}_{A\to B},\;\ch{M},\ch{M}'\in\s{CP}_{A\to B}.
\end{align}
It is said to be \emph{regular in its second argument} whenever
\begin{align}
	\label{eq:regularity-channel}
	&\gd{D}\fleft(\ch{N}\middle\|\ch{M}\fright)=\lim_{\varepsilon\searrow0}\gd{D}\fleft(\ch{N}\middle\|\ch{M}+\varepsilon\ch{I}\fright) \notag\\
	&\qquad\forall\ch{N}\in\s{C}_{A\to B},\;\ch{M}\in\s{CP}_{A\to B},
\end{align}
where $\ch{I}\in\s{CP}_{A\to B}\colon\rho\mapsto\tr[\rho]\1$ denotes the unnormalized completely depolarizing map.  The following lemma provides some basic properties which, if possessed by a bivariate generalized divergence $\gd{D}$, are inherited by the channel $\gd{D}$-divergence.

\begin{lemma}[Inheritance of properties by the channel $\gd{D}$-divergence]
\label{lem:inheritance}
Let $\gd{D}$ be a bivariate generalized divergence.  If $\gd{D}$ has any of the following properties, then the channel $\gd{D}$-divergence has the same property: \textnormal{(i)} quasiconvexity in either of its arguments or jointly in both arguments; \textnormal{(ii)} lower semicontinuity in either of its arguments or jointly in both arguments; \textnormal{(iii)} nonincreasing monotonicity in its second argument; and \textnormal{(iv)} nonincreasing monotonicity and regularity in its second argument.
\end{lemma}

\begin{proof}
The inheritance of (i) follows from the fact that the supremum of a family of quasiconvex functions is quasiconvex.  The inheritance of (ii) follows from the fact that the supremum of a family of lower semicontinuous functions is lower semicontinuous.  The inheritance of (iii) follows straightforwardly from the definition of the channel $\gd{D}$-divergence (see \eqref{eq:divergence-channel}) and the nonincreasing monotonicity of $\gd{D}$ in its second argument (see \eqref{eq:antimonotonicity}).  The proof of the inheritance of (iv) is provided in Appendix~\ref{app:regularity-channel}.
\end{proof}

In what follows, we review the definitions and properties of several channel divergences that are relevant to this work.  For $\alpha\in(1,2]$, the \emph{geometric Rényi channel divergence} is defined as
\begin{align}
	&\g{D}_\alpha\fleft(\ch{N}\middle\|\ch{M}\fright) \notag\\
	&\coloneq\sup_{\rho\in\s{D}_{RA}}\g{D}_\alpha\fleft(\ch{N}_{A\to B}\fleft[\rho_{RA}\fright]\middle\|\ch{M}_{A\to B}\fleft[\rho_{RA}\fright]\fright) \label{eq:geometric-channel}\\
	&=\begin{cases}
		\frac{1}{\alpha-1}\ln\left\lVert\tr_B\fleft[J_\ch{M}^\frac{1}{2}\left(J_\ch{M}^{-\frac{1}{2}}J_\ch{N}J_\ch{M}^{-\frac{1}{2}}\right)^\alpha J_\ch{M}^\frac{1}{2}\fright]\right\rVert_\infty&\text{if }J_\ch{N}^0\leq J_\ch{M}^0, \\
		\infty&\text{otherwise},
	\end{cases} \label{eq:geometric-channel-choi}
\end{align}
where the closed-form expression in terms of the channels' Choi operators in \eqref{eq:geometric-channel-choi} was established in Refs.~\cite[Theorem~3.2]{fang2021GeometricRenyiDivergence}.  The geometric Rényi channel divergence is also known to possess a semidefinite representation for $(1,2]$, and thus it can be evaluated via an SDP~\cite[Theorem~3.6]{fang2021GeometricRenyiDivergence}.  As a direct consequence of \eqref{eq:monotonicity-geometric} and \eqref{eq:geometric-channel}, the geometric Rényi channel divergence is monotonically nondecreasing in $\alpha$: for all $\alpha,\beta\in(1,2]$ such that $\alpha\leq\beta$,
\begin{align}
	\g{D}_\alpha\fleft(\ch{N}\middle\|\ch{M}\fright)&=\sup_{\rho\in\s{D}_{RA}}\g{D}_\alpha\fleft(\ch{N}_{A\to B}\fleft[\rho_{RA}\fright]\middle\|\ch{M}_{A\to B}\fleft[\rho_{RA}\fright]\fright) \\
	&\leq\sup_{\rho\in\s{D}_{RA}}\g{D}_{\beta}\fleft(\ch{N}_{A\to B}\fleft[\rho_{RA}\fright]\middle\|\ch{M}_{A\to B}\fleft[\rho_{RA}\fright]\fright) \\
	&=\g{D}_{\beta}\fleft(\ch{N}\middle\|\ch{M}\fright)\quad\forall\ch{N}\in\s{C}_{A\to B},\;\ch{M}\in\s{CP}_{A\to B}. \label{eq:monotonicity-geometric-channel}
\end{align}
We highlight the following chain rule for our purposes~\cite[Theorem~3.4]{fang2021GeometricRenyiDivergence}: for all $\alpha\in(1,2]$,
\begin{align}
	\label{eq:chain}
	&\g{D}_\alpha\fleft(\ch{N}_{A\to B}\fleft[\rho_{RA}\fright]\middle\|\ch{M}_{A\to B}\fleft[\sigma_{RA}\fright]\fright) \notag\\
	&\leq\g{D}_\alpha\fleft(\ch{N}\middle\|\ch{M}\fright)+\g{D}_\alpha\fleft(\rho_{RA}\middle\|\sigma_{RA}\fright) \notag\\
	&\qquad\forall\ch{N}\in\s{C}_{A\to B},\;\ch{M}\in\s{CP}_{A\to B},\;\rho\in\s{D}_{RA},\;\sigma\in\s{PSD}_{RA}.
\end{align}
As shown in Ref.~\cite[Lemma~35]{ding2023BoundingForwardClassical}, the limit of the geometric Rényi channel divergence as $\alpha\searrow1$ is given by the \emph{Belavkin--Staszewski channel divergence}:
\begin{align}
	\g{D}\fleft(\ch{N}\middle\|\ch{M}\fright)&\coloneq\sup_{\rho\in\s{D}_{RA}}\g{D}\fleft(\ch{N}_{A\to B}\fleft[\rho_{RA}\fright]\middle\|\ch{M}_{A\to B}\fleft[\rho_{RA}\fright]\fright) \label{eq:belavkin-channel}\\
	&=\begin{cases}
		\left\lVert\tr_B\fleft[J_\ch{N}^\frac{1}{2}\ln\left(J_\ch{N}^\frac{1}{2}J_\ch{M}^{-1}J_\ch{N}^\frac{1}{2}\right)J_\ch{N}^\frac{1}{2}\fright]\right\rVert_\infty&\text{if }J_\ch{N}^0\leq J_\ch{M}^0, \\
		\infty&\text{otherwise},
	\end{cases} \label{eq:belavkin-channel-choi}\\
	&=\lim_{\alpha\searrow1}\g{D}_\alpha\fleft(\ch{N}\middle\|\ch{M}\fright). \label{eq:belavkin-channel-limit}
\end{align}
where the closed-form expression in \eqref{eq:belavkin-channel-choi} was established in Ref.~\cite[Theorem~3.2]{fang2021GeometricRenyiDivergence}.

\subsection{Left divergence radii}
\label{sec:radius}

Based on bivariate divergence measures, one can devise multivariate divergence measures that quantify the dissimilarity of multiple states.  Let $\rho_{[r]}\in\s{PSD}_A^{[r]}$ be a tuple of positive semidefinite operators.  For a bivariate generalized divergence $\gd{D}$, the \emph{left $\gd{D}$-radius} is defined as~\cite{mosonyi2021DivergenceRadiiStrong,mosonyi2024GeometricRelativeEntropies}
\begin{align}
	R^\gd{D}\fleft(\rho_{[r]}\fright)&\coloneq\inf_{\tau\in\s{D}_A}\max_{x\in[r]}\gd{D}\fleft(\tau\middle\|\rho_x\fright) \label{eq:radius-1}\\
	&=\inf_{\tau\in\s{D}_A}\sup_{s_{[r]}\in\s{P}_r}\sum_{x\in[r]}s_x\gd{D}\fleft(\tau\middle\|\rho_x\fright). \label{eq:radius-2}
\end{align}
That is, it evaluates the radius of the smallest ball (as measured by $\gd{D}$) that contains all the operators in $\rho_{[r]}$.  We observe that the left $\gd{D}$-radius obeys the DPI:
\begin{align}
	R^\gd{D}\fleft(\ch{N}\fleft[\rho_{[r]}\fright]\fright)&=\inf_{\tau\in\s{D}_B}\max_{x\in[r]}\gd{D}\fleft(\tau\middle\|\ch{N}\fleft[\rho_x\fright]\fright) \label{eq:DPI-radius-1}\\
	&\leq\inf_{\tau\in\s{D}_A}\max_{x\in[r]}\gd{D}\fleft(\ch{N}\fleft[\tau\fright]\middle\|\ch{N}\fleft[\rho_x\fright]\fright) \\
	&\leq\inf_{\tau\in\s{D}_A}\max_{x\in[r]}\gd{D}\fleft(\tau\middle\|\rho_x\fright) \\
	&=R^\gd{D}\fleft(\rho_{[r]}\fright)\quad\forall\ch{N}\in\s{C}_{A\to B},\;\rho_{[r]}\in\s{PSD}_A^{[r]}, \label{eq:DPI-radius-2}
\end{align}
where $\ch{N}[\rho_{[r]}]\equiv(\ch{N}[\rho_1],\ch{N}[\rho_2],\dots,\ch{N}[\rho_r])$.  Therefore, it provides a legitimate way of quantifying the dissimilarity of multiple states.  The following lemma provides an equivalent representation for the left $\gd{D}$-radius under some basic conditions.

\begin{lemma}[Alternative representation of $R^\gd{D}$]
\label{lem:radius-alternative}
Let $\gd{D}$ be a bivariate generalized divergence such that either \textnormal{(I)} it is convex and lower semicontinuous in its first argument, or \textnormal{(II)} it has the following properties: \textnormal{(i)} quasiconvexity in its first argument; \textnormal{(ii)} lower semicontinuity in its first argument; \textnormal{(iii)} nonincreasing monotonicity in its second argument; \textnormal{(iv)} regularity in its second argument; and \textnormal{(v)} $\gd{D}(\rho\|\sigma)<\infty$ if $\rho^0\leq\sigma^0$.  Then
\begin{align}
	\label{eq:radius-alternative}
	R^\gd{D}\fleft(\rho_{[r]}\fright)&=\sup_{s_{[r]}\in\s{P}_r}\inf_{\tau\in\s{D}_A}\sum_{x\in[r]}s_x\gd{D}\fleft(\tau\middle\|\rho_x\fright)\quad\forall\rho_{[r]}\in\s{PSD}_A^{[r]}.
\end{align}
\end{lemma}

\begin{proof}
The statement for (I) was established in Ref.~\cite[Proposition~A.1]{mosonyi2021DivergenceRadiiStrong}.  The statement for (II) is a special case of Lemma~\ref{lem:radius-channel-alternative}, to be provided below.
\end{proof}

In classical probability theory, the left Kullback--Leibler radius is also known as the \emph{multivariate classical Chernoff divergence}~\cite{mishra2024OptimalErrorExponents}, which for a tuple of classical positive operators $\varrho_{[r]}\in\s{PSD}_Y^{[r]}$ is defined as
\begin{align}
	\label{eq:chernoff-classical}
	C\fleft(\varrho_{[r]}\fright)&\coloneq\sup_{s_{[r]}\in\s{P}_r}-\ln\left(\sum_{y\in[d_Y]}\prod_{x\in[r]}p_{y|x}^{s_x}\right) \\
	&=R^D\fleft(\varrho_{[r]}\fright),
\end{align}
where $\varrho_x\equiv\sum_{y\in[d_Y]}p_{y|x}\op{y}{y}$ for all $x\in[r]$.  Consequently, the left $\gd{D}$-radius reduces classically to the multivariate classical Chernoff divergence whenever $\gd{D}$ reduces classically to the Kullback--Leibler divergence, and in this situation we call the left $\gd{D}$-radius a \emph{barycentric Chernoff divergence}.\footnote{In this situation, the quantity on the right-hand side of \eqref{eq:radius-alternative} without the supremization in $s_{[r]}$ is known as a barycentric Rényi divergence~\cite{mosonyi2024GeometricRelativeEntropies}.  A barycentric Chernoff divergence, as we term it, can be thought of as a ``Chernoffization'' of a barycentric Rényi divergence.}  As an example of a barycentric Chernoff divergence, the left Umegaki radius is also known as the \emph{multivariate log-Euclidean Chernoff divergence}~\cite{mishra2024OptimalErrorExponents}, which is defined as
\begin{align}
	C^\flat\fleft(\rho_{[r]}\fright)&\coloneq\sup_{s_{[r]}\in\s{P}_r}\lim_{\varepsilon\searrow0}-\ln\tr\fleft[\exp\left(\sum_{x\in[r]}s_x\ln\left(\rho_x+\varepsilon\1\right)\right)\fright] \label{eq:euclidean} \\
	&=\sup_{s_{[r]}\in\s{P}_r}-\ln\tr\fleft[\Pi\exp\left(\sum_{x\in[r]}s_x\Pi\left(\ln\rho_x\right)\Pi\right)\fright] \label{eq:euclidean-projection}\\
	&=R^D\fleft(\rho_{[r]}\fright), \label{eq:euclidean-radius}
\end{align}
where $\Pi=\bigwedge_{x\in[r]}\rho_x^0$.  Here \eqref{eq:euclidean-projection} can be inferred from Ref.~\cite[Propositions~5.29 and 5.55]{mosonyi2024GeometricRelativeEntropies}, and \eqref{eq:euclidean-radius} is a direct consequence of Ref.~\cite[Proposition~5.29]{mosonyi2024GeometricRelativeEntropies} and Lemma~\ref{lem:radius-alternative} (also see Ref.~\cite[Eq.~(221)]{mishra2024OptimalErrorExponents}).  The multivariate log-Euclidean Chernoff divergence is a multivariate generalization of the bivariate log-Euclidean Chernoff divergence (see Refs.~\cite{audenaert2015A$z$RenyiRelativeEntropies,mosonyi2017StrongConverseExponent} for the definition of the bivariate log-Euclidean Rényi divergence).  As a direct corollary of Ref.~\cite[Proposition~5.49]{nuradha2025MultivariateFidelities}, the multivariate log-Euclidean Chernoff divergence is weakly additive in the sense that
\begin{align}
	\label{eq:euclidean-additivity}
	C^\flat\fleft(\rho_{[r]}^{\otimes n}\fright)&=nC^\flat\fleft(\rho_{[r]}\fright)\quad\forall\rho_{[r]}\in\s{PSD}_A^{[r]},
\end{align}
where $\rho_{[r]}^{\otimes n}\equiv(\rho_1^{\otimes n},\rho_2^{\otimes n},\dots,\rho_r^{\otimes n})$ denotes the $n$-fold tuple resulting from $\rho_{[r]}$.

Let $\ch{N}_{[r]}\in\s{CP}_{A\to B}^{[r]}$ be a tuple of completely positive maps.  In analogy to \eqref{eq:radius-1}, for a generalized divergence $\gd{D}$, we define the \emph{left channel $\gd{D}$-radius} as
\begin{align}
	R^\gd{D}\fleft(\ch{N}_{[r]}\fright)&\coloneq\inf_{\ch{T}\in\s{C}_{A\to B}}\max_{x\in[r]}\gd{D}\fleft(\ch{T}\middle\|\ch{N}_x\fright) \label{eq:radius-channel-1}\\
	&=\inf_{\ch{T}\in\s{C}_{A\to B}}\sup_{s_{[r]}\in\s{P}_r}\sum_{x\in[r]}s_x\gd{D}\fleft(\ch{T}\middle\|\ch{N}_x\fright). \label{eq:radius-channel-2}
\end{align}
The following lemma provides an equivalent representation for the left channel $\gd{D}$-radius analogous to that in Lemma~\ref{lem:radius-alternative}.

\begin{lemma}[Alternative representation of the channel $R^\gd{D}$]
\label{lem:radius-channel-alternative}
Let $\gd{D}$ be a bivariate generalized divergence with the following properties: \textnormal{(i)} quasiconvexity in its first argument; \textnormal{(ii)} lower semicontinuity in its first argument; \textnormal{(iii)} nonincreasing monotonicity in its second argument; \textnormal{(iv)} regularity in its second argument; and \textnormal{(v)} $\gd{D}(\ch{N}\|\ch{M})<\infty$ if $J_\ch{N}^0\leq J_\ch{M}^0$.  Then
\begin{align}
	&R^\gd{D}\fleft(\ch{N}_{[r]}\fright)=\sup_{s_{[r]}\in\s{P}_r}\inf_{\ch{T}\in\s{C}_{A\to B}}\sum_{x\in[r]}s_x\gd{D}\fleft(\ch{T}\middle\|\ch{N}_x\fright) \notag\\
	&\qquad\forall\ch{N}_{[r]}\in\s{CP}_{A\to B}^{[r]}.
\end{align}
\end{lemma}

\begin{proof}
See Appendix~\ref{app:radius-channel-alternative}.
\end{proof}

We call the left channel $\gd{D}$-radius a \emph{barycentric Chernoff channel divergence} whenever $\gd{D}$ reduces classically to the Kullback--Leibler divergence, or equivalently, whenever the left $\gd{D}$-radius is a barycentric Chernoff divergence.  The following lemma shows that when defined on a classical input system, the left channel $\gd{D}$-radius can be expressed in terms of the left $\gd{D}$-radius maximized over classical input states.

\begin{lemma}[$R^\gd{D}$ of classical-to-quantum channels]
\label{lem:radius-channel-classical}
Let $\ch{N}_{[r]}\in\s{CP}_{Y\to B}^{[r]}$ be a tuple of classical-to-quantum positive maps, where $\ch{N}_x\colon\rho\mapsto\sum_{y\in[d_Y]}\bra{y}\rho\ket{y}\nu_{x,y}$ and $\nu_{x,y}\in\s{PSD}_B$ for all $x\in[r]$ and all $y\in[d_Y]$.  Let $\gd{D}$ be a bivariate generalized divergence that is jointly quasiconvex and lower semicontinuous in its first argument.  Then
\begin{align}
	R^\gd{D}\fleft(\ch{N}_{[r]}\fright)&=\max_{y\in[d_Y]}R^\gd{D}\fleft(\nu_{[r],y}\fright),
\end{align}
where $\nu_{[r],y}\equiv(\nu_{1,y},\nu_{2,y},\dots,\nu_{r,y})$ for all $y\in[d_Y]$.
\end{lemma}

\begin{proof}
See Appendix~\ref{app:radius-channel-classical}.
\end{proof}

\section{Extended divergences and their properties}
\label{sec:extension}

In this section, we generalize the definitions of the sandwiched Rényi divergence and the hypothesis-testing divergence to an extended domain, with their first argument now allowed to be a Hermitian (and not necessarily positive semidefinite) operator.  We also investigate properties of these extended divergence measures.

Let $\gamma\in\s{Herm}_A$ be a Hermitian operator with $\gamma\neq0$, and let $\sigma\in\s{PSD}_A$ be a positive semidefinite operator.  For $\alpha\in(1,\infty)$, the \emph{extended sandwiched Rényi divergence}, proposed in Ref.~\cite{wang2020AlogarithmicNegativity} in the context of entanglement theory (also see Ref.~\cite{kosaki1984ApplicationsComplexInterpolation}), is defined as
\begin{align}
	\label{eq:sandwiched-extended}
	\sw{D}_\alpha\fleft(\gamma\middle\|\sigma\fright)&\coloneq\begin{cases}
		\frac{1}{\alpha-1}\ln\left\lVert\sigma^\frac{1-\alpha}{2\alpha}\gamma\sigma^\frac{1-\alpha}{2\alpha}\right\rVert_\alpha^\alpha&\text{if }\gamma^0\leq\sigma^0, \\
		\infty&\text{otherwise}.
	\end{cases}
\end{align}
Indeed, the extended definition of the sandwiched Rényi divergence has precisely the same formula as the original definition presented in \eqref{eq:sandwiched},\footnote{However, there are alternative expressions (see, e.g., Refs.~\cite[Definition~2]{muller-lennert2013QuantumRenyiEntropies} and \cite[Eq.~(4)]{wilde2014StrongConverseClassical}) for the original sandwiched Rényi divergence that are equivalent to \eqref{eq:sandwiched} on states but \emph{not} well defined (let alone being equivalent to \eqref{eq:sandwiched-extended}) when the first argument is a Hermitian operator.} except for an enlarged domain of its first argument.  This extended definition is justified by multiple desirable properties, as first examined in Ref.~\cite{wang2020AlogarithmicNegativity} and further developed and summarized here.

\begin{theorem}[Properties of the extended $\sw{D}_\alpha$]
\label{thm:sandwiched-extended}
Let $\gamma\in\s{Herm}_A$ be a Hermitian operator with $\gamma\neq0$, and let $\sigma\in\s{PSD}_A$ be a positive semidefinite operator.  For all $\alpha\in(1,\infty)$, the extended sandwiched Rényi divergence has the following properties.
\begin{enumerate}
	\item \emph{Data-processing inequality~\cite[Lemma~2]{wang2020AlogarithmicNegativity}.}  Let $\ch{N}\in\s{C}_{A\to B}$ be a channel.  Then
	\begin{align}
		\sw{D}_\alpha\fleft(\ch{N}\fleft[\gamma\fright]\middle\|\ch{N}\fleft[\sigma\fright]\fright)&\leq\sw{D}_\alpha\fleft(\gamma\middle\|\sigma\fright).
	\end{align}
	\item \emph{Nondecreasing monotonicity in $\alpha$~\cite[Lemma~5]{wang2020AlogarithmicNegativity}.}  For all $\beta\in[\alpha,\infty)$,
	\begin{align}
		\sw{D}_\alpha\fleft(\gamma\middle\|\sigma\fright)-\frac{\alpha}{\alpha-1}\ln\left\lVert\gamma\right\rVert_1&\leq\sw{D}_{\beta}\fleft(\gamma\middle\|\sigma\fright)-\frac{\beta}{\beta-1}\ln\left\lVert\gamma\right\rVert_1.
	\end{align}
	\item \emph{Additivity.}  Let $\gamma'\in\s{Herm}_A$ be a Hermitian operator with $\gamma'\neq0$, and let $\sigma'\in\s{PSD}_B$ be a positive semidefinite operator.  Then
	\begin{align}
		\sw{D}_\alpha\fleft(\gamma\otimes\gamma'\middle\|\sigma\otimes\sigma'\fright)&=\sw{D}_\alpha\fleft(\gamma\middle\|\sigma\fright)+\sw{D}_\alpha\fleft(\gamma'\middle\|\sigma'\fright). \label{eq:genRenyi additivity}
	\end{align}
	\item \emph{Direct-sum property.}  Let $\gamma_{[r]}\in\s{Herm}_A^{[r]}$ be a tuple of Hermitian operators such that $\gamma_x\neq0$ for all $x\in[r]$, and let $\sigma_{[r]}\in\s{PSD}_A^{[r]}$ be a tuple of positive semidefinite operators.  Let $p_{[r]}\in\s{P}_r$ be a probability distribution, and let $q_{[r]}\in[0,\infty)^{[r]}$ be a tuple of nonnegative real numbers.  Then
	\begin{align}
		\sw{Q}_\alpha\fleft(\cq{\gamma}_{XA}\middle\|\cq{\sigma}_{XA}\fright)&=\sum_{x\in[r]}p_x^\alpha q_x^{1-\alpha}\sw{Q}_\alpha\fleft(\gamma_x\middle\|\sigma_x\fright), \label{eq:genRenyi direct sum}
	\end{align}
	where $\sw{Q}_\alpha(\gamma\|\sigma)\equiv\exp((\alpha-1)\sw{D}_\alpha(\gamma\|\sigma))$ denotes the extended sandwiched Rényi quasi-divergence and
	\begin{align}
		\cq{\gamma}_{XA}&\equiv\sum_{x\in[r]}p_x\op{x}{x}_X\otimes\gamma_{x,A}, \\
		\cq{\sigma}_{XA}&\equiv\sum_{x\in[r]}q_x\op{x}{x}_X\otimes\sigma_{x,A}.
	\end{align}
	\item \emph{Joint quasiconvexity.}  Let $\gamma_{[r]}\in\s{Herm}_A^{[r]}$ be a tuple of Hermitian operators such that $\gamma_x\neq0$ for all $x\in[r]$, and let $\sigma_{[r]}\in\s{PSD}_A^{[r]}$ be a tuple of positive semidefinite operators.  Let $p_{[r]}\in\s{P}_r$ be a probability distribution.  Then
	\begin{align}
		\sw{D}_\alpha\fleft(\sum_{x\in[r]}p_x\gamma_x\middle\|\sum_{x\in[r]}p_x\sigma_x\fright)&\leq\max_{x\in[r]}\sw{D}_\alpha\fleft(\gamma_x\middle\|\sigma_x\fright). \label{eq:genRenyi quasiconvexity}
	\end{align}
	\item \emph{Nonincreasing monotonicity in its second argument.}  Let $\sigma'\in\s{PSD}_A$ be a positive semidefinite operator.  Then
	\begin{align}
		\sw{D}_\alpha\fleft(\gamma\middle\|\sigma+\sigma'\fright)&\leq\sw{D}_\alpha\fleft(\gamma\middle\|\sigma\fright).
	\end{align}
	\item \emph{Limit as $\alpha\searrow1$.}
	\begin{align}
		\lim_{\alpha\searrow1}\left(\sw{D}_\alpha\fleft(\gamma\middle\|\sigma\fright)-\frac{\alpha}{\alpha-1}\ln\left\lVert\gamma\right\rVert_1\right)&=D\fleft(\frac{\left\lvert\gamma\right\rvert}{\left\lVert\gamma\right\rVert_1}\middle\|\sigma\fright), \label{eq:genRenyi limit at 1}
	\end{align}
	where $\lvert\gamma\rvert\coloneq(\gamma^\dagger\gamma)^\frac{1}{2}$.  Note that $\frac{\lvert\gamma\rvert}{\lVert\gamma\rVert_1}\in\s{D}_A$.
\end{enumerate}
\end{theorem}

\begin{proof}
See Appendix~\ref{app:sandwiched-extended} for the proof of Theorem~\ref{thm:sandwiched-extended}.3--\ref{thm:sandwiched-extended}.7.
\end{proof}

The limit of the extended sandwiched Rényi divergence as $\alpha\to\infty$ is known as the \emph{extended max-divergence}~\cite{wang2020AlogarithmicNegativity}, which is defined as
\begin{align}
	D_{\max}\fleft(\gamma\middle\|\sigma\fright)&\coloneq\inf_{\lambda\in[0,\infty)}\left\{\ln\lambda\colon-\lambda\sigma\leq\gamma\leq\lambda\sigma\right\} \label{eq:max-extended}\\
	&=\begin{cases}
		\ln\left\lVert\sigma^{-\frac{1}{2}}\gamma\sigma^{-\frac{1}{2}}\right\rVert_\infty&\text{if }\gamma^0\leq\sigma^0, \\
		\infty&\text{otherwise},
	\end{cases} \\
	&=\lim_{\alpha\to\infty}\sw{D}_\alpha\fleft(\gamma\middle\|\sigma\fright). \label{eq:max-extended-limit}
\end{align}
The extended max-divergence reduces to the original max-divergence~\cite{datta2009MinMaxrelativeEntropies} when its first argument is a state.  A systematic study of the properties of the extended max-divergence can be found in Ref.~\cite[Appendix~H]{mishra2024OptimalErrorExponents}.

Let $\tau\in\aff(\s{D}_A)$ be a unit-trace Hermitian operator, and let $\sigma\in\s{PSD}_A$ be a positive semidefinite operator.  For $\varepsilon\in[0,1]$, we define the \emph{extended hypothesis-testing divergence} as
\begin{align}
	\label{eq:hypothesis-extended}
	D_\abb{H}^\varepsilon\fleft(\tau\middle\|\sigma\fright)&\coloneq\sup_{\Lambda\in\s{PSD}_A}\left\{-\ln\tr\fleft[\Lambda\sigma\fright]\colon\Lambda\leq\1,\;\tr\fleft[\Lambda\tau\fright]\geq1-\varepsilon\right\}.
\end{align}
The extended hypothesis-testing divergence reduces to the original hypothesis-testing divergence~\cite{buscemi2010QuantumCapacityChannels,wang2012OneshotClassicalquantumCapacity} when its first argument is a state.  The following lemma is a generalization of a useful bound on the hypothesis-testing divergence in terms of the sandwiched Rényi divergence in Ref.~\cite[Lemma~5]{cooney2016StrongConverseExponents}, now applicable to their extended definitions.

\begin{lemma}[Connection between the extended $D_\abb{H}^\varepsilon$ and $\sw{D}_\alpha$]
\label{lem:hypothesis-extended}
Let $\tau\in\aff(\s{D}_A)$ be a unit-trace Hermitian operator, and let $\sigma\in\s{PSD}_A$ be a positive semidefinite operator.  Then for all $\varepsilon\in[0,1)$ and all $\alpha\in(1,\infty)$,
\begin{align}
	D_\abb{H}^\varepsilon\fleft(\tau\middle\|\sigma\fright)&\leq\sw{D}_\alpha\fleft(\tau\middle\|\sigma\fright)+\frac{\alpha}{\alpha-1}\ln\left(\frac{1}{1-\varepsilon}\right).
\end{align}
\end{lemma}

\begin{proof}
See Appendix~\ref{app:hypothesis-extended}.
\end{proof}

\section{Quantum state exclusion}
\label{sec:exclusion-state}

In this section, we define and analyze the task of quantum state exclusion from an information-theoretic perspective.  Our analysis includes an exact characterization of the one-shot error probability in terms of the extended hypothesis-testing divergence and a converse bound in terms of an extended version of the left sandwiched Rényi radius.  The one-shot analysis then leads to an upper bound on the asymptotic error exponent of state exclusion, precisely given by the multivariate log-Euclidean Chernoff divergence.

\subsection{Setting}
\label{sec:setting-state}

In contrast to a discrimination task, which aims at finding out the true hypothesis, an exclusion task is concerned with ruling out a false hypothesis.  Despite their equivalence in the binary setting (i.e., when there are only two hypotheses), the two tasks are substantially different when more hypotheses are involved.  In quantum state exclusion, the experimenter receives a system $A$ in an unknown state.  The source of the system is represented by an ensemble of states, $\en{E}\equiv(p_{[r]},\rho_{[r]})$ with $r\geq2$, and this indicates that for each $x\in[r]$, there is a prior probability $p_x$ with which the unknown state is $\rho_x\in\s{D}_A$.  Without loss of generality, we henceforth always assume that $p_{[r]}\in\itr(\s{P}_r)$, i.e., that $p_x\in(0,1)$ for all $x\in[r]$.  The experimenter's goal is to submit an index $x'\in[r]$ that \emph{differs} from the actual label of the state they received.\footnote{In comparison, in quantum state discrimination, the experimenter's goal is to submit an index \emph{equal} to the actual label of the state.}  We can readily see that when $r\geq3$, state exclusion is in general easier than state discrimination, in the sense that the experimenter has a higher chance of success in the former task than in the latter when faced with the same ensemble.

The most general strategy of the experimenter for state exclusion is to perform a measurement, represented by a general POVM $\Lambda_{[r]}\in\s{M}_{A,r}$, on the system $A$ and submit the measurement outcome.  An error occurs if and only if the outcome coincides with the actual label of the state.  Consequently, the \emph{(one-shot) error probability} of state exclusion for the ensemble $\en{E}$ is given by
\begin{align}
	\label{eq:error-probability-state}
	P_\abb{err}\fleft(\en{E}\fright)&\coloneq\inf_{\Lambda_{[r]}\in\s{M}_{A,r}}\sum_{x\in[r]}p_x\tr\fleft[\Lambda_x\rho_x\fright],
\end{align}
where the infimization over all POVMs is due to the experimenter's motivation to avoid an error.  We say that the tuple of states $\rho_{[r]}$ is \emph{perfectly antidistinguishable} whenever $P_\abb{err}(\en{E})=0$.\footnote{Note that $P_\abb{err}(\en{E})=0$ if and only if there exists a POVM $\Lambda_{[r]}\in\s{M}_{A,r}$ such that $\tr[\Lambda_x\rho_x]=0$ for all $x\in[r]$.  This implies that perfect antidistinguishability is a property of the tuple of states $\rho_{[r]}$, independent of the prior probability distribution $p_{[r]}$.}

When the experimenter collects $n$ copies of the same unknown state generated by the source, state exclusion is in effect conducted on the $n$-fold ensemble $\en{E}^n\equiv(p_{[r]},\rho_{[r]}^{\otimes n})$, where $\rho_{[r]}^{\otimes n}\equiv(\rho_1^{\otimes n},\rho_2^{\otimes n},\dots,\rho_r^{\otimes n})$.  As $n$ increases, the error probability decays exponentially fast in general~\cite{mishra2024OptimalErrorExponents}, and this motivates us to investigate the \emph{(asymptotic) error exponent} of state exclusion, which for the ensemble $\en{E}$ is defined as
\begin{align}
	\label{eq:error-exponent-state}
	&\liminf_{n\to\infty}-\frac{1}{n}\ln P_\abb{err}\fleft(\en{E}^n\fright) \notag\\
	&=\liminf_{n\to\infty}\sup_{\Lambda_{[r]}^{(n)}\in\s{M}_{A^n,r}}-\frac{1}{n}\ln\left(\sum_{x\in[r]}p_x\tr\fleft[\Lambda_x^{(n)}\rho_x^{\otimes n}\fright]\right),
\end{align}
where $A^n$ denotes the composite system consisting of $n$ copies of $A$.  As shown in Ref.~\cite[Eq.~(30)]{mishra2024OptimalErrorExponents}, the error exponent of state exclusion for the ensemble $\en{E}$ does not depend on the prior probability distribution $p_{[r]}$.  Analytical expressions for the error exponent have been derived in several special cases: when $r=2$, it is simply given by the error exponent of binary state discrimination, namely, the quantum Chernoff divergence~\cite{audenaert2007DiscriminatingStatesQuantum,nussbaum2009ChernoffLowerBound}; when the states of concern are classical, it is given by the multivariate classical Chernoff divergence~\cite{mishra2024OptimalErrorExponents}.  However, a universal formula that applies to general ensembles with $r\geq3$ has been lacking, although multiple relevant bounds have been established in Ref.~\cite{mishra2024OptimalErrorExponents}.  A major contribution of this paper is to derive a universal single-letter upper bound on the error exponent of state exclusion that is tighter than the best previously known efficiently computable upper bound.

\subsection{Characterization of the one-shot error probability}
\label{sec:oneshot-state}

Our upper bound on the asymptotic error exponent is derived based on one-shot analysis.  The following proposition characterizes the error probability of state exclusion in the one-shot regime, in terms of the extended hypothesis-testing divergence.

\begin{proposition}[Characterization of the error probability]
\label{prop:oneshot-hypothesis}
Let $\en{E}\equiv(p_{[r]},\rho_{[r]})$ be an ensemble of states with $p_{[r]}\in\itr(\s{P}_r)$ and $\rho_{[r]}\in\s{D}_A^{[r]}$.  Then
\begin{align}
	\label{eq:oneshot-hypothesis}
	-\ln P_\abb{err}\fleft(\en{E}\fright)&=\inf_{\tau\in\aff\fleft(\s{D}_A\fright)}D_\abb{H}^{1-\frac{1}{r}}\fleft(\pi_X\otimes\tau_A\middle\|\cq{\rho}_{XA}\fright),
\end{align}
where $\pi_X\equiv\frac{1}{r}\sum_{x\in[r]}\op{x}{x}_X$ denotes the uniform state of a system $X$ and
\begin{align}
	\label{eq:cq}
	\cq{\rho}_{XA}&\equiv\sum_{x\in[r]}p_x\op{x}{x}_X\otimes\rho_{x,A}.
\end{align}
\end{proposition}

\begin{proof}
For a POVM $\Lambda_{[r]}\in\s{M}_{A,r}$, let
\begin{align}
	\cq{\Lambda}_{XA}&\equiv\sum_{x\in[r]}\op{x}{x}_X\otimes\Lambda_{x,A}.
\end{align}
We observe that $\cq{\Lambda}\in\s{PSD}_{XA}$ and that
\begin{align}
	\cq{\Lambda}_{XA}&\leq\1_{XA}, \label{pf:oneshot-hypothesis-1}\\
	\tr\fleft[\cq{\Lambda}_{XA}\cq{\rho}_{XA}\fright]&=\sum_{x\in[r]}p_x\tr\fleft[\Lambda_x\rho_x\fright], \label{pf:oneshot-hypothesis-2}\\
	\tr\fleft[\cq{\Lambda}_{XA}\left(\pi_X\otimes\tau_A\right)\fright]&=\frac{1}{r}\sum_{x\in[r]}\tr\fleft[\Lambda_x\tau\fright]=\frac{1}{r}\quad\forall\tau\in\aff\fleft(\s{D}_A\fright). \label{pf:oneshot-hypothesis-3}
\end{align}
Then for every unit-trace Hermitian operator $\tau\in\aff(\s{D}_A)$, we have that
\begin{align}
	P_\abb{err}\fleft(\en{E}\fright)&=\inf_{\Lambda_{[r]}\in\s{M}_{A,r}}\sum_{x\in[r]}p_x\tr\fleft[\Lambda_x\rho_x\fright] \\
	&\geq\inf_{\cq{\Lambda}\in\s{PSD}_{XA}}\left\{\tr\fleft[\cq{\Lambda}_{XA}\cq{\rho}_{XA}\fright]\colon\cq{\Lambda}_{XA}\leq\1_{XA},\vphantom{\tr\fleft[\cq{\Lambda}_{XA}\left(\pi_X\otimes\tau_A\right)\fright]\geq\frac{1}{r}}\right. \notag\\
	&\qquad\left.\tr\fleft[\cq{\Lambda}_{XA}\left(\pi_X\otimes\tau_A\right)\fright]\geq\frac{1}{r}\right\} \label{pf:oneshot-hypothesis-4}\\
	&=\exp\left(-D_\abb{H}^{1-\frac{1}{r}}\fleft(\pi_X\otimes\tau_A\middle\|\cq{\rho}_{XA}\fright)\right). \label{pf:oneshot-hypothesis-5}
\end{align}
Here \eqref{pf:oneshot-hypothesis-4} follows from \eqref{pf:oneshot-hypothesis-1}--\eqref{pf:oneshot-hypothesis-3}; Eq.~\eqref{pf:oneshot-hypothesis-5} follows from the definition of the extended hypothesis-testing divergence (see \eqref{eq:hypothesis-extended}).  This shows that
\begin{align}
	-\ln P_\abb{err}\fleft(\en{E}\fright)&\leq\inf_{\tau\in\aff\fleft(\s{D}_A\fright)}D_\abb{H}^{1-\frac{1}{r}}\fleft(\pi_X\otimes\tau_A\middle\|\cq{\rho}_{XA}\fright). \label{pf:oneshot-hypothesis-6}
\end{align}
To show the opposite inequality, it follows from the dual SDP formulation of $P_\abb{err}(\en{E})$~\cite[Eq.~(III.15)]{yuen1975OptimumTestingMultiple} (also see Ref.~\cite[Eq.~(4)]{bandyopadhyay2014ConclusiveExclusionQuantum}) that
\begin{align}
	&P_\abb{err}\fleft(\en{E}\fright) \notag\\
	&=\sup_{\gamma\in\s{Herm}_A}\left\{\tr\fleft[\gamma\fright]\colon\gamma\leq p_x\rho_x\;\forall x\in[r]\right\} \label{pf:oneshot-hypothesis-7}\\
	&=\sup_{\gamma\in\s{Herm}_A}\left\{\tr\fleft[\gamma\fright]\colon\1_X\otimes\gamma_A\leq\cq{\rho}_{XA}\right\} \\
	&=\sup_{\tau\in\aff\fleft(\s{D}_A\fright)}\sup_{\eta\in[0,\infty)}\left\{\eta\colon\eta\1_X\otimes\tau_A\leq\cq{\rho}_{XA}\right\} \label{pf:oneshot-hypothesis-8}\\
	&=\sup_{\tau\in\aff\fleft(\s{D}_A\fright)}\sup_{\eta\in[0,\infty)}\inf_{\cq{\Lambda}\in\s{PSD}_{XA}}\left(\eta+\tr\fleft[\cq{\Lambda}_{XA}\left(\cq{\rho}_{XA}-\eta\1_X\otimes\tau_A\right)\fright]\right) \label{pf:oneshot-hypothesis-9}\\
	&=\sup_{\tau\in\aff\fleft(\s{D}_A\fright)}\sup_{\eta\in[0,\infty)}\inf_{\cq{\Lambda}\in\s{PSD}_{XA}}\left(\tr\fleft[\cq{\Lambda}_{XA}\cq{\rho}_{XA}\fright]\right. \notag\\
	&\qquad\left.\vphantom{}+\eta\left(1-\tr\fleft[\cq{\Lambda}_{XA}\left(\1_X\otimes\tau_A\right)\fright]\right)\right) \\
	&\leq\sup_{\tau\in\aff\fleft(\s{D}_A\fright)}\inf_{\cq{\Lambda}\in\s{PSD}_A}\sup_{\eta\in[0,\infty)}\left(\tr\fleft[\cq{\Lambda}_{XA}\cq{\rho}_{XA}\fright]\right. \notag\\
	&\qquad\left.\vphantom{}+\eta\left(1-\tr\fleft[\cq{\Lambda}_{XA}\left(\1_X\otimes\tau_A\right)\fright]\right)\right) \label{pf:oneshot-hypothesis-10}\\
	&=\sup_{\tau\in\aff\fleft(\s{D}_A\fright)}\inf_{\cq{\Lambda}\in\s{PSD}_{XA}}\left\{\tr\fleft[\cq{\Lambda}_{XA}\cq{\rho}_{XA}\fright]\colon\tr\fleft[\cq{\Lambda}_{XA}\left(\1_X\otimes\tau_A\right)\fright]\geq1\right\} \label{pf:oneshot-hypothesis-11}\\
	&\leq\sup_{\tau\in\aff\fleft(\s{D}_A\fright)}\inf_{\cq{\Lambda}\in\s{PSD}_{XA}}\left\{\tr\fleft[\cq{\Lambda}_{XA}\cq{\rho}_{XA}\fright]\colon\cq{\Lambda}_{XA}\leq\1_{XA},\vphantom{\tr\fleft[\cq{\Lambda}_{XA}\left(\pi_X\otimes\tau_A\right)\fright]\geq\frac{1}{r}}\right. \notag\\
	&\qquad\left.\tr\fleft[\cq{\Lambda}_{XA}\left(\pi_X\otimes\tau_A\right)\fright]\geq\frac{1}{r}\right\} \label{pf:oneshot-hypothesis-12}\\
	&=\sup_{\tau\in\aff\fleft(\s{D}_A\fright)}\exp\left(-D_\abb{H}^{1-\frac{1}{r}}\fleft(\pi_X\otimes\tau_A\middle\|\cq{\rho}_{XA}\fright)\right). \label{pf:oneshot-hypothesis-13}
\end{align}
Here \eqref{pf:oneshot-hypothesis-7} is the dual SDP formulation of $P_\abb{err}(\en{E})$; Eq.~\eqref{pf:oneshot-hypothesis-8} applies the substitution $\sigma=\eta\tau$ and follows from the fact that $\eta=0$ is a feasible solution to the infimization; Eq.~\eqref{pf:oneshot-hypothesis-9} introduces $\cq{\Lambda}\in\s{PSD}_{XA}$ as a Lagrange multiplier corresponding to the constraint $\eta\1_X\otimes\tau_A\leq\cq{\rho}_{XA}$; Eq.~\eqref{pf:oneshot-hypothesis-10} applies an exchange in order between the supremization and the infimization; Eq.~\eqref{pf:oneshot-hypothesis-11} follows from a reverse Lagrangian reasoning that eliminates the multiplier $\eta$; Eq.~\eqref{pf:oneshot-hypothesis-12} follows from the fact that an additional constraint on $\cq{\Lambda}$ does not decrease the optimal value of the objective function.  Combining \eqref{pf:oneshot-hypothesis-6} and \eqref{pf:oneshot-hypothesis-13} leads to the desired statement.
\end{proof}

\begin{remark}[Necessity of taking the infimum over $\aff(\s{D}_A)$ in Proposition~\ref{prop:oneshot-hypothesis}]
\label{rem:affine}
For an ensemble of classical states, the infimum on the right-hand side of \eqref{eq:oneshot-hypothesis} can equivalently be taken over states.  However, for a general ensemble of quantum states, the infimum has to be taken over general unit-trace Hermitian operators rather than over states for the equality in \eqref{eq:oneshot-hypothesis} to hold, as Example~\ref{ex:affine} provides an ensemble for which taking the infimum over states would result in a violation of the equality.  As such, the general need for the infimum to be taken over unit-trace Hermitian operators represents a sharp distinction between the classical and quantum theories of information.
\end{remark}

\begin{example}
\label{ex:affine}
Consider an ensemble of qubit states $\en{E}\equiv(p_{[7]},\rho_{[7]})$ with $\rho_{[7]}\in\s{D}_A^{[7]}$, where
\begin{align}
	p_x&\coloneq\frac{1}{7}\quad\forall x\in[7], \\
	\rho_x&\coloneq\begin{cases}
		\op{1}{1}&\text{if }x\in[3], \\
		\frac{1}{2}\left(\ket{1}+\ket{2}\right)\left(\bra{1}+\bra{2}\right)&\text{if }x\in[7]\setminus[3].
	\end{cases}
\end{align}
Define the following operator:
\begin{align}
	\cq{\Lambda}_{\star,XA}&\coloneq\sum_{x\in[3]}\op{x}{x}_X\otimes\op{2}{2}_A \notag\\
	&\qquad+\sum_{x\in[7]\setminus[3]}\op{x}{x}_X\otimes\frac{1}{2}\left(\ket{1}-\ket{2}\right)\left(\bra{1}-\bra{2}\right)_A.
\end{align}
We observe that $\cq{\Lambda}_\star\in\s{PSD}_{XA}$ and that
\begin{align}
	\cq{\Lambda}_{\star,XA}&\leq\1_{XA}, \label{pf:affine-1}\\
	\tr\fleft[\cq{\Lambda}_{\star,XA}\cq{\rho}_{XA}\fright]&=0, \label{pf:affine-2}\\
	\tr_X\fleft[\cq{\Lambda}_{\star,XA}\fright]&=3\op{2}{2}_A+2\left(\ket{1}-\ket{2}\right)\left(\bra{1}-\bra{2}\right)_A\geq\1_A, \label{pf:affine-3}
\end{align}
where \eqref{pf:affine-2} follows the notation of \eqref{eq:cq}.  It follows from \eqref{pf:affine-3} that
\begin{align}
	\tr\fleft[\cq{\Lambda}_{\star,XA}\left(\pi_X\otimes\tau_A\right)\fright]&\geq\frac{1}{7}\quad\forall\tau\in\s{D}_A. \label{pf:affine-4}
\end{align}
Then it follows from the definition of the hypothesis-testing divergence (see \eqref{eq:hypothesis-extended}) that
\begin{align}
	&\inf_{\tau\in\s{D}_A}D_\abb{H}^{1-\frac{1}{r}}\fleft(\pi_X\otimes\tau_A\middle\|\cq{\rho}_{XA}\fright) \notag\\
	&=\inf_{\tau\in\s{D}_A}\sup_{\cq{\Lambda}\in\s{PSD}_{XA}}\left\{-\ln\tr\fleft[\cq{\Lambda}_{XA}\cq{\rho}_{XA}\fright]\colon\cq{\Lambda}_{XA}\leq\1_{XA},\vphantom{\tr\fleft[\cq{\Lambda}_{XA}\left(\pi_X\otimes\tau_A\right)\fright]\geq\frac{1}{7}}\right. \notag\\
	&\qquad\left.\tr\fleft[\cq{\Lambda}_{XA}\left(\pi_X\otimes\tau_A\right)\fright]\geq\frac{1}{7}\right\} \\
	&\geq\inf_{\tau\in\s{D}_A}-\ln\tr\fleft[\cq{\Lambda}_{\star,XA}\cq{\rho}_{XA}\fright] \label{pf:affine-5}\\
	&=\infty \label{pf:affine-6}\\
	&>-\ln P_\abb{err}\fleft(\en{E}\fright). \label{pf:affine-7}
\end{align}
Here \eqref{pf:affine-5} follows from \eqref{pf:affine-1} and \eqref{pf:affine-4}; Eq.~\eqref{pf:affine-6} follows from \eqref{pf:affine-2}; Eq.~\eqref{pf:affine-7} follows from the observation that the tuple of states $\rho_{[7]}$ is not perfectly antidistinguishable and thus $P_\abb{err}(\en{E})>0$.  This shows that taking the infimum on the right-hand side of \eqref{eq:oneshot-hypothesis} over states instead of over unit-trace Hermitian operators can result in a violation of the equality in \eqref{eq:oneshot-hypothesis}, as claimed in Remark~\ref{rem:affine}.
\end{example}

\begin{remark}[Connection with state discrimination]
\label{rem:discrimination}
Proposition~\ref{prop:oneshot-hypothesis} can be understood as a counterpart of Ref.~\cite[Theorem~1]{vazquez-vilar2016MultipleQuantumHypothesis} in state exclusion, while Ref.~\cite[Theorem~1]{vazquez-vilar2016MultipleQuantumHypothesis} provides an exact characterization of the one-shot error probability of quantum state discrimination in terms of the hypothesis-testing divergence.  We restate Ref.~\cite[Theorem~1]{vazquez-vilar2016MultipleQuantumHypothesis} in Proposition~\ref{prop:discrimination} of Appendix~\ref{app:discrimination} and present an alternative proof thereof, following an approach similar to the proof of Proposition~\ref{prop:oneshot-hypothesis}.
\end{remark}

\begin{remark}[Connection with $D_{\max}$]
\label{rem:dmax}
It follows from \eqref{pf:oneshot-hypothesis-8} that
\begin{align}
	-\ln P_\abb{err}\fleft(\en{E}\fright)&=\inf_{\tau\in\aff\fleft(\s{D}_A\fright)}D_{\max}'\fleft(\1_X\otimes\tau_A\middle\|\cq{\rho}_{XA}\fright) \label{pf:prior-1}\\
	&=\inf_{\tau\in\aff\fleft(\s{D}_A\fright)}\max_{x\in[r]}D_{\max}'\fleft(\tau\middle\|p_x\rho_x\fright) \label{pf:prior-2}\\
	&\leq\inf_{\tau\in\aff\fleft(\s{D}_A\fright)}\max_{x\in[r]}D_{\max}'\fleft(\tau\middle\|\rho_x\fright)-\ln\min_{x\in[r]}p_x \label{pf:prior-3}\\
	&\leq\inf_{\tau\in\aff\fleft(\s{D}_A\fright)}\max_{x\in[r]}D_{\max}\fleft(\tau\middle\|\rho_x\fright)-\ln\min_{x\in[r]}p_x, \label{pf:prior-4}
\end{align}
where for $\gamma\in\s{Herm}_A$ and $\sigma\in\s{PSD}_A$,
\begin{align}
	D_{\max}'\fleft(\gamma\middle\|\sigma\fright)&\coloneq\inf_{\lambda\in[0,\infty)}\left\{\ln\lambda\colon\gamma\leq\lambda\sigma\right\}
\end{align}
provides a lower bound on the extended max-divergence in \eqref{eq:max-extended}.  We note that \eqref{pf:prior-3} and \eqref{pf:prior-4} provide alternative representations for Ref.~\cite[Eqs.~(185) and (188)]{mishra2024OptimalErrorExponents}, respectively.  In addition, the right-hand side of \eqref{pf:prior-1} can be understood as a conditional entropy induced by the reverse $D_{\max}'$, and the right-hand side of \eqref{pf:prior-2} can be understood as the left $D_{\max}'$-radius of the tuple of subnormalized states $(p_x\rho_x)_{x\in[r]}$.  Analogous expressions for the success probability of state discrimination have been shown before in Ref.~\cite{konig2009OperationalMeaningMin} and Refs.~\cite{mosonyi2009GeneralizedRelativeEntropies,mosonyi2021DivergenceRadiiStrong}, respectively:
\begin{align}
	\ln P_\abb{succ}^\abb{disc}\fleft(\en{E}\fright)&\coloneq\sup_{\Lambda_{[r]}\in\s{M}_{A,r}}\sum_{x\in[r]}p_x\tr\fleft[\Lambda_x\rho_x\fright] \\
	&=\underbrace{\inf_{\tau\in\s{D}_A}D_{\max}\fleft(\cq{\rho}_{XA}\middle\|\1_X\otimes\tau_A\fright)}_{\eqcolon-H_{\min}\fleft(X\middle|A\fright)_\cq{\rho}} \label{pf:prior-5}\\
	&=\inf_{\tau\in\s{D}_A}\max_{x\in[r]}D_{\max}\fleft(p_x\rho_x\middle\|\tau\fright) \\
	&\leq\inf_{\tau\in\s{D}_A}\max_{x\in[r]}D_{\max}\fleft(\rho_x\middle\|\tau\fright)+\ln\max_{x\in[r]}p_x, \label{pf:prior-6}
\end{align}
where $H_{\min}(X|A)_\cq{\rho}$ is the conditional min-entropy.  Note the reverse order of the arguments in \eqref{pf:prior-5}--\eqref{pf:prior-6} compared to \eqref{pf:prior-1}--\eqref{pf:prior-4}.
\end{remark}

Recall that Lemma~\ref{lem:hypothesis-extended} provides an upper bound on the extended hypothesis-testing divergence in terms of the extended sandwiched Rényi divergence.  Applying this connection to Proposition~\ref{prop:oneshot-hypothesis}, we obtain a converse bound on the error probability in terms of an extended version of the left sandwiched Rényi radius.

\begin{proposition}[Converse bound on the error probability]
\label{prop:oneshot-state}
Let $\en{E}\equiv(p_{[r]},\rho_{[r]})$ be an ensemble of states with $p_{[r]}\in\itr(\s{P}_r)$ and $\rho_{[r]}\in\s{D}_A^{[r]}$.  Then for all $\alpha\in(1,\infty)$,
\begin{align}
	\label{eq:oneshot-state}
	&-\ln P_\abb{err}\fleft(\en{E}\fright) \notag\\
	&\leq\sup_{s_{[r]}\in\s{P}_r}\inf_{\tau\in\aff\fleft(\s{D}_A\fright)}\sum_{x\in[r]}s_x\sw{D}_\alpha\fleft(\tau\middle\|\rho_x\fright)+\frac{\alpha}{\alpha-1}\ln\left(\frac{1}{p_{\min}}\right),
\end{align}
where $p_{\min}\equiv\min_{x\in[r]}p_x$.
\end{proposition}

\begin{proof}
For $\delta\in(0,1)$, define an ensemble of states $\en{E}_\delta\coloneq(p_{[r]},\ch{D}_\delta[\rho_{[r]}])$, where 
\begin{align}
	\ch{D}_\delta\in\s{C}_{A\to A}&\colon\rho\mapsto\left(1-\delta\right)\rho+\delta\tr\fleft[\rho\fright]\pi
\end{align}
denotes the depolarizing channel with parameter $\delta$.  Recall that $\sw{Q}_\alpha(\rho\|\sigma)\equiv\exp((\alpha-1)\sw{D}_\alpha(\rho\|\sigma))$ denotes the extended sandwiched Rényi quasi-divergence.  Applying Proposition~\ref{prop:oneshot-hypothesis} to the ensemble $\en{E}_\delta$, for all $\delta\in(0,1)$ and $\alpha\in(1,\infty)$, we have that
\begin{align}
	&-\ln P_\abb{err}\fleft(\en{E}_\delta\fright) \notag\\
	&=\inf_{\tau\in\aff\fleft(\s{D}_A\fright)}D_\abb{H}^{1-\frac{1}{r}}\fleft(\pi_X\otimes\tau_A\middle\|\ch{D}_{\delta,A\to A}\fleft[\cq{\rho}_{XA}\fright]\fright) \\
	&\leq\inf_{\tau\in\aff\fleft(\s{D}_A\fright)}\sw{D}_\alpha\fleft(\pi_X\otimes\tau_A\middle\|\ch{D}_{\delta,A\to A}\fleft[\cq{\rho}_{XA}\fright]\fright)+\frac{\alpha}{\alpha-1}\ln r \label{pf:oneshot-state-1}\\
	&=\inf_{\tau\in\aff\fleft(\s{D}_A\fright)}\frac{1}{\alpha-1}\ln\sw{Q}_\alpha\fleft(\pi_X\otimes\tau_A\middle\|\ch{D}_{\delta,A\to A}\fleft[\cq{\rho}_{XA}\fright]\fright) \notag\\
	&\qquad+\frac{\alpha}{\alpha-1}\ln r \\
	&=\inf_{\tau\in\aff\fleft(\s{D}_A\fright)}\frac{1}{\alpha-1}\ln\left(\sum_{x\in[r]}\frac{p_x^{1-\alpha}}{r^\alpha}\sw{Q}_\alpha\fleft(\tau\middle\|\ch{D}_\delta\fleft[\rho_x\fright]\fright)\right) \notag\\
	&\qquad+\frac{\alpha}{\alpha-1}\ln r \label{pf:oneshot-state-2}\\
	&=\inf_{\tau\in\aff\fleft(\s{D}_A\fright)}\frac{1}{\alpha-1}\ln\left(\sum_{x\in[r]}p_xp_x^{-\alpha}\sw{Q}_\alpha\fleft(\tau\middle\|\ch{D}_\delta\fleft[\rho_x\fright]\fright)\right) \\
	&\leq\inf_{\tau\in\aff\fleft(\s{D}_A\fright)}\max_{x\in[r]}\frac{1}{\alpha-1}\ln\left(p_x^{-\alpha}\sw{Q}_\alpha\fleft(\tau\middle\|\ch{D}_\delta\fleft[\rho_x\fright]\fright)\right) \\
	&=\inf_{\tau\in\aff\fleft(\s{D}_A\fright)}\max_{x\in[r]}\left(\sw{D}_\alpha\fleft(\tau\middle\|\ch{D}_\delta\fleft[\rho_x\fright]\fright)+\frac{\alpha}{\alpha-1}\ln\left(\frac{1}{p_x}\right)\right) \\
	&\leq\inf_{\tau\in\aff\fleft(\s{D}_A\fright)}\max_{x\in[r]}\sw{D}_\alpha\fleft(\tau\middle\|\ch{D}_\delta\fleft[\rho_x\fright]\fright)+\frac{\alpha}{\alpha-1}\ln\left(\frac{1}{p_{\min}}\right) \\
	&=\inf_{\tau\in\aff\fleft(\s{D}_A\fright)}\max_{x\in[r]}\sw{D}_\alpha\fleft(\tau\middle\|\ch{D}_\delta\fleft[\rho_x\fright]\fright)+\frac{\alpha}{\alpha-1}\ln\left(\frac{1}{p_{\min}}\right) \\
	&=\inf_{\tau\in\aff\fleft(\s{D}_A\fright)}\sup_{s_{[r]}\in\s{P}_r}\sum_{x\in[r]}s_x\sw{D}_\alpha\fleft(\tau\middle\|\ch{D}_\delta\fleft[\rho_x\fright]\fright)+\frac{\alpha}{\alpha-1}\ln\left(\frac{1}{p_{\min}}\right) \\
	&=\sup_{s_{[r]}\in\s{P}_r}\inf_{\tau\in\aff\fleft(\s{D}_A\fright)}\sum_{x\in[r]}s_x\sw{D}_\alpha\fleft(\tau\middle\|\ch{D}_\delta\fleft[\rho_x\fright]\fright)+\frac{\alpha}{\alpha-1}\ln\left(\frac{1}{p_{\min}}\right) \label{pf:oneshot-state-3}\\
	&\leq\sup_{s_{[r]}\in\s{P}_r}\inf_{\tau\in\aff\fleft(\s{D}_A\fright)}\sum_{x\in[r]}s_x\sw{D}_\alpha\fleft(\tau\middle\|\left(1-\delta\right)\rho_x\fright) \notag\\
	&\qquad+\frac{\alpha}{\alpha-1}\ln\left(\frac{1}{p_{\min}}\right) \label{pf:oneshot-state-4}\\
	&=\sup_{s_{[r]}\in\s{P}_r}\inf_{\tau\in\aff\fleft(\s{D}_A\fright)}\sum_{x\in[r]}s_x\sw{D}_\alpha\fleft(\tau\middle\|\rho_x\fright)+\ln\left(\frac{1}{1-\delta}\right) \notag\\
	&\qquad+\frac{\alpha}{\alpha-1}\ln\left(\frac{1}{p_{\min}}\right). \label{pf:oneshot-state-5}
\end{align}
Here \eqref{pf:oneshot-state-1} follows from Lemma~\ref{lem:hypothesis-extended}; Eq.~\eqref{pf:oneshot-state-2} uses the direct-sum property of the extended sandwiched Rényi quasi-divergence (see Theorem~\ref{thm:sandwiched-extended}.4); Eq.~\eqref{pf:oneshot-state-3} applies the Sion minimax theorem~\cite{sion1958GeneralMinimaxTheorems}; Eq.~\eqref{pf:oneshot-state-4} uses the nonincreasing monotonicity of the extended sandwiched Rényi divergence in its second argument (see Theorem~\ref{thm:sandwiched-extended}.6); Eq.~\eqref{pf:oneshot-state-5} follows from the definition of the extended sandwiched Rényi divergence (see \eqref{eq:sandwiched-extended}).  The application of the Sion minimax theorem in \eqref{pf:oneshot-state-3} is justified by the following observations: the feasible regions of $s_{[r]}$ and $\tau$ are both convex with the former being compact; the objective function is linear in $s_{[r]}$, quasiconvex in $\tau$ due to the joint quasiconvexity of the extended sandwiched Rényi divergence (see Theorem~\ref{thm:sandwiched-extended}.5), and lower semicontinuous in $\tau$ due to the continuity of the $\alpha$-norm (see \eqref{eq:sandwiched-extended} and \eqref{eq:norm}).  We thus conclude that the following inequality holds for all $\delta\in(0,1)$ and $\alpha\in(1,\infty)$:
\begin{align}
	-\ln P_\abb{err}\fleft(\en{E}_\delta\fright)&\leq\sup_{s_{[r]}\in\s{P}_r}\inf_{\tau\in\aff\fleft(\s{D}_A\fright)}\sum_{x\in[r]}s_x\sw{D}_\alpha\fleft(\tau\middle\|\rho_x\fright)+\ln\left(\frac{1}{1-\delta}\right) \notag\\
	&\qquad+\frac{\alpha}{\alpha-1}\ln\left(\frac{1}{p_{\min}}\right). \label{pf:oneshot-state-6}
\end{align}
Note that for every POVM $\Lambda_{[r]}\in\s{M}_{A,r}$ and all $\delta\in(0,1)$, we have that
\begin{align}
	&\left\lvert\sum_{x\in[r]}p_x\tr\fleft[\Lambda_x\rho_x\fright]-\sum_{x\in[r]}p_x\tr\fleft[\Lambda_x\ch{D}_\delta\fleft[\rho_x\fright]\fright]\right\rvert \notag\\
	&=\delta\left\lvert\sum_{x\in[r]}p_x\tr\fleft[\Lambda_x\left(\rho_x-\pi\right)\fright]\right\rvert \\
	&\leq\delta\sum_{x\in[r]}p_x\left\lvert\tr\fleft[\Lambda_x\left(\rho_x-\pi\right)\fright]\right\rvert \label{pf:oneshot-state-7}\\
	&\leq\delta\sum_{x\in[r]}p_x\left\lVert\Lambda_x\right\rVert_\infty\left\lVert\rho_x-\pi\right\rVert_1 \label{pf:oneshot-state-8}\\
	&\leq2\delta.
\end{align}
Here \eqref{pf:oneshot-state-7} applies Hölder's inequality; Eq.~\eqref{pf:oneshot-state-8} follows from the facts that $\lVert\Lambda_x\rVert_\infty\leq1$ and $\frac{1}{2}\lVert\rho_x-\pi\rVert_1\leq1$ for all $x\in[r]$.  Then it follows from the definition of the error probability (see \eqref{eq:error-probability-state}) that
\begin{align}
	&\left\lvert P_\abb{err}\fleft(\en{E}_\delta\fright)-P_\abb{err}\fleft(\en{E}\fright)\right\rvert \notag\\
	&\leq\sup_{\Lambda_{[r]}\in\s{M}_{A,r}}\left\lvert\sum_{x\in[r]}p_x\tr\fleft[\Lambda_x\rho_x\fright]-\sum_{x\in[r]}p_x\tr\fleft[\Lambda_x\ch{D}_\delta\fleft[\rho_x\fright]\fright]\right\rvert \\
	&\leq2\delta,
\end{align}
and thus
\begin{align}
	\lim_{\delta\searrow0}P_\abb{err}\fleft(\en{E}_\delta\fright)&=P_\abb{err}\fleft(\en{E}\fright).
\end{align}
Taking the limit as $\delta\searrow0$ on both sides of \eqref{pf:oneshot-state-6} leads to the desired statement.
\end{proof}

\begin{remark}[Extended left divergence radii]
\label{rem:radius-extended}
The first term on the right-hand side of \eqref{eq:oneshot-state} can be understood as an extended version of the left sandwiched Rényi radius in the representation of \eqref{eq:radius-alternative}.  Specifically, the original infimum over states on the right-hand side of \eqref{eq:radius-alternative} is replaced with an infimum over unit-trace Hermitian operators in \eqref{eq:oneshot-state}.  More generally, for a bivariate generalized divergence $\gd{D}$ with a well-defined extension that allows its first argument to be a Hermitian operator, the \emph{extended left $\gd{D}$-radius} of $\rho_{[r]}\in\s{PSD}_A^{[r]}$ can be accordingly defined as
\begin{align}
	R^{\gd{D},\abb{ext}}\fleft(\rho_{[r]}\fright)&\coloneq\inf_{\tau\in\aff\fleft(\s{D}_A\fright)}\max_{x\in[r]}\gd{D}\fleft(\tau\middle\|\rho_x\fright).
\end{align}
Following the same reasoning as \eqref{eq:DPI-radius-1}--\eqref{eq:DPI-radius-2}, the extended left $\gd{D}$-radius can also be shown to obey the DPI as long as the extension of $\gd{D}$ does.  It remains an open question whether the right-hand side of \eqref{eq:oneshot-state} provides a strictly tighter converse bound compared to replacing the extended left sandwiched Rényi radius with the nonextended one.  More generally, we leave for future work to investigate the quantitative difference between the extended and nonextended left $\gd{D}$-radii.
\end{remark}

\begin{remark}[A closed-form converse bound on the error probability]
\label{rem:closed}
In Proposition~\ref{prop:closed} of Appendix~\ref{app:closed}, we present a closed-form converse bound on the error probability different from Proposition~\ref{prop:oneshot-state}.  Specifically, we employ a relaxation of Proposition~\ref{prop:oneshot-hypothesis} in terms of the Petz--Rényi divergence~\cite{petz1985QuasientropiesStatesNeumann,petz1986QuasientropiesFiniteQuantum} and show that the resulting converse bound has a closed-form expression.  However, this converse bound is not as useful as the one in Proposition~\ref{prop:oneshot-state} in terms of deriving upper bounds on the asymptotic error exponent due to observed subtleties regarding exchange of limits.
\end{remark}

\subsection{Upper bound on the asymptotic error exponent}
\label{sec:asymptotic-state}

An upper bound on the asymptotic error exponent arises as a natural consequence of the converse bound on the one-shot error probability provided in Proposition~\ref{prop:oneshot-state}.  In what follows, we show that this ensuing upper bound can be simplified and is precisely given by the multivariate log-Euclidean Chernoff divergence.

\begin{theorem}[Log-Euclidean upper bound on the error exponent]
\label{thm:asymptotic-state}
Let $\en{E}\equiv(p_{[r]},\rho_{[r]})$ be an ensemble of states with $p_{[r]}\in\itr(\s{P}_r)$ and $\rho_{[r]}\in\s{D}_A^{[r]}$.  Then
\begin{align}
	\label{eq:asymptotic-state}
	&\limsup_{n\to\infty}-\frac{1}{n}\ln P_\abb{err}\fleft(\en{E}^n\fright) \notag\\
	&\leq\inf_{\alpha\in(1,\infty)}\limsup_{n\to\infty}\sup_{s_{[r]}\in\s{P}_r}\inf_{\tau_n\in\aff\fleft(\s{D}_{A^n}\fright)}\frac{1}{n}\sum_{x\in[r]}s_x\sw{D}_\alpha\fleft(\tau_n\middle\|\rho_x^{\otimes n}\fright).
\end{align}
Furthermore, the right-hand side of \eqref{eq:asymptotic-state} is bounded from above by the multivariate log-Euclidean Chernoff divergence of the tuple of states $\rho_{[r]}$, as defined in \eqref{eq:euclidean}.  Consequently,
\begin{align}
	\limsup_{n\to\infty}-\frac{1}{n}\ln P_\abb{err}\fleft(\en{E}^n\fright)&\leq C^\flat\fleft(\rho_{[r]}\fright) \label{eq:asymptotic-euclidean-1}\\
	&=\sup_{s_{[r]}\in\s{P}_r}\inf_{\tau\in\s{D}_A}\sum_{x\in[r]}s_xD\fleft(\tau\middle\|\rho_x\fright). \label{eq:asymptotic-euclidean-2}
\end{align}
\end{theorem}

\begin{proof}
Applying Proposition~\ref{prop:oneshot-state} to the $n$-fold ensemble $\en{E}^n$, for every positive integer $n$ and all $\alpha\in(1,\infty)$, we have that
\begin{align}
	&-\frac{1}{n}\ln P_\abb{err}\fleft(\en{E}^n\fright) \notag\\
	&\leq\sup_{s_{[r]}\in\s{P}_r}\inf_{\tau_n\in\aff\fleft(\s{D}_{A^n}\fright)}\frac{1}{n}\sum_{x\in[r]}s_x\sw{D}_\alpha\fleft(\tau_n\middle\|\rho_x^{\otimes n}\fright) \notag\\
	&\qquad+\frac{\alpha}{n\left(\alpha-1\right)}\ln\left(\frac{1}{p_{\min}}\right) \label{pf:asymptotic-euclidean-1}\\
	&\leq\sup_{s_{[r]}\in\s{P}_r}\inf_{\tau\in\aff\fleft(\s{D}_A\fright)}\frac{1}{n}\sum_{x\in[r]}s_x\sw{D}_\alpha\fleft(\tau^{\otimes n}\middle\|\rho_x^{\otimes n}\fright) \notag\\
	&\qquad+\frac{\alpha}{n\left(\alpha-1\right)}\ln\left(\frac{1}{p_{\min}}\right) \label{pf:asymptotic-euclidean-2}\\
	&=\sup_{s_{[r]}\in\s{P}_r}\inf_{\tau\in\aff\fleft(\s{D}_A\fright)}\sum_{x\in[r]}s_x\sw{D}_\alpha\fleft(\tau\middle\|\rho_x\fright)+\frac{\alpha}{n\left(\alpha-1\right)}\ln\left(\frac{1}{p_{\min}}\right). \label{pf:asymptotic-euclidean-3} 
\end{align}
Here \eqref{pf:asymptotic-euclidean-2} follows from the fact that $\tau^{\otimes n}\in\aff(\s{D}_{A^n})$ for all $\tau\in\aff(\s{D}_A)$; Eq.~\eqref{pf:asymptotic-euclidean-3} uses the additivity of the extended sandwiched Rényi divergence (see Theorem~\ref{thm:sandwiched-extended}.3).  Taking the limit superior as $n\to\infty$ and the infimum over $\alpha\in(1,\infty)$ on both sides of \eqref{pf:asymptotic-euclidean-1}, we obtain \eqref{eq:asymptotic-state}.  To show that the right-hand side of \eqref{eq:asymptotic-state} is bounded from above by the multivariate log-Euclidean Chernoff divergence of $\rho_{[r]}$ and consequently \eqref{eq:asymptotic-euclidean-1}, it follows from \eqref{eq:asymptotic-state} that
\begin{align}
	&\limsup_{n\to\infty}-\frac{1}{n}\ln P_\abb{err}\fleft(\en{E}^n\fright) \notag\\
	&\leq\inf_{\alpha\in(1,\infty)}\limsup_{n\to\infty}\sup_{s_{[r]}\in\s{P}_r}\inf_{\tau_n\in\aff\fleft(\s{D}_{A^n}\fright)}\frac{1}{n}\sum_{x\in[r]}s_x\sw{D}_\alpha\fleft(\tau_n\middle\|\rho_x^{\otimes n}\fright) \\
	&\leq\inf_{\alpha\in(1,\infty)}\sup_{s_{[r]}\in\s{P}_r}\inf_{\tau\in\aff\fleft(\s{D}_A\fright)}\sum_{x\in[r]}s_x\sw{D}_\alpha\fleft(\tau\middle\|\rho_x\fright) \label{pf:asymptotic-euclidean-4}\\
	&\leq\inf_{\alpha\in(1,\infty)}\sup_{s_{[r]}\in\s{P}_r}\inf_{\tau\in\s{D}_A}\sum_{x\in[r]}s_x\sw{D}_\alpha\fleft(\tau\middle\|\rho_x\fright) \\
	&=\inf_{\alpha\in(1,\infty)}\inf_{\tau\in\s{D}_A}\sup_{s_{[r]}\in\s{P}_r}\sum_{x\in[r]}s_x\sw{D}_\alpha\fleft(\tau\middle\|\rho_x\fright) \label{pf:asymptotic-euclidean-5}\\
	&=\inf_{\tau\in\s{D}_A'}\inf_{\alpha\in(1,\infty)}\sup_{s_{[r]}\in\s{P}_r}\sum_{x\in[r]}s_x\sw{D}_\alpha\fleft(\tau\middle\|\rho_x\fright) \label{pf:asymptotic-euclidean-6}\\
	&=\inf_{\tau\in\s{D}_A'}\sup_{s_{[r]}\in\s{P}_r}\inf_{\alpha\in(1,\infty)}\sum_{x\in[r]}s_x\sw{D}_\alpha\fleft(\tau\middle\|\rho_x\fright) \label{pf:asymptotic-euclidean-7}\\
	&=\inf_{\tau\in\s{D}_A'}\sup_{s_{[r]}\in\s{P}_r}\sum_{x\in[r]}s_xD\fleft(\tau\middle\|\rho_x\fright) \label{pf:asymptotic-euclidean-8}\\
	&=\inf_{\tau\in\s{D}_A}\sup_{s_{[r]}\in\s{P}_r}\sum_{x\in[r]}s_xD\fleft(\tau\middle\|\rho_x\fright) \\
	&=C^\flat\fleft(\rho_{[r]}\fright), \label{pf:asymptotic-euclidean-9}
\end{align}
where in \eqref{pf:asymptotic-euclidean-6}--\eqref{pf:asymptotic-euclidean-8} we denote $\s{D}_A'\equiv\{\tau\in\s{D}_A\colon\tau^0\leq\bigwedge_{x\in[r]}\rho_x^0\}$.  Here \eqref{pf:asymptotic-euclidean-4} follows from the inequality between the right-hand side of \eqref{pf:asymptotic-euclidean-1} and that of \eqref{pf:asymptotic-euclidean-3}; Eq.~\eqref{pf:asymptotic-euclidean-5} follows from applying Lemma~\ref{lem:radius-alternative} to the left sandwiched Rényi radius and \eqref{eq:radius-2}; Eq.~\eqref{pf:asymptotic-euclidean-7} applies the Mosonyi--Hiai minimax theorem~\cite[Corollary~A.2]{mosonyi2011QuantumRenyiRelative}; Eq.~\eqref{pf:asymptotic-euclidean-8} uses the nondecreasing monotonicity of the sandwiched Rényi divergence in $\alpha$ and its limit as $\alpha\searrow1$, which is given by the Umegaki divergence (see \eqref{eq:umegaki-limit}); Eq.~\eqref{pf:asymptotic-euclidean-9} follows from the equality between the left Umegaki radius and the multivariate log-Euclidean Chernoff divergence (see \eqref{eq:radius-2} and \eqref{eq:euclidean-radius}).  The application of the Mosonyi--Hiai minimax theorem in \eqref{pf:asymptotic-euclidean-7} is justified by the following observations: the feasible region of $s_{[r]}$ is compact; the objective function is monotonically nondecreasing in $\alpha$ due to the nondecreasing monotonicity of the sandwiched Rényi divergence in $\alpha$, and it is linear in $s_{[r]}$.  Thus we obtain \eqref{eq:asymptotic-euclidean-1}.  Following this and the equality between the left Umegaki radius and the multivariate log-Euclidean Chernoff divergence (see \eqref{eq:euclidean-radius}), applying Lemma~\ref{lem:radius-alternative} to the left Umegaki radius leads to \eqref{eq:asymptotic-euclidean-2}.
\end{proof}

\begin{remark}[Log-Euclidean converse bound on the nonasymptotic error probability]
\label{rem:nonasymptotic-euclidean}
We note that \eqref{pf:asymptotic-euclidean-3} provides a single-letter converse bound on the nonasymptotic error probability.  In Proposition~\ref{prop:nonasymptotic-euclidean} of Appendix~\ref{app:nonasymptotic-state}, we present a relaxation of \eqref{pf:asymptotic-euclidean-3} in terms of the multivariate log-Euclidean Chernoff divergence with a second-order correction term, thus establishing an efficiently computable converse bound on the nonasymptotic error probability.  Due to \eqref{eq:euclidean-projection}, the multivariate log-Euclidean Chernoff divergence is efficiently computable via convex programming~\cite{boyd2004ConvexOptimization}.
\end{remark}

While Theorem~\ref{thm:asymptotic-state} already establishes the desired log-Euclidean upper bound on the error exponent of state exclusion, it remains unclear whether the upper bound provided on the right-hand side of \eqref{eq:asymptotic-state} is tighter than the log-Euclidean one due to its use of the extended sandwiched Rényi divergence and regularization.  The following proposition addresses this question, showing that the two bounds are in fact equal; that is, no improvement can be made from the use of the extended sandwiched Rényi divergence or regularization in the asymptotic regime.  To prove this, we apply the properties established in Theorem~\ref{thm:sandwiched-extended}.

\begin{proposition}[Equality between the right-hand side of \eqref{eq:asymptotic-state} and $C^\flat$]
\label{prop:equality}
Let $\rho_{[r]}\in\s{D}_A^{[r]}$ be a tuple of states.  Then \eqref{eq:asymptotic-state} and \eqref{eq:asymptotic-euclidean-1} are equal; i.e.,
\begin{align}
	\label{eq:equality}
	&\inf_{\alpha\in(1,\infty)}\limsup_{n\to\infty}\sup_{s_{[r]}\in\s{P}_r}\inf_{\tau_n\in\aff\fleft(\s{D}_{A^n}\fright)}\frac{1}{n}\sum_{x\in[r]}s_x\sw{D}_\alpha\fleft(\tau_n\middle\|\rho_x^{\otimes n}\fright) \notag\\
	&=C^\flat\fleft(\rho_{[r]}\fright).
\end{align}
\end{proposition}

\begin{proof}
It is implied through \eqref{pf:asymptotic-euclidean-1}--\eqref{pf:asymptotic-euclidean-8} that
\begin{align}
	&\inf_{\alpha\in(1,\infty)}\limsup_{n\to\infty}\sup_{s_{[r]}\in\s{P}_r}\inf_{\tau_n\in\aff\fleft(\s{D}_{A^n}\fright)}\frac{1}{n}\sum_{x\in[r]}s_x\sw{D}_\alpha\fleft(\tau_n\middle\|\rho_x^{\otimes n}\fright) \notag\\
	&\leq C^\flat\fleft(\rho_{[r]}\fright). \label{pf:equality-1}
\end{align}
To show the opposite inequality, for every positive integer $n$ and all $\alpha\in(1,\infty)$, we have that
\begin{align}
	&\sup_{s_{[r]}\in\s{P}_r}\inf_{\tau_n\in\aff\fleft(\s{D}_{A^n}\fright)}\frac{1}{n}\sum_{x\in[r]}s_x\sw{D}_\alpha\fleft(\tau_n\middle\|\rho_x^{\otimes n}\fright) \notag\\
	&\geq\sup_{s_{[r]}\in\s{P}_r}\inf_{\tau_n\in\aff\fleft(\s{D}_{A^n}\fright)}\frac{1}{n}\sum_{x\in[r]}s_x\left(\sw{D}_\alpha\fleft(\tau_n\middle\|\rho_x^{\otimes n}\fright)\right. \notag\\
	&\qquad\left.\vphantom{}-\frac{\alpha}{\alpha-1}\ln\left\lVert\tau_n\right\rVert_1\right) \label{pf:equality-2}\\
	&\geq\sup_{s_{[r]}\in\s{P}_r}\inf_{\tau_n\in\aff\fleft(\s{D}_{A^n}\fright)}\lim_{\alpha\searrow1}\frac{1}{n}\sum_{x\in[r]}s_x\left(\sw{D}_\alpha\fleft(\tau_n\middle\|\rho_x^{\otimes n}\fright)\right. \notag\\
	&\qquad\left.\vphantom{}-\frac{\alpha}{\alpha-1}\ln\left\lVert\tau_n\right\rVert_1\right) \label{pf:equality-3}\\
	&=\sup_{s_{[r]}\in\s{P}_r}\inf_{\tau_n\in\aff\fleft(\s{D}_{A^n}\fright)}\frac{1}{n}\sum_{x\in[r]}s_xD\fleft(\frac{\left\lvert\tau_n\right\rvert}{\left\lVert\tau_n\right\rVert_1}\middle\|\rho_x^{\otimes n}\fright) \label{pf:equality-4}\\
	&=\sup_{s_{[r]}\in\s{P}_r}\inf_{\tau_n'\in\s{D}_{A^n}}\frac{1}{n}\sum_{x\in[r]}s_xD\fleft(\tau_n'\middle\|\rho_x^{\otimes n}\fright) \label{pf:equality-5}\\
	&=\frac{1}{n}C^\flat\fleft(\rho_{[r]}^{\otimes n}\fright) \label{pf:equality-6}\\
	&=C^\flat\fleft(\rho_{[r]}\fright). \label{pf:equality-7}
\end{align}
Here \eqref{pf:equality-2} follows from the facts that $\frac{\alpha}{\alpha-1}>0$ for all $\alpha\in(1,\infty)$ and that $\lVert\tau_n\rVert_1=\tr[|\tau_n|]\geq\tr[\tau_n]=1$ for all $\tau_n\in\aff(\s{D}_{A^n})$; Eq.~\eqref{pf:equality-3} uses the nondecreasing monotonicity of the extended sandwiched Rényi divergence in $\alpha$ (see Theorem~\ref{thm:sandwiched-extended}.2); Eq.~\eqref{pf:equality-4} uses the limit of the extended sandwiched Rényi divergence as $\alpha\searrow1$ (see Theorem~\ref{thm:sandwiched-extended}.7); Eq.~\eqref{pf:equality-5} follows from the fact that $\frac{\lvert\tau_n\rvert}{\lVert\tau_n\rVert_1}\in\s{D}_{A^n}$ for all $\tau_n\in\aff(\s{D}_{A^n})$; Eq.~\eqref{pf:equality-6} follows from the equality between the left Umegaki radius and the multivariate log-Euclidean Chernoff divergence (see Lemma~\ref{lem:radius-alternative} and \eqref{eq:euclidean-radius}); Eq.~\eqref{pf:equality-7} uses the weak additivity of the multivariate log-Euclidean Chernoff divergence (see \eqref{eq:euclidean-additivity}).  Taking the limit superior as $n\to\infty$ and the infimum over $\alpha\in(1,\infty)$ on the left-hand side of \eqref{pf:equality-2} and on the right-hand side of \eqref{pf:equality-7} leads to the desired inequality, which combined with \eqref{pf:equality-1} leads to the desired statement.
\end{proof}

The establishment of Proposition~\ref{prop:equality} also enables a transparent comparison between our log-Euclidean upper bound and an SDP upper bound previously derived in Ref.~\cite[Theorem~19]{mishra2024OptimalErrorExponents}, which can be equivalently formulated as the extended left max-radius~\cite[Theorem~20]{mishra2024OptimalErrorExponents} (see Remark~\ref{rem:radius-extended}).  As shown below, the log-Euclidean upper bound is an improvement upon the SDP upper bound.

\begin{corollary}[Comparison between Theorem~\ref{thm:asymptotic-state} and Ref.~{\cite[Theorems~19 and 20]{mishra2024OptimalErrorExponents}}]
\label{cor:comparison}
Let $\rho_{[r]}\in\s{D}_A^{[r]}$ be a tuple of states.  Then
\begin{align}
	C^\flat\fleft(\rho_{[r]}\fright)&\leq-\ln\kappa\fleft(\rho_{[r]}\fright) \label{eq:comparison-1}\\
	&=\sup_{s_{[r]}\in\s{P}_r}\inf_{\tau\in\aff(\s{D}_A)}\sum_{x\in[r]}s_xD_{\max}\fleft(\tau\middle\|\rho_x\fright), \label{eq:comparison-2}
\end{align}
where
\begin{align}
	\label{eq:kappa}
	\kappa\fleft(\rho_{[r]}\fright)&\coloneq\sup_{\gamma\in\s{Herm}_A}\left\{\tr\fleft[\gamma\fright]\colon-\rho_x\leq\gamma\leq\rho_x\;\forall x\in[r]\right\},
\end{align}
as defined in Ref.~\cite[Eq.~(189)]{mishra2024OptimalErrorExponents}.
\end{corollary}

\begin{proof}
Although Ref.~\cite[Theorem~20]{mishra2024OptimalErrorExponents} only claimed that the right-hand side of \eqref{eq:comparison-2} is an upper bound on the right-hand side of \eqref{eq:comparison-1}, we observe that the two quantities are in fact equal.  Specifically, if $\bigwedge_{x\in[r]}\rho_x^0\neq0$, then the inequality in Ref.~\cite[Eq.~(211)]{mishra2024OptimalErrorExponents} becomes an equality and so does the one in Ref.~\cite[Eq.~(204)]{mishra2024OptimalErrorExponents}; otherwise, both sides of Ref.~\cite[Eq.~(204)]{mishra2024OptimalErrorExponents} are equal to $0$.  Therefore, to prove Corollary~\ref{cor:comparison}, it remains to compare the multivariate log-Euclidean Chernoff divergence and the extended left max-radius on the right-hand side of \eqref{eq:comparison-2}.  Consider that
\begin{align}
	&C^\flat\fleft(\rho_{[r]}\fright) \notag\\
	&=\sup_{s_{[r]}\in\s{P}_r}\inf_{\tau\in\s{D}_A}\sum_{x\in[r]}s_xD\fleft(\tau\middle\|\rho_x\fright) \label{pf:comparison-1}\\
	&=\sup_{s_{[r]}\in\s{P}_r}\inf_{\tau\in\aff\fleft(\s{D}_A\fright)}\sum_{x\in[r]}s_xD\fleft(\frac{\left\lvert\tau\right\rvert}{\left\lVert\tau\right\rVert_1}\middle\|\rho_x\fright) \label{pf:comparison-2}\\
	&=\sup_{s_{[r]}\in\s{P}_r}\inf_{\tau\in\aff\fleft(\s{D}_A\fright)}\lim_{\alpha\searrow1}\sum_{x\in[r]}s_x\left(\sw{D}_\alpha\fleft(\tau\middle\|\rho_x\fright)-\frac{\alpha}{\alpha-1}\ln\left\lVert\tau\right\rVert_1\right) \label{pf:comparison-3}\\
	&\leq\sup_{s_{[r]}\in\s{P}_r}\inf_{\tau\in\aff\fleft(\s{D}_A\fright)}\lim_{\alpha\to\infty}\sum_{x\in[r]}s_x\left(\sw{D}_\alpha\fleft(\tau\middle\|\rho_x\fright)-\frac{\alpha}{\alpha-1}\ln\left\lVert\tau\right\rVert_1\right) \label{pf:comparison-4}\\
	&=\sup_{s_{[r]}\in\s{P}_r}\inf_{\tau\in\aff\fleft(\s{D}_A\fright)}\sum_{x\in[r]}s_x\left(D_{\max}\fleft(\tau\middle\|\rho_x\fright)-\ln\left\lVert\tau\right\rVert_1\right) \label{pf:comparison-5}\\
	&\leq\sup_{s_{[r]}\in\s{P}_r}\inf_{\tau\in\aff\fleft(\s{D}_A\fright)}\sum_{x\in[r]}s_xD_{\max}\fleft(\tau\middle\|\rho_x\fright). \label{pf:comparison-6}
\end{align}
Here \eqref{pf:comparison-1} follows from the equality between the left Umegaki radius and the multivariate log-Euclidean Chernoff divergence (see \eqref{eq:euclidean-radius} and Lemma~\ref{lem:radius-alternative}); Eq.~\eqref{pf:comparison-2} follows from the fact that $\frac{\lvert\tau\rvert}{\lVert\tau\rVert_1}\in\s{D}_A$ for all $\tau\in\aff(\s{D}_A)$; Eq.~\eqref{pf:comparison-3} uses the limit of the extended sandwiched Rényi divergence as $\alpha\searrow1$ (see Theorem~\ref{thm:sandwiched-extended}.7); Eq.~\eqref{pf:comparison-4} uses the nondecreasing monotonicity of the extended sandwiched Rényi divergence in $\alpha$ (see Theorem~\ref{thm:sandwiched-extended}.2); Eq.~\eqref{pf:comparison-5} uses the limit of the extended sandwiched Rényi divergence as $\alpha\to\infty$ (see \eqref{eq:max-extended-limit}); Eq.~\eqref{pf:comparison-6} follows from the fact that $\lVert\tau\rVert_1\geq1$ for all $\tau\in\aff(\s{D}_A)$.
\end{proof}

\begin{remark}[Strict inequality between $C^\flat$ and $-\ln\kappa$]
\label{rem:strict}
The inequality in \eqref{eq:comparison-1} is strict for certain tuples of states, as demonstrated in Example~\ref{ex:strict}.  Given Corollary~\ref{cor:comparison}, this implies that our log-Euclidean upper bound in Theorem~\ref{thm:asymptotic-state} is strictly tighter than the SDP upper bound in Ref.~\cite[Theorems~19 and 20]{mishra2024OptimalErrorExponents}.
\end{remark}

\begin{example}
\label{ex:strict}
Consider a tuple of classical states $\varrho_{[r]}\in\s{D}_Y^{[r]}$ with $r\geq2$ and $d_Y\geq 2$, where $\varrho_x\equiv\sum_{y\in[d_Y]}p_{y|x}\op{y}{y}$ for all $x\in[r]$, such that (i) $\varrho_x^0=\1$ for all $x\in[r]$ and (ii) $\varrho_x\neq\varrho_{x'}$ if $x\neq x'$.  It follows from (i) that
\begin{align}
	t&\coloneq\min_{x\in[r],y\in[d_Y]}p_{y|x}>0. \label{pf:strict-1}
\end{align}
It follows from (ii) that
\begin{align}
	u&\coloneq\min_{x\in[r]}\max_{y\in[d_Y]}\left(p_{y|x}-\min_{x'\in[r]}p_{y|x'}\right) \label{pf:strict-2}\\
	&\geq\frac{1}{d_Y}\min_{x\in[r]}\sum_{y\in[d_Y]}\left(p_{y|x}-\min_{x'\in[r]}p_{y|x'}\right) \\
	&=\frac{1}{d_Y}\min_{x\in[r]}\sum_{y\in[d_Y]}\max_{x'\in[r]}\left\lvert p_{y|x}-p_{y|x'}\right\rvert \\
	&\geq\frac{1}{d_Y}\min_{x\in[r]}\max_{x'\in[r]}\sum_{y\in[d_Y]}\left\lvert p_{y|x}-p_{y|x'}\right\rvert \\
	&=\frac{1}{d_Y}\min_{x\in[r]}\max_{x'\in[r]}\left\lVert\varrho_x-\varrho_{x'}\right\rVert_1 \\
	&=\frac{1}{d_Y}\min_{x\in[r]}\max_{x'\in[r]\setminus\{x\}}\left\lVert\varrho_x-\varrho_{x'}\right\rVert_1 \\
	&>0. \label{pf:strict-3}
\end{align}
For a probability distribution $s_{[r]}\in\s{P}_r$, define 
\begin{align}
	x_\star&\coloneq\argmax_{x\in[r]}s_x, \label{pf:strict-4}\\
	y_\star&\coloneq\argmax_{y\in[d_Y]}\left(p_{y|x_\star}-\min_{x'\in[r]}p_{y|x'}\right). \label{pf:strict-5}
\end{align}
Then we have that
\begin{align}
	&\left(\min_{x'\in[r]}p_{y|x'}\right)^{s_x} \notag\\
	&\leq p_{y|x}^{s_x}+s_xp_{y|x}^{s_x-1}\left(\min_{x'\in[r]}p_{y|x'}-p_{y|x}\right) \label{pf:strict-6}\\
	&\leq\left(1-s_x\left(p_{y|x}-\min_{x'\in[r]}p_{y|x'}\right)\right)p_{y|x}^{s_x} \label{pf:strict-7}\\
	&\leq\begin{cases}
		\left(1-\frac{u}{r}\right)p_{y_\star|x_\star}^{s_{x_\star}}&\text{if }x=x_\star\text{ and }y=y_\star, \\
		p_{y|x}^{s_x}&\text{if }x\in[r]\setminus\{x_\star\}\text{ or }y\in[d_Y]\setminus\{y_\star\}.
	\end{cases} \label{pf:strict-8}
\end{align}
Here \eqref{pf:strict-6} follows from the fact that $\lambda^s\leq\lambda_0^s+s\lambda_0^{s-1}(\lambda-\lambda_0)$ for all $\lambda,\lambda_0,s\in[0,1]$; Eq.~\eqref{pf:strict-7} follows from the fact that $p_{y|x}\in(0,1)$ for all $x\in[r]$ and $y\in[d_Y]$; Eq.~\eqref{pf:strict-8} follows from \eqref{pf:strict-2}, \eqref{pf:strict-4}, \eqref{pf:strict-5}, and the observation that $s_{x_\star}\geq\frac{1}{r}$.  It follows that
\begin{align}
	&\sum_{y\in[d_Y]}\min_{x'\in[r]}p_{y|x} \notag\\
	&=\min_{x'\in[r]}p_{y_\star|x'}+\sum_{y\in[d_Y]\setminus\{y_\star\}}\min_{x'\in[r]}p_{y|x'} \\
	&=\prod_{x\in[r]}\left(\min_{x'\in[r]}p_{y_\star|x'}\right)^{s_x}+\sum_{y\in[d_Y]\setminus\{y_\star\}}\prod_{x\in[r]}\left(\min_{x'\in[r]}p_{y|x'}\right)^{s_x} \\
	&=\left(\min_{x'\in[r]}p_{y_\star|x'}\right)^{s_{x_\star}}\prod_{x\in[r]\setminus\{x_\star\}}\left(\min_{x'\in[r]}p_{y_\star|x'}\right)^{s_x} \notag\\
	&\qquad+\sum_{y\in[d_Y]\setminus\{y_\star\}}\prod_{x\in[r]}\left(\min_{x'\in[r]}p_{y|x'}\right)^{s_x} \\
	&\leq\left(1-\frac{u}{r}\right)p_{y_\star|x_\star}^{s_{x_\star}}\prod_{x\in[r]\setminus\{x_\star\}}p_{y_\star|x}^{s_x}+\sum_{y\in[d_Y]\setminus\{y_\star\}}\prod_{x\in[r]}p_{y|x}^{s_x} \label{pf:strict-9}\\
	&=\sum_{y\in[d_Y]}\prod_{x\in[r]}p_{y|x}^{s_x}-\frac{u}{r}\prod_{x\in[r]}p_{y_\star|x}^{s_x} \\
	&\leq\sum_{y\in[d_Y]}\prod_{x\in[r]}p_{y|x}^{s_x}-\frac{ut}{r}. \label{pf:strict-10}
\end{align}
Here \eqref{pf:strict-9} follows from \eqref{pf:strict-8}; Eq.~\eqref{pf:strict-10} follows from \eqref{pf:strict-1}.  Then it follows from \eqref{eq:kappa} that
\begin{align}
	&-\ln\kappa\fleft(\varrho_{[r]}\fright) \notag\\
	&=\inf_{\gamma\in\s{Herm}_A}\left\{-\ln\tr\fleft[\gamma\fright]\colon-\varrho_x\leq\gamma\leq\varrho_x\;\forall x\in[r]\right\} \\
	&=\inf_{c_{[d_Y]}\in\spa{R}^{[d_Y]}}\left\{-\ln\left(\sum_{y\in[d_Y]}c_y\right)\colon-p_{y|x}\leq c_y\leq p_{y|x}\right. \notag\\
	&\qquad\left.\vphantom{-\ln\left(\sum_{y\in[d_Y]}c_y\right)\colon-p_{y|x}\leq c_y\leq p_{y|x}}\forall x\in[r],\;y\in[d_Y]\right\} \\
	&=-\ln\left(\sum_{y\in[d_Y]}\min_{x\in[r]}p_{y|x}\right) \\
	&\geq-\ln\inf_{s_{[r]}\in\s{P}_r}\left(\sum_{y\in[d_Y]}\prod_{x\in[r]}p_{y|x}^{s_x}-\frac{ut}{r}\right) \label{pf:strict-11}\\
	&\geq-\ln\inf_{s_{[r]}\in\s{P}_r}\left(\sum_{y\in[d_Y]}\prod_{x\in[r]}p_{y|x}^{s_x}\right) \notag\\
	&\qquad+\frac{ut}{r}\inf_{s_{[r]}\in\s{P}_r}\left(\sum_{y\in[d_Y]}\prod_{x\in[r]}p_{y|x}^{s_x}\right)^{-1} \label{pf:strict-12}\\
	&=C^\flat\fleft(\varrho_{[r]}\fright)+\frac{ut}{r}\exp\left(-C^\flat\fleft(\varrho_{[r]}\fright)\right) \label{pf:strict-13}\\
	&>C^\flat\fleft(\varrho_{[r]}\fright). \label{pf:strict-14}
\end{align}
Here \eqref{pf:strict-11} follows from \eqref{pf:strict-10}; Eq.~\eqref{pf:strict-12} follows from the fact that $\ln\lambda\leq\ln\lambda_0+\frac{\lambda-\lambda_0}{\lambda_0}$ for all $\lambda,\lambda_0\in(0,\infty)$; Eq.~\eqref{pf:strict-13} follows from the definition of the multivariate classical Chernoff divergence (see \eqref{eq:chernoff-classical}) and the fact that the multivariate log-Euclidean Chernoff divergence is a barycentric Chernoff divergence and thus reduces classically to the multivariate classical Chernoff divergence; Eq.~\eqref{pf:strict-14} follows from \eqref{pf:strict-1} and \eqref{pf:strict-3}.  Given Corollary~\ref{cor:comparison}, this shows that our log-Euclidean upper bound in Theorem~\ref{thm:asymptotic-state} is strictly tighter than the SDP upper bound in Ref.~\cite[Theorems~19 and 20]{mishra2024OptimalErrorExponents}, as claimed in Remark~\ref{rem:strict}.
\end{example}

\begin{remark}[On the (un)achievability of the log-Euclidean upper bound]
\label{rem:unachievability}
For an ensemble of classical states (or equivalently, probability distributions), it is known that the error exponent of state exclusion is exactly characterized by the multivariate classical Chernoff divergence~\cite[Theorem~6]{mishra2024OptimalErrorExponents}.  This indicates that the log-Euclidean upper bound in Theorem~\ref{thm:asymptotic-state} is achievable when the states of concern are classical.  Aside from that, the log-Euclidean upper bound is achievable when the states of concern include at least three distinct pure states, as in such a case both the multivariate log-Euclidean Chernoff divergence and the error exponent of state exclusion are infinite, a fact we show in Proposition~\ref{prop:pure} of Appendix~\ref{app:pure}.  However, there are evident (counter)examples in which the log-Euclidean upper bound is not achievable.  For instance, in the special case of $r=2$, the bivariate log-Euclidean Chernoff divergence $C^\flat(\rho_1,\rho_2)$ deviates from the actual error exponent since the latter is known to be equal to the quantum Chernoff divergence~\cite{audenaert2007DiscriminatingStatesQuantum,nussbaum2009ChernoffLowerBound}, which is defined as
\begin{align}
	C\fleft(\rho_1,\rho_2\fright)&\coloneq\sup_{s\in[0,1]}-\ln\tr\fleft[\rho_1^s\rho_2^{1-s}\fright].
\end{align}
As implied by Ref.~\cite[Theorem~2.1]{hiai1994EqualityCasesMatrix}, $C^\flat(\rho_1,\rho_2)=C(\rho_1,\rho_2)$ if and only if $\rho_1$ and $\rho_2$ commute.
\end{remark}

\section{Quantum channel exclusion}
\label{sec:exclusion-channel}

In this section, we introduce the task of quantum channel exclusion and analyze its information-theoretic limit by applying our results from the previous section.  We present a single-letter upper bound on the asymptotic error exponent of channel exclusion, given by a barycentric Chernoff channel divergence based on the Belavkin--Staszewski divergence, and it is efficiently computable via an SDP.  Our finding also implies, in the special case of $r=2$, an efficiently computable upper bound on the error exponent of symmetric binary channel discrimination.  Finally, for classical channels, we show that the upper bound is achievable and thus characterizes the exact error exponent of classical channel exclusion, which is given by the multivariate classical Chernoff divergence maximized over classical input states.

\subsection{Setting}
\label{sec:setting-channel}

\begin{figure*}[t]
\centering
\begin{quantikz}
\makeebit[angle=-30]{$\rho$}\slice[style=brown,label style=brown]{$\rho_{x,1}$} & & \wire[l][1]["R_1\qquad"{above,pos=0.5}]{a}\slice[style=brown,label style=brown]{$\rho_{x,1}'$} & & \gate[2,style={fill=red!20}]{\ch{A}_1}\slice[style=brown,label style=brown]{$\rho_{x,2}$} & \wire[l][1]["\qquad R_2"{above,pos=0.5}]{a} & \quad\dots\quad\slice[style=brown,label style=brown]{$\rho_{x,i}$} & & \wire[l][1]["R_i\qquad"{above,pos=0.5}]{a}\slice[style=brown,label style=brown]{$\rho_{x,i}'$} & & \gate[2,style={fill=red!20}]{\ch{A}_i}\slice[style=brown,label style=brown]{$\rho_{x,i+1}$} & \wire[l][1]["\qquad R_{i+1}"{above,pos=0.5}]{a} & \quad\dots\quad\slice[style=brown,label style=brown]{$\rho_{x,n}$} & & \wire[l][1]["R_n\qquad"{above,pos=0.5}]{a}\slice[style=brown,label style=brown]{$\rho_{x,n}'$} & & \meter[2,style={fill=red!20}]{\Lambda_{[r]}} \\
& & \gate[style={fill=blue!20}]{\ch{N}_x}\wire[l][1]["A_1\quad"{above,pos=0.5}]{a} & & \wire[l][1]["B_1\quad"{above,pos=0.5}]{a} & \wire[l][1]["\qquad A_2"{above,pos=0.5}]{a} & \quad\dots\quad & & \gate[style={fill=blue!20}]{\ch{N}_x}\wire[l][1]["A_i\quad"{above,pos=0.5}]{a} & & \wire[l][1]["B_i\quad"{above,pos=0.5}]{a} & \wire[l][1]["\qquad A_{i+1}"{above,pos=0.5}]{a} & \quad\dots\quad & & \gate[style={fill=blue!20}]{\ch{N}_x}\wire[l][1]["A_n\quad"{above,pos=0.5}]{a} & & \wire[l][1]["B_n\quad"{above,pos=0.5}]{a} & \setwiretype{c}
\end{quantikz}
\caption{A general adaptive strategy for quantum channel exclusion involving $n$ invocations.  Each blue box represents an invocation of the processing device being tested.  The red boxes and the initial state $\rho$ are components of the strategy that the experimenter is free to choose.}
\label{fig:exclusion-channel}
\end{figure*}
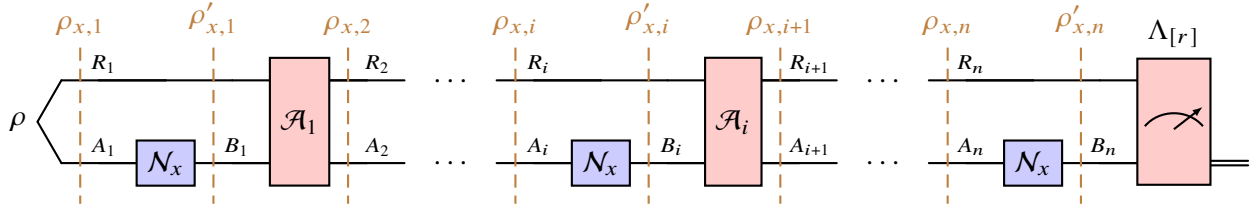

In \emph{quantum channel exclusion}, as defined in Ref.~\cite[Appendix~E~4]{huang2024ExactQuantumSensing}, the experimenter is faced with a processing device that implements an unknown channel.  The device is represented by an ensemble of channels, $\en{N}\equiv(p_{[r]},\ch{N}_{[r]})$ with $r\geq2$, and this indicates that for each $x\in[r]$, there is a prior probability $p_x$ with which the device always implements the channel $\ch{N}_x\in\s{C}_{A\to B}$.  Without loss of generality, we assume that $p_{[r]}\in\itr(\s{P}_r)$.  The experimenter's goal is to submit an index $x'\in[r]$ that \emph{differs} from the label of the channel that the device actually implements.  When $r=2$, the task reduces to symmetric binary channel discrimination (see, e.g., Refs.~\cite{harrow2010AdaptiveNonadaptiveStrategies,wilde2020AmortizedChannelDivergence}).

When the experimenter is allowed to invoke the processing device $n$ times, the most general strategy for channel exclusion is a so-called \emph{adaptive} strategy, represented by a quantum comb with $n$ empty slots~\cite{gutoski2007GeneralTheoryQuantum,chiribella2008MemoryEffectsQuantum,chiribella2008QuantumCircuitArchitecture}.\footnote{Throughout Sec.~\ref{sec:exclusion-channel}, we focus on the setting where the experimenter can invoke a single processing device $n$ times; in this setting the most general strategy is an adaptive strategy, since the $n$ invocations are intrinsically associated with a definite causal order.  In Appendix~\ref{app:indefinite}, we consider a more permissive setting where the experimenter has access to $n$ copies of a processing device and can invoke each copy a single time; in this setting, strategies beyond adaptive ones are possible, as the $n$ copies of the device can be arranged following an indefinite causal order~\cite{chiribella2013QuantumComputationsDefinite,bavaresco2021StrictHierarchyParallel} (see Remark~\ref{rem:indefinite}).}  Specifically, such a strategy consists of (i) preparing a bipartite state, (ii) passing one system of the state through the device in each of the $n$ invocations with ancillary global processing between invocations, and (iii) measuring the state after the final invocation and submitting the outcome (see Figure~\ref{fig:exclusion-channel} for an illustration).  The strategy is adaptive in the sense that the input of each invocation is allowed to depend on the outputs of all previous invocations.  When such dependence between invocations is absent, we say that the strategy is a \emph{parallel} strategy.  For $i\in[n]$, let $A_i\cong A$ and $B_i\cong B$ denote the input and output system of the $i$th invocation of the processing device, respectively.  Formally, an adaptive strategy with $n$ invocations is represented by a tuple 
\begin{align}
	\ch{S}^{(n)}&\equiv\left(\rho,\ch{A}_{[n-1]},\Lambda_{[r]}\right),
\end{align}
where $\rho\in\s{D}_{R_1A_1}$ is a state, $\ch{A}_i\in\s{C}_{R_iB_i\to R_{i+1}A_{i+1}}$ is a channel for each $i\in[n-1]$, and $\Lambda_{[r]}\in\s{M}_{R_nB_n,r}$ is a POVM.  We introduce the following shorthand for intermediate states, as also marked in Figure~\ref{fig:exclusion-channel}:
\begin{align}
	\rho_{x,1,R_1A_1}&\equiv\rho_{R_1A_1}\quad\forall x\in[r], \label{eq:adaptive-1}\\
	\rho_{x,i,R_iB_i}'&\equiv\ch{N}_{x,A_i\to B_i}\fleft[\rho_{x,i,R_iA_i}\fright]\quad\forall i\in[n],\;x\in[r], \label{eq:adaptive-2}\\
	\rho_{x,i+1,R_{i+1}A_{i+1}}&\equiv\ch{A}_{i,R_iB_i\to R_{i+1}A_{i+1}}\fleft[\rho_{x,i,R_iB_i}'\fright] \notag\\
	&\qquad\forall i\in[n-1],\;x\in[r]. \label{eq:adaptive-3}
\end{align}
The \emph{(nonasymptotic) error probability} of channel exclusion \emph{with $n$ invocations} for the ensemble $\en{N}$ is thus given by
\begin{align}
	\label{eq:error-probability-channel}
	P_\abb{err}\fleft(n;\en{N}\fright)&\coloneq\inf_{\ch{S}^{(n)}}\sum_{x\in[r]}p_x\tr\fleft[\Lambda_x\rho_{x,n}'\fright].
\end{align}
The \emph{(asymptotic) error exponent} of channel exclusion for the ensemble $\en{N}$ is defined as
\begin{align}
	\label{eq:error-exponent-channel}
	\liminf_{n\to\infty}-\frac{1}{n}\ln P_\abb{err}\fleft(n;\en{N}\fright).
\end{align}
Following the same reasoning as in state exclusion~\cite[Eq.~(30)]{mishra2024OptimalErrorExponents}, it can be shown that the error exponent of channel exclusion concerning the ensemble $\en{N}$ does not depend on the prior probability distribution $p_{[r]}$.

\subsection{Converse bound on the nonasymptotic error probability}
\label{sec:nonasymptotic-channel}

As suggested by the resemblance between \eqref{eq:error-probability-channel} and \eqref{eq:error-probability-state}, performing channel exclusion on the ensemble of channels $\en{N}$ with $n$ invocations can in part be translated into performing state exclusion on the induced ensemble of states $\en{E}'\coloneq(p_{[r]},\rho_{[r],n}')$ in the one-shot regime, where $\rho_{[r],n}'\equiv(\rho_{1,n}',\rho_{2,n}',\dots,\rho_{r,n}')$.  Specifically,
\begin{align}
	\label{eq:error-probability-channel-state}
	&P_\abb{err}\fleft(n;\en{N}\fright) \notag\\
	&=\inf_{\substack{\rho\in\s{D}_{R_1A_1}, \\ \ch{A}_{[n-1]}}}\left\{P_\abb{err}\fleft(\en{E}'\fright)\colon\ch{A}_i\in\s{C}_{R_iB_i\to R_{i+1}A_{i+1}}\;\forall i\in[n-1]\right\}.
\end{align}
This connection serves as a bridge that enables us to construct a converse bound on the error probability of channel exclusion with $n$ invocations based on the converse bound on the one-shot error probability of state exclusion, which we already established in Proposition~\ref{prop:oneshot-state}.  In what follows, we construct a converse bound in terms of the left geometric Rényi channel radius, and we further show that it is efficiently computable via an SDP.

\begin{proposition}[Converse bound on the nonasymptotic error probability for channels]
\label{prop:nonasymptotic-channel}
Let $\en{N}\equiv(p_{[r]},\ch{N}_{[r]})$ be an ensemble of channels with $p_{[r]}\in\itr(\s{P}_r)$ and $\ch{N}_{[r]}\in\s{C}_{A\to B}^{[r]}$.  Then for every positive integer $n$ and all $\alpha\in(1,2]$,
\begin{align}
	\label{eq:nonasymptotic-channel}
	&-\frac{1}{n}\ln P_\abb{err}\fleft(n;\en{N}\fright) \notag\\
	&\leq\sup_{s_{[r]}\in\s{P}_r}\inf_{\ch{T}\in\s{C}_{A\to B}}\sum_{x\in[r]}s_x\g{D}_\alpha\fleft(\ch{T}\middle\|\ch{N}_x\fright)+\frac{\alpha}{n\left(\alpha-1\right)}\ln\left(\frac{1}{p_{\min}}\right),
\end{align}
where $p_{\min}\equiv\min_{x\in[r]}p_x$.
\end{proposition}

\begin{proof}
Let $\ch{S}^{(n)}\equiv(\rho,\ch{A}_{[n-1]},\Lambda_{[r]})$ be an adaptive strategy with $n$ invocations, and let $\ch{T}\in\s{C}_{A\to B}$ be a channel.  Define the following states according to \eqref{eq:adaptive-1}--\eqref{eq:adaptive-3} but with $\ch{N}_x$ substituted with $\ch{T}$:
\begin{align}
	\omega_{1,R_1A_1}&\coloneq\rho_{R_1A_1}, \label{pf:nonasymptotic-channel-1}\\
	\omega_{i,R_iB_i}'&\coloneq\ch{T}_{A_i\to B_i}\fleft[\omega_{i,R_iA_i}\fright]\quad\forall i\in[n], \label{pf:nonasymptotic-channel-2}\\
	\omega_{i+1,R_{i+1}A_{i+1}}&\coloneq\ch{A}_{i,R_iB_i\to R_{i+1}A_{i+1}}\fleft[\omega_{i,R_iB_i}'\fright]\quad\forall i\in[n-1]. \label{pf:nonasymptotic-channel-3}
\end{align}
Note that these would be the intermediate states if the adaptive strategy $\ch{S}^{(n)}$ were to invoke the channel $\ch{T}$ instead of the processing device each time.  As such, the intermediate state just before the final measurement $\Lambda_{[r]}$ would be $\omega_n'\in\s{D}_{R_nB_n}$.  Applying Proposition~\ref{prop:oneshot-state} to the ensemble of states $\en{E}'\equiv(p_{[r]},\rho_{[r],n}')$, for all $\alpha\in(1,2]$, we have that
\begin{align}
	&-\ln P_\abb{err}\fleft(\en{E}'\fright) \notag\\
	&\leq\sup_{s_{[r]}\in\s{P}_r}\inf_{\tau\in\aff\fleft(\s{D}_{R_nB_n}\fright)}\sum_{x\in[r]}s_x\sw{D}_\alpha\fleft(\tau\middle\|\rho_{x,n}'\fright)+\frac{\alpha}{\alpha-1}\ln\left(\frac{1}{p_{\min}}\right) \\
	&\leq\sup_{s_{[r]}\in\s{P}_r}\inf_{\ch{T}\in\s{C}_{A\to B}}\sum_{x\in[r]}s_x\sw{D}_\alpha\fleft(\omega_n'\middle\|\rho_{x,n}'\fright)+\frac{\alpha}{\alpha-1}\ln\left(\frac{1}{p_{\min}}\right) \label{pf:nonasymptotic-channel-4}\\
	&\leq\sup_{s_{[r]}\in\s{P}_r}\inf_{\ch{T}\in\s{C}_{A\to B}}\sum_{x\in[r]}s_x\g{D}_\alpha\fleft(\omega_n'\middle\|\rho_{x,n}'\fright)+\frac{\alpha}{\alpha-1}\ln\left(\frac{1}{p_{\min}}\right). \label{pf:nonasymptotic-channel-5}
\end{align}
Here \eqref{pf:nonasymptotic-channel-4} follows from the fact that $\omega_n'\in\aff(\s{D}_{R_nB_n})$ for all $\ch{T}\in\s{C}_{A\to B}$; Eq.~\eqref{pf:nonasymptotic-channel-5} follows from the order between the sandwiched and geometric Rényi divergences (see \eqref{eq:order}).  Applying the chain rule for the geometric Rényi channel divergence (see \eqref{eq:chain}), we have that
\begin{align}
	\g{D}_\alpha\fleft(\omega_i'\middle\|\rho_{x,i}'\fright)&\leq\g{D}_\alpha\fleft(\ch{T}\middle\|\ch{N}_x\fright)+\g{D}_\alpha\fleft(\omega_i\middle\|\rho_{x,i}\fright)\quad\forall i\in[n], \label{pf:nonasymptotic-channel-6}\\
	\g{D}_\alpha\fleft(\omega_{i+1}\middle\|\rho_{x,i+1}\fright)&\leq\g{D}_\alpha\fleft(\omega_i'\middle\|\rho_{x,i}'\fright)\quad\forall i\in[n-1]. \label{pf:nonasymptotic-channel-7}
\end{align}
Here \eqref{pf:nonasymptotic-channel-6} follows from \eqref{eq:adaptive-2} and \eqref{pf:nonasymptotic-channel-2}; Eq.~\eqref{pf:nonasymptotic-channel-7} follows from \eqref{eq:adaptive-3} and \eqref{pf:nonasymptotic-channel-3} and the fact that $\g{D}_\alpha(\ch{A}\|\ch{A})=0$ for every channel $\ch{A}$.  It follows recursively from \eqref{pf:nonasymptotic-channel-6} and \eqref{pf:nonasymptotic-channel-7} that
\begin{align}
	\g{D}_\alpha\fleft(\omega_n'\middle\|\rho_{x,n}'\fright)&\leq n\g{D}_\alpha\fleft(\ch{T}\middle\|\ch{N}_x\fright)+\g{D}_\alpha\fleft(\omega_1\middle\|\rho_{x,1}\fright) \\
	&=n\g{D}_\alpha\fleft(\ch{T}\middle\|\ch{N}_x\fright)+\g{D}_\alpha\fleft(\rho\middle\|\rho\fright) \label{pf:nonasymptotic-channel-8}\\
	&=n\g{D}_\alpha\fleft(\ch{T}\middle\|\ch{N}_x\fright). \label{pf:nonasymptotic-channel-9}
\end{align}
Here \eqref{pf:nonasymptotic-channel-8} follows from \eqref{eq:adaptive-1} and \eqref{pf:nonasymptotic-channel-1}.  Inserting \eqref{pf:nonasymptotic-channel-9} to \eqref{pf:nonasymptotic-channel-5} and dividing both sides of the inequality by $n$, we obtain that
\begin{align}
	&-\frac{1}{n}\ln P_\abb{err}\fleft(\en{E}'\fright) \notag\\
	&\leq\sup_{s_{[r]}\in\s{P}_r}\inf_{\ch{T}\in\s{C}_{A\to B}}\sum_{x\in[r]}s_x\g{D}_\alpha\fleft(\ch{T}\middle\|\ch{N}_x\fright)+\frac{\alpha}{n\left(\alpha-1\right)}\ln\left(\frac{1}{p_{\min}}\right). \label{pf:nonasymptotic-channel-10}
\end{align}
Finally, since \eqref{pf:nonasymptotic-channel-10} holds for every adaptive strategy $\ch{S}^{(n)}$, it follows from \eqref{eq:error-probability-channel-state} that
\begin{align}
	&-\frac{1}{n}\ln P_\abb{err}\fleft(n;\en{N}\fright) \notag\\
	&=\sup_{\rho\in\s{D}_{R_1A},\ch{A}_{[n-1]}}\left\{-\frac{1}{n}\ln P_\abb{err}\fleft(\en{E}'\fright)\colon\ch{A}_i\in\s{C}_{R_iA\to R_{i+1}B}\right. \notag\\
	&\qquad\left.\vphantom{-\frac{1}{n}\ln P_\abb{err}\fleft(\en{E}'\fright)\colon\ch{A}_i\in\s{C}_{R_iA\to R_{i+1}B}}\forall i\in[n-1]\right\} \\
	&\leq\sup_{s_{[r]}\in\s{P}_r}\inf_{\ch{T}\in\s{C}_{A\to B}}\sum_{x\in[r]}s_x\g{D}_\alpha\fleft(\ch{T}\middle\|\ch{N}_x\fright)+\frac{\alpha}{n\left(\alpha-1\right)}\ln\left(\frac{1}{p_{\min}}\right),
\end{align}
thus concluding the proof.
\end{proof}

\begin{remark}[Alternative converse bound on the nonasymptotic error probability for channels]
\label{rem:nonasymptotic-alternative}
The chain rule of the geometric Rényi channel divergence (see \eqref{eq:chain}) is essential to the proof of Proposition~\ref{prop:nonasymptotic-channel}, as employed in establishing \eqref{pf:nonasymptotic-channel-9}.  On the other hand, as the regularized sandwiched Rényi channel divergence is known to follow a similar chain rule~\cite[Corollary~5.2]{fawzi2021DefiningQuantumDivergences}, putting it in the same place as the geometric Rényi channel divergence in the proof immediately leads to an alternative converse bound on the nonasymptotic error probability for channels: for every integer $n$ and all $\alpha\in(1,\infty)$,
\begin{align}
	\label{eq:nonasymptotic-alternative}
	&-\frac{1}{n}\ln P_\abb{err}\fleft(n;\en{N}\fright) \notag\\
	&\leq\sup_{s_{[r]}\in\s{P}_r}\inf_{\ch{T}\in\s{C}_{A\to B}}\sum_{x\in[r]}s_x\sw{D}_\alpha^\abb{reg}\fleft(\ch{T}\middle\|\ch{N}_x\fright)+\frac{\alpha}{n\left(\alpha-1\right)}\ln\left(\frac{1}{p_{\min}}\right),
\end{align}
where
\begin{align}
	\label{eq:sandwiched-channel-regularised}
	\sw{D}_\alpha^\abb{reg}\fleft(\ch{N}\middle\|\ch{M}\fright)&\coloneq\lim_{n\to\infty}\frac{1}{n}\sw{D}_\alpha\fleft(\ch{N}^{\otimes n}\middle\|\ch{M}^{\otimes n}\fright).
\end{align}
The existence of the limit in \eqref{eq:sandwiched-channel-regularised} was argued in the discussion below Ref.~\cite[Eq.~(21)]{fawzi2021DefiningQuantumDivergences}.  The converse bound in \eqref{eq:nonasymptotic-alternative} is at least as strong as that in Proposition~\ref{prop:nonasymptotic-channel} for $\alpha\in(1,2]$ due to the observation that $\sw{D}_\alpha^\abb{reg}(\ch{N}\|\ch{M})\leq\g{D}_\alpha(\ch{N}\|\ch{M})$, which follows from the order between the sandwiched and geometric Rényi divergences (see \eqref{eq:order}) and the additivity of the geometric Rényi channel divergence~\cite[Theorem~3.3]{fang2021GeometricRenyiDivergence}.  However, to the best of our knowledge, there is no known guarantee on the efficient computability of the regularized sandwiched Rényi channel divergence and hence that of the alternative converse bound (see Ref.~\cite[Section~5.1]{fawzi2021DefiningQuantumDivergences} for further discussions).
\end{remark}

\begin{remark}[Semidefinite representation of the channel $R^{\g{D}_\alpha}$]
\label{rem:SDP}
Applying Lemma~\ref{lem:radius-channel-alternative} to the left geometric Rényi channel radius, we know that the first term on the right-hand side of \eqref{eq:nonasymptotic-channel} is equal to the left geometric Rényi channel radius.  We thus argue that the converse bound in Proposition~\ref{prop:nonasymptotic-channel} is efficiently computable via an SDP due to the existence of a semidefinite representation of the left geometric Rényi channel radius.  To show the latter, it follows from the definition of the left geometric Rényi channel radius (see \eqref{eq:radius-channel-1}) that
\begin{align}
	&R^{\g{D}_\alpha}\fleft(\ch{N}_{[r]}\fright) \notag\\
	&=\inf_{\ch{T}\in\s{C}_{A\to B}}\max_{x\in[r]}\g{D}_\alpha\fleft(\ch{T}\middle\|\ch{N}_x\fright) \\
	&=\inf_{\ch{T}\in\s{C}_{A\to B}}\inf_{\lambda\in(0,\infty)}\left\{\ln\lambda\colon\ln\lambda\geq\g{D}_\alpha\fleft(\ch{T}\middle\|\ch{N}_x\fright)\;\forall x\in[r]\right\}. \label{pf:SDP-1}
\end{align}
Combining \eqref{pf:SDP-1} with Ref.~\cite[Theorem~3.6]{fang2021GeometricRenyiDivergence}, we obtain a semidefinite representation of the left geometric Rényi channel radius: for $\alpha\equiv1+2^{-\ell}\in(1,2]$ with $\ell$ a nonnegative integer,
\begin{align}
	&R^{\g{D}_\alpha}\fleft(\ch{N}_{[r]}\fright) \notag\\
	&=\inf_{\substack{J_\ch{T}\in\s{PSD}_{AB}, \\
		\lambda\in(0,\infty), \\
		M_{[r]}\in\s{Herm}_{AB}^{[r]}, \\
		N_{[\ell+1],[r]}\in\s{Herm}_{AB}^{[\ell+1],[r]}
	}}\left\{\begin{array}[c]{c}
		2^\ell\ln\lambda\colon \\
		\begin{bmatrix}
			M_x & J_\ch{T} \\
			J_\ch{T} & N_{\ell+1,x}
		\end{bmatrix}\geq0\;\forall x\in[r], \\
		\begin{bmatrix}
			J_\ch{T} & N_{i+1,x} \\
			N_{i+1,x} & N_{i,x}
		\end{bmatrix}\geq0 \\
		\qquad\forall i\in[\ell+1],\;x\in[r], \\
		N_{1,x}=J_{\ch{N}_x}\;\forall x\in[r], \\
		\tr_B\fleft[J_{\ch{T},AB}\fright]=\1_A, \\
		\lambda\1_A\geq\tr_B\fleft[M_{x,AB}\fright]\;\forall x\in[r]
	\end{array}\right\}. \label{pf:SDP-2}
\end{align}
For a more general rational number $\alpha\in(1,2]$, we can apply Ref.~\cite[Theorem~3]{fawzi2017LiebsConcavityTheorem} to obtain a semidefinite representation of $R^{\g{D}_\alpha}(\ch{N}_{[r]})$.
\end{remark}

\begin{remark}[Barycentric converse bound on the nonasymptotic error probability for channels]
\label{rem:nonasymptotic-barycentric}
In Proposition~\ref{prop:nonasymptotic-barycentric} of Appendix~\ref{app:nonasymptotic-channel}, as a relaxation of Proposition~\ref{prop:nonasymptotic-channel}, we present an alternative converse bound on the nonasymptotic error probability for channels in terms of a barycentric Chernoff channel divergence with a second-order correction term, following an approach similar to that of Proposition~\ref{prop:nonasymptotic-euclidean} (also see Remark~\ref{rem:nonasymptotic-euclidean}).
\end{remark}

\subsection{Upper bound on the asymptotic error exponent}
\label{sec:asymptotic-channel}

The converse bound on the nonasymptotic error probability of channel exclusion in Proposition~\ref{prop:nonasymptotic-channel} straightforwardly implies a single-letter upper bound on the asymptotic error exponent, given by the left Belavkin--Staszewski channel radius.  Since the Belavkin--Staszewski divergence reduces classically to the Kullback--Leibler divergence, the left Belavkin--Staszewski channel radius is a barycentric Chernoff channel divergence.

\begin{theorem}[Barycentric upper bound on the error exponent for channels]
\label{thm:asymptotic-channel}
Let $\en{N}\equiv(p_{[r]},\ch{N}_{[r]})$ be an ensemble of channels with $p_{[r]}\in\itr(\s{P}_r)$ and $\ch{N}_{[r]}\in\s{C}_{A\to B}^{[r]}$.  Then
\begin{align}
	\limsup_{n\to\infty}-\frac{1}{n}\ln P_\abb{err}\fleft(n;\en{N}\fright)&\leq R^\g{D}\fleft(\ch{N}_{[r]}\fright) \label{eq:asymptotic-channel-1}\\
	&=\sup_{s_{[r]}\in\s{P}_r}\inf_{\ch{T}\in\s{C}_{A\to B}}\sum_{x\in[r]}s_x\g{D}\fleft(\ch{T}\middle\|\ch{N}_x\fright). \label{eq:asymptotic-channel-2}
\end{align}
If $\ch{N}_{[r]}\in\s{C}_{Y\to B}^{[r]}$ is a tuple of classical-to-quantum channels, where $\ch{N}_x\colon\rho\mapsto\sum_{y\in[d_Y]}\bra{y}\rho\ket{y}\nu_{x,y}$ and $\nu_{x,y}\in\s{D}_B$ for all $x\in[r]$ and all $y\in[d_Y]$, then
\begin{align}
	\limsup_{n\to\infty}-\frac{1}{n}\ln P_\abb{err}\fleft(n;\en{N}\fright)&\leq\max_{y\in[d_Y]}R^\g{D}\fleft(\nu_{[r],y}\fright), \label{eq:asymptotic-channel-3}
\end{align}
where $\nu_{[r],y}\equiv(\nu_{1,y},\nu_{2,y},\dots,\nu_{r,y})$ for all $y\in[d_Y]$.
\end{theorem}

\begin{proof}
Taking the limit superior as $n\to\infty$ and the infimum over $\alpha\in(1,2]$ on both sides of \eqref{eq:nonasymptotic-channel}, we have that
\begin{align}
	&\limsup_{n\to\infty}-\frac{1}{n}\ln P_\abb{err}\fleft(n;\en{N}\fright) \notag\\
	&\leq\inf_{\alpha\in(1,2]}\sup_{s_{[r]}\in\s{P}_r}\inf_{\ch{T}\in\s{C}_{A\to B}}\sum_{x\in[r]}s_x\g{D}\fleft(\ch{T}\middle\|\ch{N}_x\fright) \\
	&=\inf_{\alpha\in(1,2]}\inf_{\ch{T}\in\s{C}_{A\to B}}\sup_{s_{[r]}\in\s{P}_r}\sum_{x\in[r]}s_x\g{D}_\alpha\fleft(\ch{T}\middle\|\ch{N}_x\fright) \label{pf:asymptotic-channel-1}\\
	&=\inf_{\ch{T}\in\s{C}_{A\to B}'}\inf_{\alpha\in(1,2]}\sup_{s_{[r]}\in\s{P}_r}\sum_{x\in[r]}s_x\g{D}_\alpha\fleft(\ch{T}\middle\|\ch{N}_x\fright) \label{pf:asymptotic-channel-2}\\
	&=\inf_{\ch{T}\in\s{C}_{A\to B}'}\sup_{s_{[r]}\in\s{P}_r}\inf_{\alpha\in(1,2]}\sum_{x\in[r]}s_x\g{D}_\alpha\fleft(\ch{T}\middle\|\ch{N}_x\fright) \label{pf:asymptotic-channel-3}\\
	&=\inf_{\ch{T}\in\s{C}_{A\to B}'}\sup_{s_{[r]}\in\s{P}_r}\sum_{x\in[r]}s_x\g{D}\fleft(\ch{T}\middle\|\ch{N}_x\fright) \label{pf:asymptotic-channel-4}\\
	&=\inf_{\ch{T}\in\s{C}_{A\to B}}\sup_{s_{[r]}\in\s{P}_r}\sum_{x\in[r]}s_x\g{D}\fleft(\ch{T}\middle\|\ch{N}_x\fright) \\
	&=R^\g{D}\fleft(\ch{N}_{[r]}\fright), \label{pf:asymptotic-channel-5}
\end{align}
where in \eqref{pf:asymptotic-channel-2}--\eqref{pf:asymptotic-channel-4} we denote $\s{C}_{A\to B}'\equiv\{\ch{T}\in\s{C}_{A\to B}\colon J_\ch{T}^0\leq\bigwedge_{x\in[r]}J_{\ch{N}_x}^0\}$.  Here \eqref{pf:asymptotic-channel-1} follows from applying Lemma~\ref{lem:radius-channel-alternative} to the left geometric Rényi channel divergence and \eqref{eq:radius-channel-2}; Eq.~\eqref{pf:asymptotic-channel-3} applies the Mosonyi--Hiai minimax theorem~\cite[Corollary~A.2]{mosonyi2011QuantumRenyiRelative}; Eq.~\eqref{pf:asymptotic-channel-4} uses the nondecreasing monotonicity of the geometric Rényi channel divergence in $\alpha$ (see \eqref{eq:monotonicity-geometric-channel}) and its limit as $\alpha\searrow1$, which is given by the Belavkin--Staszewski channel divergence (see \eqref{eq:belavkin-channel-limit}); Eq.~\eqref{pf:asymptotic-channel-5} follows from \eqref{eq:radius-channel-2}.  The application of the Mosonyi--Hiai minimax theorem in \eqref{pf:asymptotic-channel-3} is justified by the following observations: the feasible region of $s_{[r]}$ is compact; the objective function is monotonically nondecreasing in $\alpha$ due to the nondecreasing monotonicity of the geometric Rényi channel divergence in $\alpha$ (see \eqref{eq:monotonicity-geometric-channel}), and it is linear in $s_{[r]}$.  Thus we obtain \eqref{eq:asymptotic-channel-1}.  Following this, applying Lemma~\ref{lem:radius-channel-alternative} to the left Belavkin--Staszewski channel radius leads to \eqref{eq:asymptotic-channel-2}.  Applying Lemma~\ref{lem:radius-channel-classical} to the left Belavkin--Staszewski channel radius leads to \eqref{eq:asymptotic-channel-3}.
\end{proof}

\begin{remark}[Efficient computability of the barycentric upper bound on the error exponent for channels]
\label{rem:computability}
Since the Belavkin--Staszewski channel divergence is the limit of the geometric Rényi channel divergence as $\alpha\searrow1$ (see \eqref{eq:belavkin-channel-limit}), as $\ell$ increases, the right-hand side of \eqref{pf:SDP-2} quickly converges to the left Belavkin--Staszewski channel radius.  This shows that the barycentric upper bound on the error exponent of channel exclusion in Theorem~\ref{thm:asymptotic-channel} can be efficiently evaluated via an SDP.
\end{remark}

\begin{remark}[Binary channel discrimination as a special case]
\label{rem:binary-channel}
In the special case of $r=2$, Theorem~\ref{thm:asymptotic-channel} provides an upper bound on the error exponent of symmetric binary channel discrimination.  Specifically, the error exponent of symmetric channel discrimination between a pair of channels $\ch{N}_1,\ch{N}_2\in\s{C}_{A\to B}$ is bounded from above by
\begin{align}
	\sup_{s\in[0,1]}\inf_{\ch{T}\in\s{C}_{A\to B}}\left(s\g{D}\fleft(\ch{T}\middle\|\ch{N}_1\fright)+\left(1-s\right)\g{D}\fleft(\ch{T}\middle\|\ch{N}_2\fright)\right).
\end{align}
To the best of our knowledge, this is the first known such upper bound that is universal and efficiently computable.
\end{remark}

Furthermore, we show that the barycentric upper bound in Theorem~\ref{thm:asymptotic-channel} is \emph{achievable} when the channels of concern are classical, thus solving the exact error exponent of classical channel exclusion.  In this case, the error exponent has a simple expression and is given by the multivariate classical Chernoff divergence maximized over classical input states.  The proof of achievability is based on a parallel strategy and employs the error exponent of classical state exclusion in Ref.~\cite[Theorem~6]{mishra2024OptimalErrorExponents}.  This means that adaptive strategies provide no advantage for classical channel exclusion in the asymptotic regime, generalizing a similar result on symmetric binary classical channel discrimination~\cite{hayashi2009DiscriminationTwoChannels}.

\begin{theorem}[Exact error exponent for classical channels]
\label{thm:asymptotic-channel-classical}
Let $\en{P}\equiv(p_{[r]},\ch{P}_{[r]})$ be an ensemble of classical channels with $p_{[r]}\in\itr(\s{P}_r)$ and $\ch{P}_{[r]}\in\s{C}_{Y\to Z}^{[r]}$, where $\ch{P}_x\colon\rho\mapsto\sum_{y\in[d_Y]}\bra{y}\rho\ket{y}\varrho_{x,y}$ and $\varrho_{x,y}\equiv\sum_{z\in[d_Z]}p_{z|x,y}\op{z}{z}\in\s{D}_Z$ for all $x\in[r]$ and all $y\in[d_Y]$.  Then
\begin{align}
	\lim_{n\to\infty}-\frac{1}{n}\ln P_\abb{err}\fleft(n;\en{P}\fright)&=\max_{y\in[d_Y]}C\fleft(\varrho_{[r],y}\fright),
\end{align}
where $\varrho_{[r],y}\equiv(\varrho_{1,y},\varrho_{2,y},\dots,\varrho_{r,y})$ for all $y\in[d_Y]$.  Furthermore, the error exponent of classical channel exclusion can be achieved using a parallel strategy.
\end{theorem}

\begin{proof}
Since the left Belavkin--Staszewski radius is a barycentric Chernoff divergence, it reduces classically to the multivariate classical Chernoff divergence.  Then it follows from Theorem~\ref{thm:asymptotic-channel} that
\begin{align}
	\limsup_{n\to\infty}-\frac{1}{n}\ln P_\abb{err}\fleft(n;\en{P}\fright)&\leq\max_{y\in[d_Y]}R^\g{D}\fleft(\varrho_{[r],y}\fright) \\
	&=\max_{y\in[d_Y]}C\fleft(\varrho_{[r],y}\fright). \label{pf:asymptotic-channel-classical-1}
\end{align}
To show the achievability (i.e., the opposite inequality), consider the following parallel strategy with $n$ invocations: (i) for each invocation, pass the state $\op{y_\star}{y_\star}$ through the device, where $y_\star\coloneq\argmax_{y\in[d_Y]}C(\varrho_{[r],y})$, and collect the output state; (ii) perform state exclusion on the joint system containing all the $n$ output states.  This strategy reduces channel exclusion to state exclusion for the ensemble of states $\en{V}\equiv(p_{[r]},\varrho_{[r],y_\star})$, and we thus have that
\begin{align}
	\liminf_{n\to\infty}-\frac{1}{n}\ln P_\abb{err}\fleft(n;\en{P}\fright)&\geq\liminf_{n\to\infty}-\frac{1}{n}\ln P_\abb{err}\fleft(\en{V}^n\fright) \\
	&=C\fleft(\varrho_{[r],y_\star}\fright) \label{pf:asymptotic-channel-classical-4}\\
	&=\max_{y\in[d_Y]}C\fleft(\varrho_{[r],y}\fright). \label{pf:asymptotic-channel-classical-5}
\end{align}
Here \eqref{pf:asymptotic-channel-classical-4} follows from the fact that the error exponent of classical state exclusion is given by the multivariate classical Chernoff divergence~\cite[Theorem~6]{mishra2024OptimalErrorExponents}.  Combining \eqref{pf:asymptotic-channel-classical-1} and \eqref{pf:asymptotic-channel-classical-5} leads to the desired statement.
\end{proof}

\begin{remark}[Upper bound on the error exponent for channels under indefinite-causal-order strategies]
\label{rem:indefinite}
When assuming a definite causal order between multiple invocations, the most general strategy for quantum channel exclusion is an adaptive strategy, as considered so far in Sec.~\ref{sec:exclusion-channel}.  In a more permissive setting where an \emph{indefinite causal order} is allowed~\cite{chiribella2013QuantumComputationsDefinite,bavaresco2021StrictHierarchyParallel}, however, it is unclear whether the upper bound on the error exponent derived in Theorem~\ref{thm:asymptotic-channel} still applies.  In Theorem~\ref{thm:asymptotic-channel-indefinite} of Appendix~\ref{app:indefinite}, we establish an SDP upper bound on the error exponent of quantum channel exclusion without assuming a definite causal order, and the upper bound is expressed in terms of the left max-channel radius.  In fact, since the proof of Theorem~\ref{thm:asymptotic-channel-indefinite} does not assume any compositional structure of the strategy at all, the upper bound in Theorem~\ref{thm:asymptotic-channel-indefinite} applies to all strategies with the basic properties of convex linearity, positivity, and normalization.
\end{remark}

\section{Conclusion}
\label{sec:conclusion}

\subsection{Summary of results}
\label{sec:summary}

In this paper, we investigated the fundamental limits of quantum state exclusion and quantum channel exclusion from an information-theoretic perspective, in both the nonasymptotic and asymptotic regimes.  We particularly focused on developing converse bounds on the error probabilities and upper bounds on the error exponents of these tasks.  The upper bounds on the error exponents take the form of a barycentric Chernoff divergence.

In the study of quantum state exclusion, we first provided an exact characterization of the one-shot error probability in terms of the extended hypothesis-testing divergence (Proposition~\ref{prop:oneshot-hypothesis}), and this serves as the very basis of our entire converse bound approach.  Leveraging a quantitative relation between the extended hypothesis-testing divergence and the extended sandwiched Rényi divergence (Lemma~\ref{lem:hypothesis-extended}), this characterization then leads to a converse bound on the error probability in terms of an extended version of the left sandwiched Rényi radius (Proposition~\ref{prop:oneshot-state}).  Taking this converse bound to the asymptotic regime, we end up with a single-letter upper bound on the error exponent of state exclusion, given by the multivariate log-Euclidean Chernoff divergence (Theorem~\ref{thm:asymptotic-state}).  In fact, the log-Euclidean upper bound can be shown to be the tightest possible asymptotic upper bound one can construct based on the aforementioned one-shot converse bound (Proposition~\ref{prop:equality}).  In addition, it is an improvement upon the upper bound presented in Ref.~\cite[Theorem~20]{mishra2024OptimalErrorExponents} (Corollary~\ref{cor:comparison}), making it the best efficiently computable upper bound on the error exponent of state exclusion discovered so far.  Since the log-Euclidean upper bound is tight for classical states, our approach offers an alternative and simpler (both conceptually and technically) proof for the converse part of the exact error exponent of classical state exclusion~\cite[Theorem~6]{mishra2024OptimalErrorExponents} and, in the special case of $r=2$, for the classical Chernoff bound~\cite[Theorem~2.1]{nussbaum2009ChernoffLowerBound}.  This gives an additional pedagogical value to our approach, even in the quantum setting, considering that the classical Chernoff bound is an important ingredient for proving the quantum Chernoff bound~\cite[Theorem~2.2]{nussbaum2009ChernoffLowerBound}.  As a separate remark, the various properties of the extended sandwiched Rényi divergence that we established (Theorem~\ref{thm:sandwiched-extended}) played an instrumental role in the derivation and analysis of the log-Euclidean upper bound.  These properties may be of independent interest, and we expect them to find further applications apart from their immediate use in state exclusion.

We also studied the task of quantum channel exclusion, which was originally introduced in Ref.~\cite[Appendix~E~4]{huang2024ExactQuantumSensing}.  We reviewed the definitions of the nonasymptotic error probability and the ensuing asymptotic error exponent of channel exclusion, when the most general adaptive strategy is allowed.  Observing that every channel exclusion strategy involves state exclusion as a subprocedure, we derived a converse bound on the nonasymptotic error probability of channel exclusion based on the just-established converse bound for state exclusion, in terms of the left geometric Rényi channel radius (Proposition~\ref{prop:nonasymptotic-channel}).  This further leads to a single-letter upper bound on the asymptotic error exponent of channel exclusion, given by a barycentric Chernoff divergence based on the Belavkin--Staszewski divergence (Theorem~\ref{thm:asymptotic-channel}).  Both the nonasymptotic and asymptotic bounds can be efficiently evaluated via an SDP, and in the special case of $r=2$, they provide the first known universal and efficiently computable converse bounds for symmetric binary channel discrimination.  Finally, we showed that when the channels are classical, our upper bound on the error exponent of channel exclusion is achievable by a parallel strategy.  This determines the exact error exponent of classical channel exclusion, simultaneously generalizing existing results on classical state exclusion and symmetric binary classical channel discrimination.

\subsection{Future directions}
\label{sec:future}

There are still many important and unanswered questions in the study of quantum state and channel exclusion.  The foremost open problem is to provide an exact expression for the error exponent of state exclusion.  While this paper establishes a simple upper bound in terms of the multivariate log-Euclidean Chernoff divergence, our understanding of the achievability (i.e., lower bound) part of the problem is relatively limited.  What we do know is that the log-Euclidean upper bound is \emph{not} achievable in general, as it fails to reduce to the quantum Chernoff divergence in the special case of $r=2$.  This further implies that the exact expression for the error exponent must fall \emph{beyond} the family of barycentric Chernoff divergences, due to the minimality of the multivariate log-Euclidean Chernoff divergence within the family.  Another anomalous aspect of state exclusion is exclusion between pure states.  The error exponent of state exclusion becomes infinite as long as the ensemble of concern contains at least three distinct pure states, but the same cannot be said if only two distinct pure states are found.  This sharp contrast poses a significant challenge to conjecturing a legitimate potential formula for the error exponent, since few recipes are known to generate functions consistent with these special cases.  All this evidence suggests that constructing multivariate divergence measures of a novel nature is much needed for the exact expression for the error exponent of state exclusion to be determined.

On the other hand, the tasks of quantum state and channel exclusion provide an operational context and motivation for studying multivariate divergence measures.  Remarkably, while previous studies of discrimination between multiple states in different settings~\cite{bjelakovic2005QuantumVersionSanovs,brandao2010GeneralizationQuantumSteins,li2016DiscriminatingQuantumStates,brandao2020AdversarialHypothesisTesting,bunth2023EquivariantRelativeSubmajorization,hayashi2024GeneralizedQuantumSteins,lami2025SolutionGeneralizedQuantum} drew considerable attention to divergence measures with multiple arguments, few of these entailed measures are ``inherently multivariate,'' as most of them take the form of an optimization of pairwise divergences.  In contrast, state exclusion leads us to divergence measures that are multivariate in a ``truer'' sense, such as the barycentric Chernoff divergences, or more generally, the left divergence radii, involved in this paper.  We expect exclusion tasks to be the arena where such inherently multivariate divergence measures (also see Ref.~\cite{furuya2023MonotonicMultistateQuantum}) find their operational interpretations and applications.

\section*{Acknowledgements}
\label{sec:acknowledgements}

MM and MMW are grateful to Johannes Jakob Meyer for inviting them to the ``Friends of Hypothesis Testing'' email exchange, and MMW is especially grateful as it  reminded him of Ref.~\cite[Theorem~1]{vazquez-vilar2016MultipleQuantumHypothesis}.  HKM and MMW are also grateful to Prof.~Michael Nussbaum for many discussions during the development of Ref.~\cite{mishra2024OptimalErrorExponents}, which helped to shape their thinking about the quantum hypothesis exclusion problem.  

KJ acknowledges support from the NSF under grant no.~2329662.  HKM and MMW acknowledge support from the NSF under grant no.~2304816 and grant no.~2329662 and AFRL under agreement no.~FA8750-23-2-0031.  The work of MM was partially funded by the National Research, Development and Innovation Office of Hungary via the research grants K146380 and Excellence 151342, and by the Ministry of Culture and Innovation and the National Research, Development and Innovation Office within the Quantum Information National Laboratory of Hungary (Grant No.~2022-2.1.1-NL-2022-00004).

This material is based on research sponsored by Air Force Research Laboratory under agreement number FA8750-23-2-0031.  The U.S. Government is authorized to reproduce and distribute reprints for Governmental purposes notwithstanding any copyright notation thereon.  The views and conclusions contained herein are those of the authors and should not be interpreted as necessarily representing the official policies or endorsements, either expressed or implied, of Air Force Research Laboratory or the U.S. Government.


\appendices

\section{Divergence measures for channels}
\label{app:divergence}

\subsection{Proof of Lemma~\ref{lem:inheritance}(iv)}
\label{app:regularity-channel}

\begin{proposition}[Nonincreasing monotonicity and regularity by the channel $\gd{D}$-divergence in its second argument, Lemma~\ref{lem:inheritance}(iv)]
\label{prop:regularity-channel}
Let $\gd{D}$ be a bivariate generalized divergence that is monotonically nonincreasing and regular in its second argument.  Then the channel $\gd{D}$-divergence is also monotonically nonincreasing and regular in its second argument.  Consequently,
\begin{align}
	&\gd{D}\fleft(\ch{N}\middle\|\ch{M}\fright)=\sup_{\varepsilon\in(0,1)}\gd{D}\fleft(\ch{N}\middle\|\ch{M}+\varepsilon\ch{I}\fright) \notag\\
	&\qquad\forall\ch{N}\in\s{C}_{A\to B},\;\ch{M}\in\s{CP}_{A\to B}, \label{pf:inheritance-1}
\end{align}
where $\ch{I}\in\s{CP}_{A\to B}\colon\rho\mapsto\tr[\rho]\1$ denotes the unnormalized completely depolarizing map.
\end{proposition}

\begin{proof}
It follows from Lemma~\ref{lem:inheritance}(iii) that the channel $\gd{D}$-divergence is monotonically nonincreasing in its second argument.  To show that the channel $\gd{D}$-divergence is regular in its second argument, for all $\varepsilon\in(0,1)$, we observe that
\begin{align}
	&\gd{D}\fleft(\ch{N}\middle\|\ch{M}\fright) \notag\\
	&=\sup_{\rho\in\s{D}_{RA}}\gd{D}\fleft(\ch{N}_{A\to B}\fleft[\rho_{RA}\fright]\middle\|\ch{M}_{A\to B}\fleft[\rho_{RA}\fright]\fright) \label{pf:inheritance-2}\\
	&\geq\sup_{\rho\in\s{D}_{RA}}\gd{D}\fleft(\ch{N}_{A\to B}\fleft[\rho_{RA}\fright]\middle\|\ch{M}_{A\to B}\fleft[\rho_{RA}\fright]+\varepsilon\ch{I}_{A\to B}\fleft[\rho_{RA}\fright]\fright) \label{pf:inheritance-3}\\
	&=\gd{D}\fleft(\ch{N}\middle\|\ch{M}+\varepsilon\ch{I}\fright)\quad\forall\ch{N}\in\s{C}_{A\to B},\;\ch{M}\in\s{CP}_{A\to B}. \label{pf:inheritance-4}
\end{align}
Here \eqref{pf:inheritance-2} and \eqref{pf:inheritance-4} follow from the definition of the channel $\gd{D}$-divergence (see \eqref{eq:divergence-channel}); Eq.~\eqref{pf:inheritance-3} uses the nonincreasing monotonicity of $\gd{D}$ in its second argument (see \eqref{eq:antimonotonicity}).  Meanwhile, we have that
\begin{align}
	&\gd{D}\fleft(\ch{N}\middle\|\ch{M}\fright) \notag\\
	&=\sup_{\rho\in\s{D}_{RA}}\gd{D}\fleft(\ch{N}_{A\to B}\fleft[\rho_{RA}\fright]\middle\|\ch{M}_{A\to B}\fleft[\rho_{RA}\fright]\fright) \\
	&=\sup_{\rho\in\s{D}_{RA}}\lim_{\varepsilon\searrow0}\gd{D}\fleft(\ch{N}_{A\to B}\fleft[\rho_{RA}\fright]\middle\|\ch{M}_{A\to B}\fleft[\rho_{RA}\fright]+\varepsilon\1_{RB}\fright) \label{pf:inheritance-5}\\
	&=\sup_{\rho\in\s{D}_{RA}}\sup_{\varepsilon\in(0,1)}\gd{D}\fleft(\ch{N}_{A\to B}\fleft[\rho_{RA}\fright]\middle\|\ch{M}_{A\to B}\fleft[\rho_{RA}\fright]+\varepsilon\1_{RB}\fright) \label{pf:inheritance-6}\\
	&\leq\sup_{\rho\in\s{D}_{RA}}\sup_{\varepsilon\in(0,1)}\gd{D}\fleft(\ch{N}_{A\to B}\fleft[\rho_{RA}\fright]\middle\|\ch{M}_{A\to B}\fleft[\rho_{RA}\fright]+\varepsilon\rho_R\otimes\1_B\fright) \label{pf:inheritance-7}\\
	&=\sup_{\varepsilon\in(0,1)}\sup_{\rho\in\s{D}_{RA}}\gd{D}\fleft(\ch{N}_{A\to B}\fleft[\rho_{RA}\fright]\middle\|\left(\ch{M}_{A\to B}+\varepsilon\ch{I}_{A\to B}\right)\fleft[\rho_{RA}\fright]\fright) \\
	&=\sup_{\varepsilon\in(0,1)}\gd{D}\fleft(\ch{N}\middle\|\ch{M}+\varepsilon\ch{I}\fright)\quad\forall\ch{N}\in\s{C}_{A\to B},\;\ch{M}\in\s{CP}_{A\to B}. \label{pf:inheritance-8}
\end{align}
Here \eqref{pf:inheritance-5} uses the regularity of $\gd{D}$ in its second argument (see \eqref{eq:regularity}); Eqs.~\eqref{pf:inheritance-6} and \eqref{pf:inheritance-7} use the nonincreasing monotonicity of $\gd{D}$ in its second argument (see \eqref{eq:antimonotonicity}).  Combining \eqref{pf:inheritance-4} and \eqref{pf:inheritance-8}, we have that
\begin{align}
	&\gd{D}\fleft(\ch{N}\middle\|\ch{M}\fright) \notag\\
	&=\sup_{\varepsilon\in(0,1)}\gd{D}\fleft(\ch{N}\middle\|\ch{M}+\varepsilon\ch{I}\fright) \\
	&=\lim_{\varepsilon\searrow0}\gd{D}\fleft(\ch{N}\middle\|\ch{M}+\varepsilon\ch{I}\fright)\quad\forall\ch{N}\in\s{C}_{A\to B},\;\ch{M}\in\s{CP}_{A\to B}. \label{pf:inheritance-9}
\end{align}
Here \eqref{pf:inheritance-9} uses the nonincreasing monotonicity of the channel $\gd{D}$-divergence (see \eqref{eq:antimonotonicity-channel}).
\end{proof}

\subsection{Proof of Lemma~\ref{lem:radius-channel-alternative}}
\label{app:radius-channel-alternative}

{
\newcounter{temporary-theorem}
\setcounter{temporary-theorem}{\value{theorem}}
\setcounter{theorem}{100}
\renewcommand{\thetheorem}{\ref{lem:radius-channel-classical}}
\begin{lemma}[Alternative representation of the channel $R^\gd{D}$, restatement]
Let $\gd{D}$ be a bivariate generalized divergence with the following properties: \textnormal{(i)} quasiconvexity in its first argument; \textnormal{(ii)} lower semicontinuity in its first argument; \textnormal{(iii)} nonincreasing monotonicity in its second argument; \textnormal{(iv)} regularity in its second argument; and \textnormal{(v)} $\gd{D}(\ch{N}\|\ch{M})<\infty$ if $J_\ch{N}^0\leq J_\ch{M}^0$.  Then
\begin{align}
	&R^\gd{D}\fleft(\ch{N}_{[r]}\fright)=\sup_{s_{[r]}\in\s{P}_r}\inf_{\ch{T}\in\s{C}_{A\to B}}\sum_{x\in[r]}s_x\gd{D}\fleft(\ch{T}\middle\|\ch{N}_x\fright) \notag\\
	&\qquad\forall\ch{N}_{[r]}\in\s{CP}_{A\to B}^{[r]}.
\end{align}
\end{lemma}
\setcounter{theorem}{\value{temporary-theorem}}
}

\begin{proof}
For all $\varepsilon\in(0,1)$, we observe that
\begin{align}
	&\left(\ch{T}_{A\to B}\fleft[\rho_{RA}\fright]\right)^0 \notag\\
	&\leq\left(\tr_B\circ\ch{T}_{A\to B}\fleft[\rho_{RA}\fright]\otimes\1_B\right)^0 \\
	&=\left(\rho_R\otimes\1_B\right)^0 \\
	&=\left(\ch{I}_{A\to B}\fleft[\rho_{RA}\fright]\right)^0 \\
	&=\left(\ch{N}_{A\to B}\fleft[\rho_{RA}\fright]+\varepsilon\ch{I}_{A\to B}\fleft[\rho_{RA}\fright]\right)^0 \notag\\
	&\qquad\forall\ch{T}\in\s{C}_{A\to B},\;\ch{N}\in\s{CP}_{A\to B},\;\rho\in\s{D}_{RA}.
\end{align}
This implies that
\begin{align}
	\gd{D}\fleft(\ch{T}\middle\|\ch{N}+\varepsilon\ch{I}\fright)&<\infty\quad\forall\ch{T}\in\s{C}_{A\to B},\;\ch{N}\in\s{CP}_{A\to B}.
\end{align}
Let $\ch{N}_{[r]}\in\s{CP}_{A\to B}^{[r]}$ be a tuple of completely positive maps.  It follows from \eqref{eq:radius-channel-2} that
\begin{align}
	R^\gd{D}\fleft(\ch{N}_{[r]}\fright)&=\inf_{\ch{T}\in\s{C}_{A\to B}}\sup_{s_{[r]}\in\s{P}_r}\sum_{x\in[r]}s_x\gd{D}\fleft(\ch{T}\middle\|\ch{N}_x\fright) \\
	&=\inf_{\ch{T}\in\s{C}_{A\to B}}\sup_{s_{[r]}\in\s{P}_r}\sup_{\varepsilon\in(0,1)}\sum_{x\in[r]}s_x\gd{D}\fleft(\ch{T}\middle\|\ch{N}_x+\varepsilon\ch{I}\fright) \label{pf:alternative-13}\\
	&=\sup_{\varepsilon\in(0,1)}\inf_{\ch{T}\in\s{C}_{A\to B}}\sup_{s_{[r]}\in\s{P}_r}\sum_{x\in[r]}s_x\gd{D}\fleft(\ch{T}\middle\|\ch{N}_x+\varepsilon\ch{I}\fright) \label{pf:alternative-14}\\
	&=\sup_{\varepsilon\in(0,1)}\sup_{s_{[r]}\in\s{P}_r}\inf_{\ch{T}\in\s{C}_{A\to B}}\sum_{x\in[r]}s_x\gd{D}\fleft(\ch{T}\middle\|\ch{N}_x+\varepsilon\ch{I}\fright) \label{pf:alternative-15}\\
	&=\sup_{s_{[r]}\in\s{P}_r}\inf_{\ch{T}\in\s{C}_{A\to B}}\sup_{\varepsilon\in(0,1)}\sum_{x\in[r]}s_x\gd{D}\fleft(\ch{T}\middle\|\ch{N}_x+\varepsilon\ch{I}\fright) \label{pf:alternative-16}\\
	&=\sup_{s_{[r]}\in\s{P}_r}\inf_{\ch{T}\in\s{C}_{A\to B}}\sum_{x\in[r]}s_x\gd{D}\fleft(\ch{T}\middle\|\ch{N}_x\fright) \label{pf:alternative-17}.
\end{align}
Here \eqref{pf:alternative-13} and \eqref{pf:alternative-17} follows from \eqref{pf:inheritance-1}; Eqs.~\eqref{pf:alternative-14} and \eqref{pf:alternative-16} apply the Mosonyi--Hiai minimax theorem~\cite[Corollary~A.2]{mosonyi2011QuantumRenyiRelative}; Eq.~\eqref{pf:alternative-15} applies the Sion minimax theorem~\cite{sion1958GeneralMinimaxTheorems}.  The application of the Mosonyi--Hiai minimax theorem in \eqref{pf:alternative-14} and \eqref{pf:alternative-16} is justified by the following observations: the feasible region of $\ch{T}$ is compact; the objective function is monotonically nonincreasing in $\varepsilon$ due to the inheritance of the nonincreasing monotonicity by the channel $\gd{D}$-divergence in its second argument (see Lemma~\ref{lem:inheritance}(iii)), and it is lower semicontinuous in $\ch{T}$ due to the inheritance of the lower semicontinuity by the channel $\gd{D}$-divergence in its first argument (see Lemma~\ref{lem:inheritance}(ii)).  The application of the Sion minimax theorem in \eqref{pf:alternative-15} is justified by the following observations: the feasible regions of $s_{[r]}$ and $\ch{T}$ are both convex and compact; the objective function is linear in $s_{[r]}$, quasiconvex in $\ch{T}$ due to the inheritance of the quasiconvexity by the channel $\gd{D}$-divergence in its first argument (see Lemma~\ref{lem:inheritance}(i)), and lower semicontinuous in $\ch{T}$ due to the inheritance of the lower semicontinuity by the channel $\gd{D}$-divergence in its first argument (see Lemma~\ref{lem:inheritance}).
\end{proof}

\subsection{Proof of Lemma~\ref{lem:radius-channel-classical}}
\label{app:radius-channel-classical}

{
\setcounter{temporary-theorem}{\value{theorem}}
\setcounter{theorem}{100}
\renewcommand{\thetheorem}{\ref{lem:radius-channel-classical}}
\begin{lemma}[$R^\gd{D}$ of classical-to-quantum channels, restatement]
Let $\ch{N}_{[r]}\in\s{CP}_{Y\to B}^{[r]}$ be a tuple of classical-to-quantum positive maps, where $\ch{N}_x\colon\rho\mapsto\sum_{y\in[d_Y]}\bra{y}\rho\ket{y}\nu_{x,y}$ and $\nu_{x,y}\in\s{PSD}_B$ for all $x\in[r]$ and all $y\in[d_Y]$.  Let $\gd{D}$ be a bivariate generalized divergence that is jointly quasiconvex and lower semicontinuous in its first argument.  Then
\begin{align}
	R^\gd{D}\fleft(\ch{N}_{[r]}\fright)&=\max_{y\in[d_Y]}R^\gd{D}\fleft(\nu_{[r],y}\fright),
\end{align}
where $\nu_{[r],y}\equiv(\nu_{1,y},\nu_{2,y},\dots,\nu_{r,y})$ for all $y\in[d_Y]$.
\end{lemma}
\setcounter{theorem}{\value{temporary-theorem}}
}

\begin{proof}
For every classical-to-quantum channel $\ch{N}\in\s{C}_{Y\to B}\colon\rho\mapsto\sum_{y\in[d_Y]}\bra{y}\rho\ket{y}\nu_y$ and every classical-to-quantum positive map $\ch{N}\in\s{CP}_{Y\to B}\colon\rho\mapsto\sum_{y\in[d_Y]}\bra{y}\rho\ket{y}\mu_y$, where $\nu_y,\mu_y\in\s{D}_B$ for all $y\in[d_Y]$, we observe that
\begin{align}
	&\gd{D}\fleft(\ch{N}_{Y\to B}\fleft[\rho_{RY}\fright]\middle\|\ch{M}_{Y\to B}\fleft[\rho_{RY}\fright]\fright) \notag\\
	&=\gd{D}\fleft(\sum_{y\in[d_Y]}\bra{y}_Y\rho_{RY}\ket{y}_Y\otimes\nu_{y,B}\middle\|\sum_{y\in[d_Y]}\bra{y}_Y\rho_{RY}\ket{y}_Y\otimes\mu_{y,B}\fright) \label{pf:radius-channel-classical-1}\\
	&\leq\max_{y\in[d_Y]}\gd{D}\fleft(\frac{\bra{y}_Y\rho_{RY}\ket{y}_Y}{\bra{y}_Y\rho_Y\ket{y}_Y}\otimes\nu_y\middle\|\frac{\bra{y}_Y\rho_{RY}\ket{y}_Y}{\bra{y}_Y\rho_Y\ket{y}_Y}\otimes\mu_y\fright) \label{pf:radius-channel-classical-2}\\
	&=\max_{y\in[d_Y]}\gd{D}\fleft(\nu_y\middle\|\mu_y\fright)\quad\forall\rho\in\s{D}_{RY}. \label{pf:radius-channel-classical-3}
\end{align}
Here \eqref{pf:radius-channel-classical-2} uses the joint quasiconvexity of $\gd{D}$ (see \eqref{eq:quasiconvexity}); Eq.~\eqref{pf:radius-channel-classical-3} follows from the DPI under the preparation channel and under the partial-trace channel (see \eqref{eq:DPI}).  Taking the supremum over $\rho\in\s{D}_{RY}$ on the left-hand side of \eqref{pf:radius-channel-classical-1} and on the right-hand side of \eqref{pf:radius-channel-classical-3} with $d_R$ allowed to be arbitrarily large, we have that
\begin{align}
	&\sup_{\rho_{RY}}\gd{D}\fleft(\ch{N}_{Y\to B}\fleft[\rho_{RY}\fright]\middle\|\ch{M}_{Y\to B}\fleft[\rho_{RY}\fright]\fright) \notag\\
	&\leq\max_{y\in[d_Y]}\gd{D}\fleft(\nu_y\middle\|\mu_y\fright) \\
	&\leq\sup_{\rho_{RY}}\gd{D}\fleft(\ch{N}_{Y\to B}\fleft[\rho_{RY}\fright]\middle\|\ch{M}_{Y\to B}\fleft[\rho_{RY}\fright]\fright).  \label{pf:radius-channel-classical-4}
\end{align}
Here \eqref{pf:radius-channel-classical-4} follows from the fact that $\rho_{RY}=1_R\otimes\op{y}{y}_Y$ is a feasible solution to the supremization.  This implies that
\begin{align}
	\gd{D}\fleft(\ch{N}\middle\|\ch{M}\fright)&=\sup_{\rho_{RY}}\gd{D}\fleft(\ch{N}_{Y\to B}\fleft[\rho_{RY}\fright]\middle\|\ch{M}_{Y\to B}\fleft[\rho_{RY}\fright]\fright) \\
	&=\max_{y\in[d_Y]}\gd{D}\fleft(\nu_y\middle\|\mu_y\fright). \label{pf:radius-channel-classical-5}
\end{align}
Then it follows from the definition of the left channel $\gd{D}$-radius (see \eqref{eq:radius-channel-1}) that
\begin{align}
	R^\gd{D}\fleft(\ch{N}_{[r]}\fright)&=\inf_{\ch{T}\in\s{C}_{A\to B}}\max_{x\in[r]}\gd{D}\fleft(\ch{T}\middle\|\ch{N}_x\fright) \\
	&=\inf_{\tau_{[d_Y]}\in\s{D}_B^{[d_Y]}}\max_{x\in[r]}\max_{y\in[d_Y]}\gd{D}\fleft(\tau_y\middle\|\nu_{x,y}\fright) \label{pf:radius-channel-classical-6}\\
	&\geq\max_{y\in[d_Y]}\inf_{\tau_y\in\s{D}_B}\max_{x\in[r]}\gd{D}\fleft(\tau_y\middle\|\nu_{x,y}\fright) \label{pf:radius-channel-classical-7}\\
	&=\max_{y\in[d_Y]}R^\gd{D}\fleft(\nu_{[r],y}\fright). \label{pf:radius-channel-classical-8}
\end{align}
Here \eqref{pf:radius-channel-classical-6} follows from \eqref{pf:radius-channel-classical-5} and applies the substitution $\ch{T}\colon\rho\mapsto\sum_{y\in[d_Y]}\bra{y}\rho\ket{y}\tau_y$ and $\ch{N}_x\colon\rho\mapsto\sum_{y\in[d_Y]}\bra{y}\rho\ket{y}\nu_{x,y}$ for all $x\in[r]$; Eq.~\eqref{pf:radius-channel-classical-7} applies an exchange in order between the infimization and the maximization.  To show the opposite inequality, for $y\in[d_Y]$, define the following state:
\begin{align}
	\tau_{\star,y}&\coloneq\argmin_{\tau_y\in\s{D}_B}\max_{x\in[r]}\gd{D}\fleft(\tau_y\middle\|\nu_{x,y}\fright). \label{pf:radius-channel-classical-9}
\end{align}
The existence of $\tau_{\star,y}$ as defined in \eqref{pf:radius-channel-classical-9} is guaranteed by the following observations: the feasible region of $\tau_{\star,y}$ is compact; the objective function is lower semicontinuous in $\tau_y$ due to the fact that the maximum of a family of lower semicontinuous functions is lower semicontinuous.  On the other hand, it follows from \eqref{pf:radius-channel-classical-6} that
\begin{align}
	R^\gd{D}\fleft(\ch{N}_{[r]}\fright)&=\inf_{\tau_{[d_Y]}\in\s{D}_B^{[d_Y]}}\max_{x\in[r]}\max_{y\in[d_Y]}\gd{D}\fleft(\tau_y\middle\|\nu_{x,y}\fright) \label{pf:radius-channel-classical-10}\\
	&\leq\max_{y\in[d_Y]}\max_{x\in[r]}\gd{D}\fleft(\tau_{\star,y}\middle\|\nu_{x,y}\fright) \label{pf:radius-channel-classical-11}\\
	&=\max_{y\in[d_Y]}\min_{\tau_y\in\s{D}_B}\max_{x\in[r]}\gd{D}\fleft(\tau_y\middle\|\nu_{x,y}\fright) \\
	&=\max_{y\in[d_Y]}R^\gd{D}\fleft(\nu_{[r],y}\fright). \label{pf:radius-channel-classical-12}
\end{align}
Here \eqref{pf:radius-channel-classical-11} follows from \eqref{pf:radius-channel-classical-9} and the fact that $\tau_{\star,[d_A]}\in\s{D}_B^{[d_A]}$ is a feasible solution to the infimization in \eqref{pf:radius-channel-classical-10}; Eq.~\eqref{pf:radius-channel-classical-12} follows from the definition of the left $\gd{D}$-radius (see \eqref{eq:radius-1}).  Combining \eqref{pf:radius-channel-classical-8} and \eqref{pf:radius-channel-classical-12} leads to the desired statement.
\end{proof}

\section{Proof of Theorem~\ref{thm:sandwiched-extended}}
\label{app:sandwiched-extended}

\subsection{Proof of Theorem~\ref{thm:sandwiched-extended}.3}
\label{app:additivity}

\begin{proposition}[Additivity of the extended $\sw{D}_\alpha$, Theorem~\ref{thm:sandwiched-extended}.3]
\label{prop:additivity}
Let $\gamma\in\s{Herm}_A$ and $\gamma'\in\s{Herm}_B$ be two Hermitian operators with $\gamma,\gamma'\neq0$, and let $\sigma\in\s{PSD}_A$ and $\sigma'\in\s{PSD}_B$ be two positive semidefinite operators.  Then for all $\alpha\in(1,\infty)$,
\begin{align}
	\label{pf:additivity-1}
	\sw{D}_\alpha\fleft(\gamma\otimes\gamma'\middle\|\sigma\otimes\sigma'\fright)&=\sw{D}_\alpha\fleft(\gamma\middle\|\sigma\fright)+\sw{D}_\alpha\fleft(\gamma'\middle\|\sigma'\fright).
\end{align}
\end{proposition}

\begin{proof}
The proof is a straightforward generalization of that for the original sandwiched Rényi divergence whose first argument is a state.  Note that $\gamma^0\otimes{\gamma'}^0\leq(\sigma\otimes\sigma')^0$ if and only if $\gamma^0\leq\sigma^0$ and ${\gamma'}^0\leq{\sigma'}^0$.  If $(\gamma\otimes\gamma')^0\nleq(\sigma\otimes\sigma')^0$, then both sides of \eqref{pf:additivity-1} are equal to $\infty$.  Otherwise, it follows from the definition of the extended sandwiched Rényi divergence (see \eqref{eq:sandwiched-extended}) that
\begin{align}
	&\sw{D}_\alpha\fleft(\gamma\otimes\gamma'\middle\|\sigma\otimes\sigma'\fright) \notag\\
	&=\frac{1}{\alpha-1}\ln\left\lVert\left(\sigma\otimes\sigma'\right)^\frac{1-\alpha}{2\alpha}\left(\gamma\otimes\gamma'\right)\left(\sigma\otimes\sigma'\right)^\frac{1-\alpha}{2\alpha}\right\rVert_\alpha^\alpha \\
	&=\frac{1}{\alpha-1}\ln\left\lVert\sigma^\frac{1-\alpha}{2\alpha}\gamma\sigma^\frac{1-\alpha}{2\alpha}\otimes{\sigma'}^\frac{1-\alpha}{2\alpha}\gamma'{\sigma'}^\frac{1-\alpha}{2\alpha}\right\rVert_\alpha^\alpha \\
	&=\frac{1}{\alpha-1}\ln\left(\left\lVert\sigma^\frac{1-\alpha}{2\alpha}\gamma\sigma^\frac{1-\alpha}{2\alpha}\right\rVert_\alpha^\alpha\left\lVert{\sigma'}^\frac{1-\alpha}{2\alpha}\gamma'{\sigma'}^\frac{1-\alpha}{2\alpha}\right\rVert_\alpha^\alpha\right) \label{pf:additivity-2}\\
	&=\frac{1}{\alpha-1}\left(\ln\left\lVert\sigma^\frac{1-\alpha}{2\alpha}\gamma\sigma^\frac{1-\alpha}{2\alpha}\right\rVert_\alpha^\alpha+\ln\left\lVert{\sigma'}^\frac{1-\alpha}{2\alpha}\gamma'{\sigma'}^\frac{1-\alpha}{2\alpha}\right\rVert_\alpha^\alpha\right) \\
	&=\sw{D}_\alpha\fleft(\gamma\middle\|\sigma\fright)+\sw{D}_\alpha\fleft(\gamma'\middle\|\sigma'\fright).
\end{align}
Here \eqref{pf:additivity-2} uses the multiplicativity of the $\alpha$-norm under tensor products (see, e.g., Ref.~\cite[Eq.~(2.2.96)]{khatri2024PrinciplesQuantumCommunication}): for all $\alpha\in[1,\infty)$, we have that $\lVert\xi\otimes\xi'\rVert_\alpha=\lVert\xi\rVert_\alpha\lVert\xi'\rVert_\alpha$ for all $\xi\in\spa{B}_A$ and all $\xi'\in\spa{B}_B$.
\end{proof}

\subsection{Proof of Theorem~\ref{thm:sandwiched-extended}.4}
\label{app:direct-sum}

\begin{proposition}[Direct-sum property of the extended $\sw{D}_\alpha$, Theorem~\ref{thm:sandwiched-extended}.4]
\label{prop:direct-sum}
Let $\gamma_{[r]}\in\s{Herm}_A^{[r]}$ be a tuple of Hermitian operators such that $\gamma_x\neq0$ for all $x\in[r]$, and let $\sigma_{[r]}\in\s{PSD}_A^{[r]}$ be a tuple of positive semidefinite operators.  Let $p_{[r]}\in\s{P}_r$ be a probability distribution, and let $q_{[r]}\in[0,\infty)^{[r]}$ be a tuple of nonnegative real numbers.  Then for all $\alpha\in(1,\infty)$,
\begin{align}
	\label{pf:direct-sum-1}
	\sw{Q}_\alpha\fleft(\cq{\gamma}_{XA}\middle\|\cq{\sigma}_{XA}\fright)&=\sum_{x\in[r]}p_x^\alpha q_x^{1-\alpha}\sw{Q}_\alpha\fleft(\gamma_x\middle\|\sigma_x\fright),
\end{align}
where $\sw{Q}_\alpha(\gamma\|\sigma)\equiv\exp((\alpha-1)\sw{D}_\alpha(\gamma\|\sigma))$ denotes the extended sandwiched Rényi quasi-divergence and
\begin{align}
	\cq{\gamma}_{XA}&\equiv\sum_{x\in[r]}p_x\op{x}{x}_X\otimes\gamma_{x,A}, \\
	\cq{\sigma}_{XA}&\equiv\sum_{x\in[r]}q_x\op{x}{x}_X\otimes\sigma_{x,A}.
\end{align}
\end{proposition}

\begin{proof}
The proof is a straightforward generalization of that for the original sandwiched Rényi divergence whose first argument is a state.  Note that $\cq{\gamma}_{XA}^0\leq\cq{\sigma}_{XA}^0$ if and only if $(p_x\gamma_x)^0\leq(q_x\sigma_x)^0$ for all $x\in[r]$.  If $\cq{\gamma}_{XA}^0\nleq\cq{\sigma}_{XA}^0$, then both sides of \eqref{pf:direct-sum-1} are equal to $\infty$.  Otherwise, it follows from the definition of the extended sandwiched Rényi quasi-divergence (see \eqref{eq:sandwiched-extended}) that
\begin{align}
	&\sw{Q}_\alpha\fleft(\cq{\gamma}_{XA}\middle\|\cq{\sigma}_{XA}\fright) \notag\\
	&=\left\lVert\cq{\sigma}_{XA}^\frac{1-\alpha}{2\alpha}\cq{\gamma}_{XA}\cq{\sigma}_{XA}^\frac{1-\alpha}{2\alpha}\right\rVert_\alpha^\alpha \\
	&=\left\lVert\left(\sum_{x\in[r]}q_x^\frac{1-\alpha}{2\alpha}\op{x}{x}\otimes\sigma_x^\frac{1-\alpha}{2\alpha}\right)\left(\sum_{x\in[r]}p_x\op{x}{x}\otimes\gamma_x\right)\right. \notag\\
	&\qquad\left.\left(\sum_{x\in[r]}q_x^\frac{1-\alpha}{2\alpha}\op{x}{x}\otimes\sigma_x^\frac{1-\alpha}{2\alpha}\right)\right\rVert_\alpha^\alpha \\
	&=\left\lVert\sum_{x\in[r]}p_xq_x^\frac{1-\alpha}{\alpha}\op{x}{x}\otimes\sigma_x^\frac{1-\alpha}{2\alpha}\gamma_x\sigma_x^\frac{1-\alpha}{2\alpha}\right\rVert_\alpha^\alpha \\
	&=\sum_{x\in[r]}\left\lVert p_xq_x^\frac{1-\alpha}{\alpha}\sigma_x^\frac{1-\alpha}{2\alpha}\gamma_x\sigma_x^\frac{1-\alpha}{2\alpha}\right\rVert_\alpha^\alpha \label{pf:direct-sum-2}\\
	&=\sum_{x\in[r]}p_x^\alpha q_x^{1-\alpha}\sw{Q}_\alpha\fleft(\gamma_x\middle\|\sigma_x\fright).
\end{align}
Here \eqref{pf:direct-sum-2} uses the direct-sum property of the $\alpha$-norm (see, e.g., Ref.~\cite[Eq.~(2.2.97)]{khatri2024PrinciplesQuantumCommunication}): for all $\alpha\in[1,\infty)$,
\begin{align}
	\left\lVert\sum_{x\in[r]}\op{x}{x}\otimes\xi_x\right\rVert_\alpha^\alpha&=\sum_{x\in[r]}\left\lVert\xi_x\right\rVert_\alpha^\alpha\quad\forall\xi_{[r]}\in\spa{B}_A^{[r]}, \label{pf:direct-sum-3}
\end{align}
thus concluding the proof.
\end{proof}

\subsection{Proof of Theorem~\ref{thm:sandwiched-extended}.5}
\label{app:quasiconvexity}

\begin{proposition}[Joint quasiconvexity of the extended $\sw{D}_\alpha$, Theorem~\ref{thm:sandwiched-extended}.5]
\label{prop:quasiconvexity}
Let $\gamma_{[r]}\in\s{Herm}_A^{[r]}$ be a tuple of Hermitian operators such that $\gamma_x\neq0$ for all $x\in[r]$, and let $\sigma_{[r]}\in\s{PSD}_A^{[r]}$ be a tuple of positive semidefinite operators.  Let $p_{[r]}\in\s{P}_r$ be a probability distribution.  Then for all $\alpha\in(1,\infty)$,
\begin{align}
	\sw{D}_\alpha\fleft(\sum_{x\in[r]}p_x\gamma_x\middle\|\sum_{x\in[r]}p_x\sigma_x\fright)&\leq\max_{x\in[r]}\sw{D}_\alpha\fleft(\gamma_x\middle\|\sigma_x\fright).
\end{align}
\end{proposition}

\begin{proof}
The proof is a straightforward generalization of that for the original sandwiched Rényi divergence whose first argument is a state.  Consider that
\begin{align}
	&\sw{D}_\alpha\fleft(\sum_{x\in[r]}p_x\gamma_x\middle\|\sum_{x\in[r]}p_x\sigma_x\fright) \notag\\
	&\leq\sw{D}_\alpha\fleft(\sum_{x\in[r]}p_x\op{x}{x}\otimes\gamma_x\middle\|\sum_{x\in[r]}p_x\op{x}{x}\otimes\sigma_x\fright) \label{pf:quasiconvexity-1}\\
	&=\frac{1}{\alpha-1}\ln\sw{Q}_\alpha\fleft(\sum_{x\in[r]}p_x\op{x}{x}\otimes\gamma_x\middle\|\sum_{x\in[r]}p_x\op{x}{x}\otimes\sigma_x\fright) \\
	&=\frac{1}{\alpha-1}\ln\left(\sum_{x\in[r]}p_x\sw{Q}_\alpha\fleft(\gamma_x\middle\|\sigma_x\fright)\right) \label{pf:quasiconvexity-2}\\
	&\leq\max_{x\in[r]}\frac{1}{\alpha-1}\ln\sw{Q}_\alpha\fleft(\gamma_x\middle\|\sigma_x\fright) \\
	&=\max_{x\in[r]}\sw{D}_\alpha\fleft(\gamma_x\middle\|\sigma_x\fright).
\end{align}
Here \eqref{pf:quasiconvexity-1} follows from the DPI under the partial-trace channel (see Theorem~\ref{thm:sandwiched-extended}.1); Eq.~\eqref{pf:quasiconvexity-2} uses the direct-sum property of the extended sandwiched Rényi quasi-divergence (see Theorem~\ref{thm:sandwiched-extended}.4).
\end{proof}

\subsection{Proof of Theorem~\ref{thm:sandwiched-extended}.6}
\label{app:monotonicity}

\begin{proposition}[Nonincreasing monotonicity of the extended $\sw{D}_\alpha$ in its second argument, Theorem~\ref{thm:sandwiched-extended}.6]
\label{prop:monotonicity}
Let $\gamma\in\s{Herm}_A$ be a Hermitian operator with $\gamma\neq0$, and let $\sigma,\sigma'\in\s{PSD}_A$ be two positive semidefinite operators.  Then for all $\alpha\in(1,\infty)$,
\begin{align}
	\label{pf:monotonicity-1}
	\sw{D}_\alpha\fleft(\gamma\middle\|\sigma+\sigma'\fright)&\leq\sw{D}_\alpha\fleft(\gamma\middle\|\sigma\fright).
\end{align}
\end{proposition}

\begin{proof}
The proof is a straightforward generalization of that for the original sandwiched Rényi divergence whose first argument is a state.  Denoting
\begin{align}
	\cq{\gamma}_{XA}&\equiv\op{1}{1}_X\otimes\gamma_A, \\
	\cq{\sigma}_{XA}&\equiv\op{1}{1}_X\otimes\sigma_A+\op{2}{2}\otimes\sigma_A',
\end{align}
we have that
\begin{align}
	\sw{D}_\alpha\fleft(\gamma\middle\|\sigma+\sigma'\fright)&\leq\sw{D}_\alpha\fleft(\cq{\gamma}_{XA}\middle\|\cq{\sigma}_{XA}\fright) \label{pf:monotonicity-2}\\
	&=\frac{1}{\alpha-1}\ln\sw{Q}_\alpha\fleft(\cq{\gamma}_{XA}\middle\|\cq{\sigma}_{XA}\fright) \\
	&=\frac{1}{\alpha-1}\ln\sw{Q}_\alpha\fleft(\gamma\middle\|\sigma\fright) \label{pf:monotonicity-3}\\
	&=\sw{D}_\alpha\fleft(\gamma\middle\|\sigma\fright).
\end{align}
Here \eqref{pf:monotonicity-2} follows from the DPI under the partial-trace channel (see Theorem~\ref{thm:sandwiched-extended}.1); Eq.~\eqref{pf:monotonicity-3} applies the direct-sum property of the tuple of Hermitian operators $\gamma_{[2]}=(\gamma,\gamma)\in\s{Herm}_A^{[2]}$, the tuple of positive semidefinite operators $\sigma_{[2]}=(\sigma,\sigma')\in\s{Herm}_A^{[2]}$, the probability distribution $p_{[2]}=(1,0)\in\s{P}_2$, and the tuple of nonnegative real numbers $q_{[2]}=(1,1)\in[0,\infty)^{[2]}$ (see Theorem~\ref{thm:sandwiched-extended}.4).
\end{proof}

\subsection{Proof of Theorem~\ref{thm:sandwiched-extended}.7}
\label{app:limit}

\begin{proposition}[Limit of the extended $\sw{D}_\alpha$ as $\alpha\searrow1$, Theorem~\ref{thm:sandwiched-extended}.7]
\label{prop:limit}
Let $\gamma\in\s{Herm}_A$ be a Hermitian operator with $\gamma\neq0$, and let $\sigma\in\s{PSD}_A$ be a positive semidefinite operator.  Then
\begin{align}
	\label{pf:limit-1}
	\lim_{\alpha\searrow1}\left(\sw{D}_\alpha\fleft(\gamma\middle\|\sigma\fright)-\frac{\alpha}{\alpha-1}\ln\left\lVert\gamma\right\rVert_1\right)&=D\fleft(\frac{\left\lvert\gamma\right\rvert}{\left\lVert\gamma\right\rVert_1}\middle\|\sigma\fright).
\end{align}
\end{proposition}

\begin{proof}
The proof generalizes that for the original sandwiched Rényi divergence, whose first argument is a state.  If $\gamma^0\nleq\sigma^0$, then both sides of \eqref{pf:limit-1} are equal to $\infty$.  Otherwise, we have that $\gamma^0\leq\sigma^0$ and thus $\sigma^0\gamma\sigma^0=\gamma$.  Recall that
\begin{align}
	\sw{Q}_\alpha\fleft(\gamma\middle\|\sigma\fright)&=\exp\left(\left(\alpha-1\right)\sw{D}_\alpha\fleft(\gamma\middle\|\sigma\fright)\right) \\
	&=\left\lVert\sigma^\frac{1-\alpha}{2\alpha}\gamma\sigma^\frac{1-\alpha}{2\alpha}\right\rVert_\alpha^\alpha \\
	&=\tr\fleft[\left(\sigma^\frac{1-\alpha}{2\alpha}\gamma\sigma^\frac{1-\alpha}{\alpha}\gamma\sigma^\frac{1-\alpha}{2\alpha}\right)^\frac{\alpha}{2}\fright]. \label{pf:limit-2}
\end{align}
Here \eqref{pf:limit-2} follows from the definition of the $\alpha$-norm (see \eqref{eq:norm}).  Note that
\begin{align}
	\sw{Q}_1\fleft(\gamma\middle\|\sigma\fright)&=\tr\fleft[\left(\sigma^0\gamma\sigma^0\gamma\sigma^0\right)^\frac{1}{2}\fright] \\
	&=\tr\fleft[\left(\gamma^2\right)^\frac{1}{2}\fright] \label{pf:limit-3}\\
	&=\left\lVert\gamma\right\rVert_1.
\end{align}
Here \eqref{pf:limit-3} follows from the facts that $\sigma^0=\sigma^0\sigma^0$ and $\sigma^0\gamma\sigma^0=\gamma$.  It follows that
\begin{align}
	&\lim_{\alpha\searrow1}\left(\sw{D}_\alpha\fleft(\gamma\middle\|\sigma\fright)-\frac{\alpha}{\alpha-1}\ln\left\lVert\gamma\right\rVert_1\right) \notag\\
	&=\lim_{\alpha\searrow1}\frac{1}{\alpha-1}\left(\ln\sw{Q}_\alpha\fleft(\gamma\middle\|\sigma\fright)-\ln\left\lVert\gamma\right\rVert_1\right)-\ln\left\lVert\gamma\right\rVert_1 \\
	&=\lim_{\alpha\searrow1}\frac{1}{\alpha-1}\left(\ln\sw{Q}_\alpha\fleft(\gamma\middle\|\sigma\fright)-\ln\sw{Q}_1(\gamma\|\sigma)\right)-\ln\left\lVert\gamma\right\rVert_1 \\
	&=\frac{\de}{\de\alpha}\ln\sw{Q}_\alpha\fleft(\gamma\middle\|\sigma\fright)\bigg|_{\alpha=1}-\ln\left\lVert\gamma\right\rVert_1 \\
	&=\frac{1}{\sw{Q}_1\fleft(\gamma\middle\|\sigma\fright)}\frac{\de}{\de\alpha}\sw{Q}_\alpha\fleft(\gamma\middle\|\sigma\fright)\bigg|_{\alpha=1}-\ln\left\lVert\gamma\right\rVert_1 \\
	&=\frac{1}{\left\lVert\gamma\right\rVert_1}\frac{\de}{\de\alpha}\sw{Q}_\alpha\fleft(\gamma\middle\|\sigma\fright)\bigg|_{\alpha=1}-\ln\left\lVert\gamma\right\rVert_1. \label{pf:limit-4}
\end{align}
For $\alpha,\beta\in[1,\infty)$, define the following function:
\begin{align}
	\sw{Q}_{\alpha,\beta}\fleft(\gamma\middle\|\sigma\fright)&\coloneq\tr\fleft[\left(\sigma^\frac{1-\alpha}{2\alpha}\gamma\sigma^\frac{1-\alpha}{\alpha}\gamma\sigma^\frac{1-\alpha}{2\alpha}\right)^\frac{\beta}{2}\fright].
\end{align}
Note from \eqref{pf:limit-2} that $\sw{Q}_{\alpha,\alpha}(\gamma\|\sigma)=\sw{Q}_\alpha(\gamma\|\sigma)$.  As
\begin{align}
	\frac{\de}{\de\alpha}\sw{Q}_\alpha\fleft(\gamma\middle\|\sigma\fright)\bigg|_{\alpha=1}&=\frac{\de}{\de\alpha}\sw{Q}_{\alpha,1}\fleft(\gamma\middle\|\sigma\fright)\bigg|_{\alpha=1}+\frac{\de}{\de\beta}\sw{Q}_{1,\beta}\fleft(\gamma\middle\|\sigma\fright)\bigg|_{\beta=1}, \label{pf:limit-5}
\end{align}
we need to calculate both terms on the right-hand side of \eqref{pf:limit-5} to find the desired limit.  The first term on the right-hand side of \eqref{pf:limit-5} is given by
\begin{align}
	&\frac{\de}{\de\alpha}\sw{Q}_{\alpha,1}\fleft(\gamma\middle\|\sigma\fright) \notag\\
	&=\frac{\de}{\de\alpha}\tr\fleft[\left(\sigma^\frac{1-\alpha}{2\alpha}\gamma\sigma^\frac{1-\alpha}{\alpha}\gamma\sigma^\frac{1-\alpha}{2\alpha}\right)^\frac{1}{2}\fright] \\
	&=\frac{1}{2}\tr\fleft[\left(\sigma^\frac{1-\alpha}{2\alpha}\gamma\sigma^\frac{1-\alpha}{\alpha}\gamma\sigma^\frac{1-\alpha}{2\alpha}\right)^{-\frac{1}{2}}\frac{\de}{\de\alpha}\left(\sigma^\frac{1-\alpha}{2\alpha}\gamma\sigma^\frac{1-\alpha}{\alpha}\gamma\sigma^\frac{1-\alpha}{2\alpha}\right)\fright] \\
	&=\frac{1}{2}\tr\fleft[\left(\sigma^\frac{1-\alpha}{2\alpha}\gamma\sigma^\frac{1-\alpha}{\alpha}\gamma\sigma^\frac{1-\alpha}{2\alpha}\right)^{-\frac{1}{2}}\right. \notag\\
	&\qquad\left.\left(-\frac{1}{2\alpha^2}\sigma^\frac{1-\alpha}{2\alpha}\left(\ln\sigma\right)\gamma\sigma^\frac{1-\alpha}{\alpha}\gamma\sigma^\frac{1-\alpha}{2\alpha}\right.\right. \notag\\
	&\qquad\left.\left.\vphantom{}-\frac{1}{\alpha^2}\sigma^\frac{1-\alpha}{2\alpha}\gamma\sigma^\frac{1-\alpha}{\alpha}\left(\ln\sigma\right)\gamma\sigma^\frac{1-\alpha}{2\alpha}\right.\right. \notag\\
	&\qquad\left.\left.\vphantom{}-\frac{1}{2\alpha^2}\sigma^\frac{1-\alpha}{2\alpha}\gamma\sigma^\frac{1-\alpha}{\alpha}\gamma\sigma^\frac{1-\alpha}{2\alpha}\left(\ln\sigma\right)\right)\fright],
\end{align}
which implies that
\begin{align}
	&\frac{\de}{\de\alpha}\sw{Q}_{\alpha,1}(\gamma\|\sigma)\bigg|_{\alpha=1} \notag\\
	&=-\frac{1}{4}\tr\fleft[\left(\sigma^0\gamma\sigma^0\gamma\sigma^0\right)^{-\frac{1}{2}}\left(\sigma^0\left(\ln\sigma\right)\gamma\sigma^0\gamma\sigma^0\right.\right. \notag\\
	&\qquad\left.\vphantom{\left(\sigma^0\gamma\sigma^0\gamma\sigma^0\right)^{-\frac{1}{2}}}\left.\vphantom{}+2\sigma^0\gamma\sigma^0\left(\ln\sigma\right)\gamma\sigma^0+\sigma^0\gamma\sigma^0\gamma\sigma^0\left(\ln\sigma\right)\right)\fright] \\
	&=-\frac{1}{4}\left(\tr\fleft[\left(\gamma^2\right)^{-\frac{1}{2}}\left(\ln\sigma\right)\gamma^2\fright]+2\tr\fleft[\left(\gamma^2\right)^{-\frac{1}{2}}\gamma\left(\ln\sigma\right)\gamma\fright]\right. \notag\\
	&\qquad\left.\vphantom{}+\tr\fleft[\left(\gamma^2\right)^{-\frac{1}{2}}\gamma^2\ln\sigma\fright]\right) \\
	&=-\tr\fleft[\left\lvert\gamma\right\rvert\ln\sigma\fright]. \label{pf:limit-6}
\end{align}
The second term on the right-hand side of \eqref{pf:limit-5} is given by
\begin{align}
	\frac{\de}{\de\beta}\sw{Q}_{1,\beta}\fleft(\gamma\middle\|\sigma\fright)\bigg|_{\beta=1}&=\frac{\de}{\de\beta}\tr\fleft[\left(\sigma^0\gamma\sigma^0\gamma\sigma^0\right)^\frac{\beta}{2}\fright]\Bigg|_{\beta=1} \\
	&=\frac{\de}{\de\beta}\tr\fleft[\left\lvert\gamma\right\rvert^\beta\fright]\bigg|_{\beta=1} \\
	&=\tr\fleft[\left\lvert\gamma\right\rvert^\beta\ln\left\lvert\gamma\right\rvert\fright]\big|_{\beta=1} \\
	&=\tr\fleft[\left\lvert\gamma\right\rvert\ln\left\lvert\gamma\right\rvert\fright]. \label{pf:limit-7}
\end{align}
Inserting \eqref{pf:limit-6} and \eqref{pf:limit-7} to \eqref{pf:limit-5}, and then to \eqref{pf:limit-4}, we obtain that
\begin{align}
	&\lim_{\alpha\searrow1}\left(\sw{D}_\alpha\fleft(\gamma\middle\|\sigma\fright)-\frac{\alpha}{\alpha-1}\ln\left\lVert\gamma\right\rVert_1\right) \notag\\
	&=\frac{1}{\left\lVert\gamma\right\rVert_1}\left(\frac{\de}{\de\alpha}\sw{Q}_{\alpha,1}\fleft(\gamma\middle\|\sigma\fright)\bigg|_{\alpha=1}+\frac{\de}{\de\beta}\sw{Q}_{1,\beta}\fleft(\gamma\middle\|\sigma\fright)\bigg|_{\beta=1}\right)-\ln\left\lVert\gamma\right\rVert_1 \\
	&=\frac{1}{\left\lVert\gamma\right\rVert_1}\tr\fleft[\left\lvert\gamma\right\rvert\left(\ln\left\lvert\gamma\right\rvert-\ln\sigma\right)\fright]-\ln\left\lVert\gamma\right\rVert_1 \\
	&=\tr\fleft[\frac{\left\lvert\gamma\right\rvert}{\left\lVert\gamma\right\rVert_1}\left(\ln\left(\frac{\left\lvert\gamma\right\rvert}{\left\lVert\gamma\right\rVert_1}\right)-\ln\sigma\right)\fright] \\
	&=D\fleft(\frac{\left\lvert\gamma\right\rvert}{\left\lVert\gamma\right\rVert_1}\middle\|\sigma\fright). \label{pf:limit-8}
\end{align}
Here \eqref{pf:limit-8} follows from the fact that $\frac{\lvert\gamma\rvert}{\lVert\gamma\rVert_1}\in\s{D}_A$ for all $\gamma\in\s{Herm}_A$ and the definition of the Umegaki divergence (see \eqref{eq:umegaki}).
\end{proof}

\section{Proof of Lemma~\ref{lem:hypothesis-extended}}
\label{app:hypothesis-extended}

{
\setcounter{temporary-theorem}{\value{theorem}}
\setcounter{theorem}{100}
\renewcommand{\thetheorem}{\ref{lem:hypothesis-extended}}
\begin{lemma}[Connection between the extended $D_\abb{H}^\varepsilon$ and $\sw{D}_\alpha$, restatement]
Let $\tau\in\aff(\s{D}_A)$ be a unit-trace Hermitian operator, and let $\sigma\in\s{PSD}_A$ be a positive semidefinite operator.  Then for all $\varepsilon\in[0,1)$ and all $\alpha\in(1,\infty)$,
\begin{align}
	\label{pf:hypothesis-extended-1}
	D_\abb{H}^\varepsilon\fleft(\tau\middle\|\sigma\fright)&\leq\sw{D}_\alpha\fleft(\tau\middle\|\sigma\fright)+\frac{\alpha}{\alpha-1}\ln\left(\frac{1}{1-\varepsilon}\right).
\end{align}
\end{lemma}
\setcounter{theorem}{\value{temporary-theorem}}
}

\begin{proof}
The proof is a straightforward generalization of that regarding the original hypothesis testing and sandwiched Rényi divergences whose first argument is a state.  If $\tau^0\nleq\sigma^0$, then the right-hand side of \eqref{pf:hypothesis-extended-1} is equal to $\infty$.  Henceforth we assume that $\tau^0\leq\sigma^0$.  Let $\Lambda\in\s{PSD}_A$ be a positive semidefinite operator such that $\Lambda\leq\1$ and $\tr[\Lambda\tau]\geq1-\varepsilon$.  Let
\begin{align}
	\ch{M}\in\s{C}_{A\to X}&\colon\rho\mapsto\tr\fleft[\Lambda\rho\fright]\op{1}{1}+\tr\fleft[\left(\1-\Lambda\right)\rho\fright]\op{2}{2}
\end{align}
denote the quantum-to-classical channel represented by the POVM $(\Lambda,\1-\Lambda)\in\s{M}_{A,2}$.  It follows from the DPI under the channel $\ch{M}$ (see Theorem~\ref{thm:sandwiched-extended}.1) that
\begin{align}
	\sw{D}_\alpha\fleft(\tau\middle\|\sigma\fright)&\geq\sw{D}_\alpha\fleft(\ch{M}\fleft[\tau\fright]\middle\|\ch{M}\fleft[\sigma\fright]\fright) \\
	&=\frac{1}{\alpha-1}\ln\left\lVert\left(\ch{M}\fleft[\sigma\fright]\right)^\frac{1-\alpha}{2\alpha}\ch{M}\fleft[\tau\fright]\left(\ch{M}\fleft[\sigma\fright]\right)^\frac{1-\alpha}{2\alpha}\right\rVert_\alpha^\alpha \\
	&=\frac{1}{\alpha-1}\ln\left(\left\lvert\tr\fleft[\Lambda\tau\fright]\tr\fleft[\Lambda\sigma\fright]^\frac{1-\alpha}{\alpha}\right\rvert^\alpha\right.\notag\\
	&\qquad\left.\vphantom{}+\left\lvert\tr\fleft[\left(\1-\Lambda\right)\tau\fright]\tr\fleft[\left(\1-\Lambda\right)\sigma\fright]^\frac{1-\alpha}{\alpha}\right\rvert^\alpha\right) \label{pf:hypothesis-extended-2}\\
	&\geq\frac{1}{\alpha-1}\ln\left(\left\lvert\tr\fleft[\Lambda\tau\fright]\tr\fleft[\Lambda\sigma\fright]^\frac{1-\alpha}{\alpha}\right\rvert^\alpha\right) \\
	&=\frac{1}{\alpha-1}\ln\left(\tr\fleft[\Lambda\tau\fright]^\alpha\tr\fleft[\Lambda\sigma\fright]^{1-\alpha}\right) \label{pf:hypothesis-extended-3}\\
	&=\frac{\alpha}{\alpha-1}\ln\tr\fleft[\Lambda\tau\fright]-\ln\tr\fleft[\Lambda\sigma\fright] \\
	&\geq\frac{\alpha}{\alpha-1}\ln\left(1-\varepsilon\right)-\ln\tr\fleft[\Lambda\sigma\fright]. \label{pf:hypothesis-extended-4}
\end{align}
Here \eqref{pf:hypothesis-extended-2} uses the direct-sum property of the $\alpha$-norm (see \eqref{pf:direct-sum-3}); Eqs.~\eqref{pf:hypothesis-extended-3} and \eqref{pf:hypothesis-extended-4} follow from the fact that $\tr[\Lambda\tau]\geq1-\varepsilon>0$.  Since \eqref{pf:hypothesis-extended-4} holds for every positive semidefinite operator $\Lambda\in\s{PSD}_A$ such that $\Lambda\leq\1$ and $\tr[\Lambda\tau]\geq1-\varepsilon$, it follows that
\begin{align}
	D_\abb{H}^\varepsilon\fleft(\tau\middle\|\sigma\fright)&=\sup_{\Lambda\in\s{PSD}_A}\left\{-\ln\tr\fleft[\Lambda\sigma\fright]\colon\Lambda\leq\1,\;\tr\fleft[\Lambda\tau\fright]\geq1-\varepsilon\right\} \\
	&\leq\sw{D}_\alpha\fleft(\tau\middle\|\sigma\fright)+\frac{\alpha}{\alpha-1}\ln\left(\frac{1}{1-\varepsilon}\right),
\end{align}
thus concluding the proof.
\end{proof}

\section{Characterization of the one-shot error probability of state discrimination}
\label{app:discrimination}

\begin{proposition}[Characterization of error probability of state discrimination, {\cite[Theorem~1]{vazquez-vilar2016MultipleQuantumHypothesis}}]
\label{prop:discrimination}
Let $\en{E}\equiv(p_{[r]},\rho_{[r]})$ be an ensemble of states with $p_{[r]}\in\itr(\s{P}_r)$ and $\rho_{[r]}\in\s{D}_A^{[r]}$.  Then
\begin{align}
	-\ln P_\abb{err}^\abb{disc}\fleft(\en{E}\fright)&=\inf_{\tau\in\s{D}_A}D_\abb{H}^\frac{1}{r}\fleft(\pi_X\otimes\tau_A\middle\|\cq{\rho}_{XA}\fright),
\end{align}
where
\begin{align}
	P_\abb{err}^\abb{disc}\fleft(\en{E}\fright)&\coloneq1-P_\abb{succ}^\abb{disc}\fleft(\en{E}\fright) \\
	&=1-\sup_{\Lambda_{[r]}\in\s{M}_{A,r}}\sum_{x\in[r]}p_x\tr\fleft[\Lambda_x\rho_x\fright]
\end{align}
denotes the (one-shot) error probability of state discrimination for the ensemble $\en{E}$ and
\begin{align}
	\cq{\rho}_{XA}&\equiv\sum_{x\in[r]}p_x\op{x}{x}_X\otimes\rho_{x,A}.
\end{align}
\end{proposition}

\begin{proof}
For a POVM $\Lambda_{[r]}\in\s{M}_{A,r}$, let
\begin{align}
	\cq{\Lambda}_{XA}&\equiv\sum_{x\in[r]}\op{x}{x}_X\otimes\Lambda_{x,A}.
\end{align}
We observe that $\cq{\Lambda}\in\s{PSD}_{XA}$ and that
\begin{align}
	\cq{\Lambda}_{XA}&\leq\1_{XA}, \label{pf:discrimination-1}\\
	\tr\fleft[\cq{\Lambda}_{XA}\cq{\rho}_{XA}\fright]&=\sum_{x\in[r]}p_x\tr\fleft[\Lambda_x\rho_x\fright], \label{pf:discrimination-2}\\
	\tr\fleft[\cq{\Lambda}_{XA}\left(\pi_X\otimes\tau_A\right)\fright]&=\frac{1}{r}\sum_{x\in[r]}\tr\fleft[\Lambda_x\tau\fright]=\frac{1}{r}\quad\forall\tau\in\s{D}_A. \label{pf:discrimination-3}
\end{align}
Then for every state $\tau\in\s{D}_A$, we have that
\begin{align}
	P_\abb{err}^\abb{disc}(\en{E})&=1-\sup_{\Lambda_{[r]}\in\s{M}_{A,r}}\sum_{x\in[r]}p_x\tr\fleft[\Lambda_x\rho_x\fright] \\
	&\geq1-\sup_{\cq{\Lambda}\in\s{PSD}_{XA}}\left\{\tr\fleft[\cq{\Lambda}_{XA}\cq{\rho}_{XA}\fright]\colon\cq{\Lambda}_{XA}\leq\1_{XA},\vphantom{\tr\fleft[\cq{\Lambda}_{XA}\left(\pi_X\otimes\tau_A\right)\fright]\leq\frac{1}{r}}\right.\notag\\
	&\qquad\left.\tr\fleft[\cq{\Lambda}_{XA}\left(\pi_X\otimes\tau_A\right)\fright]\leq\frac{1}{r}\right\} \label{pf:discrimination-4}\\
	&=\inf_{\cq{\Lambda}'\in\s{PSD}_{XA}}\left\{\tr\fleft[\cq{\Lambda}_{XA}'\cq{\rho}_{XA}\fright]\colon\cq{\Lambda}_{XA}'\leq\1_{XA},\vphantom{\tr\fleft[\cq{\Lambda}_{XA}'\left(\pi_X\otimes\tau_A\right)\fright]\geq1-\frac{1}{r}}\right. \notag\\
	&\qquad\left.\tr\fleft[\cq{\Lambda}_{XA}'\left(\pi_X\otimes\tau_A\right)\fright]\geq1-\frac{1}{r}\right\} \label{pf:discrimination-5}\\
	&=\exp\left(-D_\abb{H}^\frac{1}{r}\fleft(\pi_X\otimes\tau_A\middle\|\cq{\rho}_{XA}\fright)\right). \label{pf:discrimination-6}
\end{align}
Here \eqref{pf:discrimination-4} follows from \eqref{pf:discrimination-1}--\eqref{pf:discrimination-3}; Eq.~\eqref{pf:discrimination-5} applies the substitution $\cq{\Lambda}=\1-\cq{\Lambda}'$; Eq.~\eqref{pf:discrimination-6} follows from the definition of the hypothesis-testing divergence (see \eqref{eq:hypothesis-extended}).  This shows that
\begin{align}
	-\ln P_\abb{err}^\abb{disc}\fleft(\en{E}\fright)&\leq\inf_{\tau\in\s{D}_A}D_\abb{H}^\frac{1}{r}\fleft(\pi_X\otimes\tau_A\middle\|\cq{\rho}_{XA}\fright). \label{pf:discrimination-7}
\end{align}
To show the opposite inequality, it follows from the dual SDP formulation of $P_\abb{err}^\abb{disc}(\en{E})$~\cite{yuen1975OptimumTestingMultiple} (also see Refs.~\cite[Eq.~(19)]{audenaert2014UpperBoundsError} and \cite[Theorem~1]{konig2009OperationalMeaningMin}) that
\begin{align}
	&P_\abb{err}^\abb{disc}\fleft(\en{E}\fright) \notag\\
	&=1-\inf_{\gamma\in\s{Herm}_A}\left\{\tr\fleft[\gamma\fright]\colon\gamma\geq p_x\rho_x\;\forall x\in[r]\right\} \label{pf:discrimination-8}\\
	&=1-\inf_{\gamma\in\s{PSD}_A}\left\{\tr\fleft[\gamma\fright]\colon\1_X\otimes\gamma_A\geq\cq{\rho}_{XA}\right\} \\
	&=1-\inf_{\tau\in\s{D}_A}\inf_{\eta\in[0,\infty)}\left\{\eta\colon\eta\1_X\otimes\tau_A\geq\cq{\rho}_{XA}\right\}  \label{pf:discrimination-9}\\
	&=1-\inf_{\tau\in\s{D}_A}\inf_{\eta\in[0,\infty)}\sup_{\cq{\Lambda}\in\s{PSD}_{XA}}\left(\eta+\tr\fleft[\cq{\Lambda}_{XA}\left(\cq{\rho}_{XA}-\eta\1_X\otimes\tau_A\right)\fright]\right) \label{pf:discrimination-10}\\
	&=1-\inf_{\tau\in\s{D}_A}\inf_{\eta\in[0,\infty)}\sup_{\cq{\Lambda}\in\s{PSD}_{XA}}\left(\tr\fleft[\cq{\Lambda}_{XA}\cq{\rho}_{XA}\fright]\right. \notag\\
	&\qquad\left.\vphantom{}+\eta\left(1-\tr\fleft[\cq{\Lambda}_{XA}\left(\1_X\otimes\tau_A\right)\fright]\right)\right) \\
	&\leq1-\inf_{\tau\in\s{D}_A}\sup_{\cq{\Lambda}\in\s{PSD}_A}\inf_{\eta\in[0,\infty)}\left(\tr\fleft[\cq{\Lambda}_{XA}\cq{\rho}_{XA}\fright]\right. \notag\\
	&\qquad\left.\vphantom{}+\eta\left(1-\tr\fleft[\cq{\Lambda}_{XA}\left(\1_X\otimes\tau_A\right)\fright]\right)\right) \label{pf:discrimination-11}\\
	&=1-\inf_{\tau\in\s{D}_A}\sup_{\cq{\Lambda}\in\s{PSD}_{XA}}\left\{\tr\fleft[\cq{\Lambda}_{XA}\cq{\rho}_{XA}\fright]\colon\tr\fleft[\cq{\Lambda}_{XA}\left(\1_X\otimes\tau_A\right)\fright]\leq1\right\} \label{pf:discrimination-12}\\
	&\leq1-\inf_{\tau\in\s{D}_A}\sup_{\cq{\Lambda}\in\s{PSD}_{XA}}\left\{\tr\fleft[\cq{\Lambda}_{XA}\cq{\rho}_{XA}\fright]\colon\cq{\Lambda}_{XA}\leq\1_{XA},\vphantom{\tr\fleft[\cq{\Lambda}_{XA}\left(\pi_X\otimes\tau_A\right)\fright]\leq\frac{1}{r}}\right. \notag\\
	&\qquad\left.\tr\fleft[\cq{\Lambda}_{XA}\left(\pi_X\otimes\tau_A\right)\fright]\leq\frac{1}{r}\right\} \label{pf:discrimination-13}\\
	&=\sup_{\tau\in\s{D}_A}\inf_{\cq{\Lambda}'\in\s{PSD}_{XA}}\left\{\tr\fleft[\cq{\Lambda}_{XA}'\cq{\rho}_{XA}\fright]\colon\cq{\Lambda}_{XA}'\leq\1_{XA},\vphantom{\tr\fleft[\cq{\Lambda}_{XA}'\left(\pi_X\otimes\tau_A\right)\fright]\geq1-\frac{1}{r}}\right. \notag\\
	&\qquad\left.\tr\fleft[\cq{\Lambda}_{XA}'\left(\pi_X\otimes\tau_A\right)\fright]\geq1-\frac{1}{r}\right\} \label{pf:discrimination-14}\\
	&=\sup_{\tau\in\s{D}_A}\exp\left(-D_\abb{H}^\frac{1}{r}\fleft(\pi_X\otimes\tau_A\middle\|\cq{\rho}_{XA}\fright)\right).  \label{pf:discrimination-15}
\end{align}
Here \eqref{pf:discrimination-8} is the dual SDP formulation of $P_\abb{err}^\abb{disc}(\en{E})$; Eq.~\eqref{pf:discrimination-9} applies the substitution $\sigma=\eta\tau$ and follows from the fact that $\eta=0$ is a feasible solution to the infimization; Eq.~\eqref{pf:discrimination-10} introduces $\cq{\Lambda}\in\s{PSD}_{XA}$ as the Lagrange multiplier corresponding to the constraint $\eta\1_X\otimes\tau_A\geq\cq{\rho}_{XA}$; Eq.~\eqref{pf:discrimination-11} follows from an exchange in order between the infimization and the supremization; Eq.~\eqref{pf:discrimination-12} follows from a reverse Lagrangian reasoning that eliminates the multiplier $\eta$; Eq.~\eqref{pf:discrimination-13} follows from the fact that an additional constraint on $\cq{\Lambda}$ does not increase the optimal value of the objective function; Eq.~\eqref{pf:discrimination-14} applies the substitution $\cq{\Lambda}=\1-\cq{\Lambda}'$.  Combining \eqref{pf:discrimination-7} and \eqref{pf:discrimination-15} leads to the desired statement.
\end{proof}

\section{A closed-form converse bound on the one-shot error probability of state exclusion}
\label{app:closed}

Let $\rho\in\s{D}_A$ be a state, and let $\sigma\in\s{PSD}_A$ be a positive semidefinite operator.  For $\alpha\in(0,1)\cup(1,\infty)$, the \emph{Petz--Rényi pseudo-divergence} is defined as~\cite{petz1985QuasientropiesStatesNeumann,petz1986QuasientropiesFiniteQuantum}
\begin{align}
	\label{eq:petz}
	D_\alpha\fleft(\rho\middle\|\sigma\fright)&\coloneq\begin{cases}
		\frac{1}{\alpha-1}\ln\tr\fleft[\rho^\alpha\sigma^{1-\alpha}\fright]&\text{if }\alpha\in(0,1)\text{ or }\rho^0\leq\sigma^0, \\
		\infty&\text{otherwise}.
	\end{cases}
\end{align}
The Petz--Rényi pseudo-divergence is additive, regular in its second argument, and monotonically nondecreasing in $\alpha$~\cite{tomamichel2009FullyQuantumAsymptotic,khatri2024PrinciplesQuantumCommunication}.  It obeys the DPI for $\alpha\in(0,1)\cup(1,2]$~\cite{tomamichel2009FullyQuantumAsymptotic}, and thus we also call it the \emph{Petz--Rényi divergence} for this range.  The limit of the Petz--Rényi divergence as $\alpha\to1$ is given by the Umegaki divergence~\cite{petz1986QuasientropiesFiniteQuantum}.  The Petz--Rényi pseudo-divergence provides an upper bound on the sandwiched Rényi divergence~\cite{wilde2014StrongConverseClassical}: for all $\alpha\in(1,\infty)$,
\begin{align}
	\label{eq:order-petz}
	\sw{D}_\alpha\fleft(\rho\middle\|\sigma\fright)&\leq D_\alpha\fleft(\rho\middle\|\sigma\fright)\quad\forall\rho\in\s{D}_A,\;\sigma\in\s{PSD}_A.
\end{align}
For a bipartite positive semidefinite operator $\rho\in\s{PSD}_{AB}$ and a state $\sigma\in\s{D}_A$, the \emph{Petz--Rényi lautum information}\footnote{The word ``lautum'' is ``mutual'' spelled backwards~\cite{palomar2008LautumInformation}.  When the system $A$ is classical, the quantity is also known as the ``oveloH'' (``Holevo'' spelled backwards) information~\cite{nuradha2025MultivariateFidelities}.} is defined as
\begin{align}
	\label{eq:lautum}
	L_\alpha\fleft(A;B\fright)_{\rho|\sigma}&\coloneq\inf_{\tau\in\s{D}_B}D_\alpha\fleft(\sigma_A\otimes\tau_B\middle\|\rho_{AB}\fright).
\end{align}

\begin{lemma}[Quantum Sibson identity for the Petz--Rényi lautum information]
\label{lem:lautum}
Let $\rho\in\s{PSD}_{AB}$ be a positive semidefinite operator, and let $\sigma\in\s{D}_A$ be a state.  Let $\Pi\in\s{PSD}_B$ be the projector onto the minimal subspace of $\spa{H}_B$ that contains $\bigcup_{\tau\in\s{D}_B}\{\supp(\tau)\colon\sigma_A^0\otimes\tau_B^0\leq\rho_{AB}^0\}$, or equivalently, onto $\{\ket{\psi}\in\spa{H}_B\colon\ket{\phi}_A\otimes\ket{\psi}_B\in\supp(\rho_{AB})\;\forall\ket{\phi}\in\supp(\sigma)\}$.  Then for all $\alpha\in(0,1)\cup(1,\infty)$,
\begin{align}
	&L_\alpha\fleft(A;B\fright)_{\rho|\sigma} \notag\\
	&=D_\alpha\fleft(\sigma_A\otimes\tau_{\star,B}\middle\|\rho_{AB}\fright) \label{pf:lautum-1}\\
	&=\begin{cases}
		-\ln\tr\fleft[\left(\tr_A\fleft[\sigma_A^\alpha\rho_{AB}^{1-\alpha}\fright]\fright)^\frac{1}{1-\alpha}\right]&\text{if }\alpha\in(0,1), \\
		-\ln\tr\fleft[\left(\Pi_B\tr_A\fleft[\sigma_A^\alpha\rho_{AB}^{1-\alpha}\fright]\Pi_B\fright)^\frac{1}{1-\alpha}\right]&\text{if }\alpha\in(1,\infty),
	\end{cases} \label{pf:lautum-2}\\
	&=\lim_{\varepsilon\searrow0}-\ln\tr\fleft[\left(\tr_A\fleft[\sigma_A^\alpha\left(\rho_{AB}+\varepsilon\1_{AB}\right)^{1-\alpha}\fright]\right)^\frac{1}{1-\alpha}\fright], \label{pf:lautum-3} \\
\end{align}
where
\begin{align}
	\tau_\star&\coloneq\begin{cases}
		\frac{\left(\tr_A\fleft[\sigma_A^\alpha\rho_{AB}^{1-\alpha}\fright]\right)^\frac{1}{1-\alpha}}{\tr\fleft[\left(\tr_A\fleft[\sigma_A^\alpha\rho_{AB}^{1-\alpha}\fright]\right)^\frac{1}{1-\alpha}\fright]}&\text{if }\alpha\in(0,1), \\
		\frac{\left(\Pi_B\tr_A\fleft[\sigma_A^\alpha\rho_{AB}^{1-\alpha}\fright]\Pi_B\right)^\frac{1}{1-\alpha}}{\tr\fleft[\left(\Pi_B\tr_A\fleft[\sigma_A^\alpha\rho_{AB}^{1-\alpha}\fright]\Pi_B\right)^\frac{1}{1-\alpha}\fright]}&\text{if }\alpha\in(1,\infty)\text{ and }\Pi\neq0, \\
		\tau'&\text{if }\alpha\in(1,\infty)\text{ and }\Pi=0
	\end{cases} \label{pf:lautum-4}
\end{align}
is the unique minimizer to the infimization on the right-hand side of \eqref{eq:lautum} if $\alpha\in(0,1)$ or if $\alpha\in(1,\infty)$ and $\Pi\neq0$, and $\tau'\in\s{D}_B$ can be an arbitrary state.
\end{lemma}

\begin{proof}[Proof of \eqref{pf:lautum-1} and \eqref{pf:lautum-2} for $\alpha\in(0,1)$]
The statement follows from the observation in Ref.~\cite[Eq.~(3.10)]{hayashi2016CorrelationDetectionOperational} (also see Ref.~\cite[Lemma~3 of Supplemental Material]{sharma2013FundamentalBoundReliability}) and the fact that, for all $\alpha\in(0,1)$,
\begin{align}
	\left(1-\alpha\right)D_\alpha\fleft(\sigma\middle\|\rho\fright)&=\alpha D_{1-\alpha}\fleft(\rho\middle\|\sigma\fright)\quad\forall\rho,\sigma\in\s{D}_A.
\end{align}
We provide an explicit proof of the statement as follows.  For every state $\tau\in\s{D}_B$, it follows from the definition of the Petz--Rényi divergence (see \eqref{eq:petz}) that
\begin{align}
	&D_\alpha\fleft(\sigma_A\otimes\tau_B\middle\|\rho_{AB}\fright) \notag\\
	&=\frac{1}{\alpha-1}\ln\tr\fleft[\left(\sigma_A^\alpha\otimes\tau_B^\alpha\right)\rho_{AB}^{1-\alpha}\fright] \\
	&=\frac{1}{\alpha-1}\ln\tr\fleft[\tau_B^\alpha\tr_A\fleft[\sigma_A^\alpha\rho_{AB}^{1-\alpha}\fright]\fright] \\
	&=\frac{1}{\alpha-1}\ln\tr\fleft[\tau_B^\alpha\left(\tau_{\star,B}\tr\fleft[\left(\tr_A\fleft[\sigma_A^\alpha\rho_{AB}^{1-\alpha}\fright]\right)^\frac{1}{1-\alpha}\fright]\right)^{1-\alpha}\fright] \label{pf:lautum-5}\\
	&=\frac{1}{\alpha-1}\ln\tr\fleft[\tau^\alpha\tau_\star^{1-\alpha}\fright]-\ln\tr\fleft[\left(\tr_A\fleft[\sigma_A^\alpha\rho_{AB}^{1-\alpha}\fright]\right)^\frac{1}{1-\alpha}\fright] \\
	&=D_\alpha\fleft(\tau\middle\|\tau_\star\fright)-\ln\tr\fleft[\left(\tr_A\fleft[\sigma_A^\alpha\rho_{AB}^{1-\alpha}\fright]\right)^\frac{1}{1-\alpha}\fright]. \label{pf:lautum-6}
\end{align}
Here \eqref{pf:lautum-5} follows from \eqref{pf:lautum-4}.  Then it follows from the definition of the Petz--Rényi lautum information (see \eqref{eq:lautum}) that
\begin{align}
	L_\alpha\fleft(A;B\fright)_{\rho|\sigma}&=\inf_{\tau\in\s{D}_B}D_\alpha\fleft(\sigma_A\otimes\tau_B\middle\|\rho_{AB}\fright) \\
	&=\inf_{\tau\in\s{D}_B}D_\alpha\fleft(\tau\middle\|\tau_\star\fright)-\ln\tr\fleft[\left(\tr_A\fleft[\sigma_A^\alpha\rho_{AB}^{1-\alpha}\fright]\right)^\frac{1}{1-\alpha}\fright] \label{pf:lautum-7}\\
	&=-\ln\tr\fleft[\left(\tr_A\fleft[\sigma_A^\alpha\rho_{AB}^{1-\alpha}\fright]\right)^\frac{1}{1-\alpha}\fright] \label{pf:lautum-8}\\
	&=D_\alpha\fleft(\sigma_A\otimes\tau_{\star,B}\middle\|\rho_{AB}\fright). \label{pf:lautum-9}
\end{align}
Here \eqref{pf:lautum-7} follows from \eqref{pf:lautum-6}; Eq.~\eqref{pf:lautum-8} follows from the facts that $D_\alpha(\tau\|\tau_\star)\geq D_\alpha(\tau_\star\|\tau_\star)=0$ for all $\alpha\in(0,1)$ and all $\tau\in\s{D}_B$; Eq.~\eqref{pf:lautum-9} follows from \eqref{pf:lautum-6} and the aforementioned fact.  The uniqueness of the minimizer $\tau_\star$ to the infimization on the right-hand side of \eqref{eq:lautum} follows from \eqref{pf:lautum-6} and the fact that $D_\alpha(\tau\|\tau_\star)=0$ if and only if $\tau=\tau_\star$.
\end{proof}

\begin{proof}[Proof of \eqref{pf:lautum-1} and \eqref{pf:lautum-2} for $\alpha\in{(1,\infty)}$]
If $\Pi=0$, then both sides of \eqref{pf:lautum-1} and \eqref{pf:lautum-2} are equal to $\infty$.  Henceforth we assume that $\Pi\neq0$.  Let $\tau\in\s{D}_B$ be a state such that $\sigma_A^0\otimes\tau_B^0\leq\rho_{AB}^0$.  This implies that $(\sigma_A^{1-\alpha}\otimes\tau_B^{1-\alpha})^0\leq(\rho_{AB}^{1-\alpha})^0$, which entails that
\begin{align}
	\tau^0&=\left(\tr_A\fleft[\sigma_A^\alpha\left(\sigma_A^{1-\alpha}\otimes\tau_B^{1-\alpha}\right)\fright]\right)^0 \\
	&\leq\left(\tr_A\fleft[\sigma_A^\alpha\rho_{AB}^{1-\alpha}\fright]\right)^0. \label{pf:lautum-10}
\end{align}
Then we have that
\begin{align}
	\tau^0&=\left(\Pi\tau\Pi\right)^0 \label{pf:lautum-11}\\
	&\leq\left(\Pi_B\tr_A\fleft[\sigma_A^\alpha\rho_{AB}^{1-\alpha}\fright]\Pi_B\right)^0 \label{pf:lautum-12}\\
	&=\tau_\star^0. \label{pf:lautum-13}
\end{align}
Here \eqref{pf:lautum-11} follows from the fact that $\tau=\Pi\tau\Pi$; Eq.~\eqref{pf:lautum-12} follows from \eqref{pf:lautum-10}; Eq.~\eqref{pf:lautum-13} follows from \eqref{pf:lautum-4}.  It follows from the definition of the Petz--Rényi pseudo-divergence (see \eqref{eq:petz}) and the condition $\sigma_A^0\otimes\tau_B^0\leq\rho_{AB}^0$ that
\begin{align}
	&D_\alpha\fleft(\sigma_A\otimes\tau_B\middle\|\rho_{AB}\fright) \notag\\
	&=\frac{1}{\alpha-1}\ln\tr\fleft[\left(\sigma_A^\alpha\otimes\tau_B^\alpha\right)\rho_{AB}^{1-\alpha}\fright] \\
	&=\frac{1}{\alpha-1}\ln\tr\fleft[\tau_B^\alpha\tr_A\fleft[\sigma_A^\alpha\rho_{AB}^{1-\alpha}\fright]\fright] \\
	&=\frac{1}{\alpha-1}\ln\tr\fleft[\tau_B^\alpha\Pi_B\tr_A\fleft[\sigma_A^\alpha\rho_{AB}^{1-\alpha}\fright]\Pi_B\fright] \label{pf:lautum-14}\\
	&=\frac{1}{\alpha-1}\ln\tr\fleft[\tau_B^\alpha\left(\tau_{\star,B}\tr\fleft[\left(\Pi_B\tr_A\fleft[\sigma_A^\alpha\rho_{AB}^{1-\alpha}\fright]\Pi_B\right)^\frac{1}{1-\alpha}\fright]\right)^{1-\alpha}\fright] \label{pf:lautum-15}\\
	&=\frac{1}{\alpha-1}\ln\tr\fleft[\tau^\alpha\tau_\star^{1-\alpha}\fright]-\ln\tr\fleft[\left(\Pi_B\tr_A\fleft[\sigma_A^\alpha\rho_{AB}^{1-\alpha}\fright]\Pi_B\right)^\frac{1}{1-\alpha}\fright] \\
	&=D_\alpha\fleft(\tau\middle\|\tau_\star\fright)-\ln\tr\fleft[\left(\Pi_B\tr_A\fleft[\sigma_A^\alpha\rho_{AB}^{1-\alpha}\fright]\Pi_B\right)^\frac{1}{1-\alpha}\fright]. \label{pf:lautum-16}
\end{align}
Here \eqref{pf:lautum-14} follows from the fact that $\tau=\Pi\tau\Pi$; Eq.~\eqref{pf:lautum-15} follows from \eqref{pf:lautum-4}; Eq.~\eqref{pf:lautum-16} follows from the definition of the Petz--Rényi pseudo-divergence (see \eqref{eq:petz}) and \eqref{pf:lautum-13}.  It follows from the definition of the Petz--Rényi lautum information (see \eqref{eq:lautum}) that
\begin{align}
	&L_\alpha\fleft(A;B\fright)_{\rho|\sigma} \notag\\
	&=\inf_{\tau\in\s{D}_B}\left\{D_\alpha\fleft(\sigma_A\otimes\tau_B\middle\|\rho_{AB}\fright)\colon\sigma_A^0\otimes\tau_B^0\leq\rho_{AB}^0\right\} \\
	&=\inf_{\tau\in\s{D}_B}\left\{D_\alpha\fleft(\tau\middle\|\tau_\star\fright)-\ln\tr\fleft[\left(\Pi_B\tr_A\fleft[\sigma_A^\alpha\rho_{AB}^{1-\alpha}\fright]\Pi_B\right)^\frac{1}{1-\alpha}\fright]\colon\right. \notag\\
	&\qquad\left.\vphantom{D_\alpha\fleft(\tau\middle\|\tau_\star\fright)-\ln\tr\fleft[\left(\Pi_B\tr_A\fleft[\sigma_A^\alpha\rho_{AB}^{1-\alpha}\fright]\Pi_B\right)^\frac{1}{1-\alpha}\fright]\colon}\sigma_A^0\otimes\tau_B^0\leq\rho_{AB}^0\right\} \label{pf:lautum-17}\\
	&=-\ln\tr\fleft[\left(\Pi_B\tr_A\fleft[\sigma_A^\alpha\rho_{AB}^{1-\alpha}\fright]\Pi_B\right)^\frac{1}{1-\alpha}\fright] \label{pf:lautum-18}\\
	&=D_\alpha\fleft(\sigma_A\otimes\tau_{\star,B}\middle\|\rho_{AB}\fright). \label{pf:lautum-19}
\end{align}
Here \eqref{pf:lautum-17} follows from \eqref{pf:lautum-16}; Eq.~\eqref{pf:lautum-18} follows from the facts that $\sigma_A^0\otimes\tau_{\star,B}^0\leq\rho_{AB}^0$ and that $D_\alpha(\tau\|\tau_\star)\geq D_\alpha(\tau_\star\|\tau_\star)=0$ for all $\alpha\in(1,\infty)$ and all $\tau\in\s{D}_B$; Eq.~\eqref{pf:lautum-19} follows from \eqref{pf:lautum-16} and the aforementioned facts.  The uniqueness of the minimizer $\tau_\star$ to the infimization on the right-hand side of \eqref{eq:lautum} follows from \eqref{pf:lautum-16} and the facts that $\sigma_A^0\otimes\tau_{\star,B}^0\leq\rho_{AB}^0$ and that $D_\alpha(\tau\|\tau_\star)=0$ if and only if $\tau=\tau_\star$.
\end{proof}

\begin{proof}[Proof of \eqref{pf:lautum-3}]
For all $\alpha\in(0,1)\cup(1,\infty)$, it follows from the definition of the Petz--Rényi lautum information (see \eqref{eq:lautum}) that
\begin{align}
	&L_\alpha\fleft(A;B\fright)_{\rho|\sigma} \notag\\
	&=\inf_{\tau\in\s{D}_B}D_\alpha\fleft(\sigma_A\otimes\tau_B\middle\|\rho_{AB}\fright) \\
	&=\inf_{\tau\in\s{D}_B}\lim_{\varepsilon\searrow0}D_\alpha\fleft(\sigma_A\otimes\tau_B\middle\|\rho_{AB}+\varepsilon\1_{AB}\fright) \label{pf:lautum-20}\\
	&=\inf_{\tau\in\s{D}_B}\sup_{\varepsilon\in(0,1)}D_\alpha\fleft(\sigma_A\otimes\tau_B\middle\|\rho_{AB}+\varepsilon\1_{AB}\fright) \label{pf:lautum-21}\\
	&=\sup_{\varepsilon\in(0,1)}\inf_{\tau\in\s{D}_B}D_\alpha\fleft(\sigma_A\otimes\tau_B\middle\|\rho_{AB}+\varepsilon\1_{AB}\fright) \label{pf:lautum-22}\\
	&=\lim_{\varepsilon\searrow0}\inf_{\tau\in\s{D}_B}D_\alpha\fleft(\sigma_A\otimes\tau_B\middle\|\rho_{AB}+\varepsilon\1_{AB}\fright) \label{pf:lautum-23}\\
	&=\lim_{\varepsilon\searrow0}L_\alpha\fleft(A;B\fright)_{\rho+\varepsilon\1|\sigma} \label{pf:lautum-24}\\
	&=\lim_{\varepsilon\searrow0}-\ln\tr\fleft[\left(\tr_A\fleft[\sigma_A^\alpha\left(\rho_{AB}+\varepsilon\1_{AB}\right)^{1-\alpha}\fright]\right)^\frac{1}{1-\alpha}\fright]. \label{pf:lautum-25}
\end{align}
Here \eqref{pf:lautum-20} uses the regularity of the Petz--Rényi pseudo-divergence in its second argument; Eqs.~\eqref{pf:lautum-21} and \eqref{pf:lautum-23} follow from the observation that the function $\varepsilon\mapsto D_\alpha(\sigma_A\otimes\tau_B\|\rho_{AB}+\varepsilon\1_{AB})$ is monotonically nonincreasing; Eq.~\eqref{pf:lautum-22} applies the Mosonyi--Hiai minimax theorem~\cite[Corollary~A.2]{mosonyi2011QuantumRenyiRelative}; Eq.~\eqref{pf:lautum-24} follows from the definition of the Petz--Rényi lautum information (see \eqref{eq:lautum}); Eq.~\eqref{pf:lautum-25} follows from \eqref{pf:lautum-2}.  The application of the Mosonyi--Hiai minimax theorem in \eqref{pf:lautum-22} is justified by the following observations: the feasible region of $\tau$ is compact; the objective function is monotonically nonincreasing in $\varepsilon$, and it is lower semicontinuous in $\tau$ due to the lower semicontinuity of the Petz--Rényi divergence in its first argument.
\end{proof}

\begin{proposition}[A closed-form converse bound on the error probability]
\label{prop:closed}
Let $\en{E}\equiv(p_{[r]},\rho_{[r]})$ be an ensemble of states with $p_{[r]}\in\itr(\s{P}_r)$ and $\rho_{[r]}\in\s{D}_A^{[r]}$.  Then for all $\alpha\in(1,\infty)$,
\begin{align}
	-\ln P_\abb{err}\fleft(\en{E}\fright)&\leq\lim_{\varepsilon\searrow0}-\ln\tr\fleft[\left(\sum_{x\in[r]}\left(p_x\rho_x+\varepsilon\1\right)^{1-\alpha}\right)^\frac{1}{1-\alpha}\fright] \label{pf:closed-1}\\
	&=-\ln\tr\fleft[\left(\sum_{x\in[r]}\Pi\left(p_x\rho_x\right)^{1-\alpha}\Pi\right)^\frac{1}{1-\alpha}\fright], \label{pf:closed-2}
\end{align}
where $\Pi=\bigwedge_{x\in[r]}\rho_x^0$.
\end{proposition}

\begin{proof}
It follows from Proposition~\ref{prop:oneshot-hypothesis} and the notation of \eqref{eq:cq} that
\begin{align}
	-\ln P_\abb{err}\fleft(\en{E}\fright)&=\inf_{\tau\in\aff\fleft(\s{D}_A\fright)}D_\abb{H}^{1-\frac{1}{r}}\fleft(\pi_X\otimes\tau_A\middle\|\cq{\rho}_{XA}\fright) \\
	&\leq\inf_{\tau\in\s{D}_A}D_\abb{H}^{1-\frac{1}{r}}\fleft(\pi_X\otimes\tau_A\middle\|\cq{\rho}_{XA}\fright) \\
	&\leq\inf_{\tau\in\s{D}_A}\sw{D}_\alpha\fleft(\pi_X\otimes\tau_A\middle\|\cq{\rho}_{XA}\fright)+\frac{\alpha}{\alpha-1}\ln r \label{pf:closed-3}\\
	&\leq\inf_{\tau\in\s{D}_A}D_\alpha\fleft(\pi_X\otimes\tau_A\middle\|\cq{\rho}_{XA}\fright)+\frac{\alpha}{\alpha-1}\ln r \label{pf:closed-4}\\
	&=L_\alpha\fleft(X;A\fright)_{\cq{\rho}|\pi}+\frac{\alpha}{\alpha-1}\ln r. \label{pf:closed-5}
\end{align}
Here \eqref{pf:closed-3} follows from Lemma~\ref{lem:hypothesis-extended}; Eq.~\eqref{pf:closed-4} follows from the order between the sandwiched Rényi divergence and the Petz--Rényi pseudo-divergence (see \eqref{eq:order-petz}); Eq.~\eqref{pf:closed-5} follows from the definition of the Petz--Rényi lautum (oveloH) information (see \eqref{eq:lautum}).  Applying \eqref{pf:lautum-3} of Lemma~\ref{lem:lautum} to the right-hand side of \eqref{pf:closed-5}, we have that
\begin{align}
	&L_\alpha\fleft(X;A\fright)_{\cq{\rho}|\pi}+\frac{\alpha}{\alpha-1}\ln r \notag\\
	&=\lim_{\varepsilon\searrow0}-\ln\tr\fleft[\left(\tr_X\fleft[\pi_X^\alpha\left(\cq{\rho}_{XA}+\varepsilon\1_{XA}\right)^{1-\alpha}\fright]\right)^\frac{1}{1-\alpha}\fright]+\frac{\alpha}{\alpha-1}\ln r \\
	&=\lim_{\varepsilon\searrow0}-\ln\tr\fleft[\left(\tr_X\fleft[\sum_{x\in[r]}\frac{1}{r^\alpha}\op{x}{x}_X\right.\right.\vphantom{\left(\tr_X\fleft[\sum_{x\in[r]}\frac{1}{r^\alpha}\op{x}{x}_X\sum_{x\in[r]}\frac{1}{r^\alpha}\op{x}{x}_X\otimes\left(p_x\rho_{x,A}+\varepsilon\1_A\right)^{1-\alpha}\fright]\right)^\frac{1}{1-\alpha}}\right. \notag\\
	&\qquad\left.\left.\left.\vphantom{\sum_{x\in[r]}\frac{1}{r^\alpha}\op{x}{x}_X}\otimes\left(p_x\rho_{x,A}+\varepsilon\1_A\right)^{1-\alpha}\fright]\right)^\frac{1}{1-\alpha}\fright]+\frac{\alpha}{\alpha-1}\ln r \\
	&=\lim_{\varepsilon\searrow0}-\ln\tr\fleft[\left(\sum_{x\in[r]}\left(p_x\rho_x+\varepsilon\1\right)^{1-\alpha}\right)^\frac{1}{1-\alpha}\fright]. \label{pf:closed-6}
\end{align}
Combining \eqref{pf:closed-5} and \eqref{pf:closed-6} leads to \eqref{pf:closed-1}.  Applying \eqref{pf:lautum-2} of Lemma~\ref{lem:lautum} to the right-hand side of \eqref{pf:closed-5}, we have that
\begin{align}
	&L_\alpha\fleft(X;A\fright)_{\cq{\rho}|\pi}+\frac{\alpha}{\alpha-1}\ln r \notag\\
	&=-\ln\tr\fleft[\left(\Pi_A\tr_X\fleft[\pi_X^\alpha\cq{\rho}_{XA}^{1-\alpha}\fright]\Pi_A\right)^\frac{1}{1-\alpha}\fright]+\frac{\alpha}{\alpha-1}\ln r \\
	&=-\ln\tr\fleft[\left(\Pi_A\tr_X\fleft[\sum_{x\in[r]}\frac{1}{r^\alpha}\op{x}{x}_X\otimes\left(p_x\rho_{x,A}\right)^{1-\alpha}\fright]\Pi_A\right)^\frac{1}{1-\alpha}\fright] \notag\\
	&\qquad+\frac{\alpha}{\alpha-1}\ln r \\
	&=-\ln\tr\fleft[\left(\sum_{x\in[r]}\Pi\left(p_x\rho_x\right)^{1-\alpha}\Pi\right)^\frac{1}{1-\alpha}\fright]. \label{pf:closed-7}
\end{align}
Combining \eqref{pf:closed-6} and \eqref{pf:closed-7} leads to \eqref{pf:closed-2}.
\end{proof}

\section{Barycentric converse bounds on the nonasymptotic error probabilities of hypothesis exclusion}
\label{app:nonasymptotic}

\subsection{State exclusion}
\label{app:nonasymptotic-state}

\begin{lemma}[Connection between $\sw{D}_\alpha$ and $D$~{\cite[Lemma~8]{tomamichel2009FullyQuantumAsymptotic}}]
	\label{lem:sandwiched}
	Let $\rho\in\s{D}_A$ be a state, and let $\sigma\in\s{PSD}_A$ be a positive semidefinite operator.  Then for all $\alpha\in(1,1+\frac{\ln3}{4\ln\Upsilon(\rho\|\sigma)}]$,
	\begin{align}
		\sw{D}_\alpha\fleft(\rho\middle\|\sigma\fright)&\leq D_\alpha\fleft(\rho\middle\|\sigma\fright)\leq D\fleft(\rho\middle\|\sigma\fright)+4\left(\alpha-1\right)\left(\ln\Upsilon\fleft(\rho\middle\|\sigma\fright)\right)^2,
	\end{align}
	where
	\begin{align}
		\Upsilon\fleft(\rho\middle\|\sigma\fright)&\coloneq1+\exp\left(-\frac{1}{2}D_\frac{1}{2}\fleft(\rho\middle\|\sigma\fright)\right)+\exp\left(\frac{1}{2}D_\frac{3}{2}\fleft(\rho\middle\|\sigma\fright)\right)\geq3. \label{pf:sandwiched-1}
	\end{align}
\end{lemma}

\begin{proposition}[Log-Euclidean converse bound on the nonasymptotic error probability]
\label{prop:nonasymptotic-euclidean}
Let $\en{E}\equiv(p_{[r]},\rho_{[r]})$ be an ensemble of states with $p_{[r]}\in\itr(\s{P}_r)$ and $\rho_{[r]}\in\s{D}_A^{[r]}$.  Then for every positive integer $n$,
\begin{align}
	&-\frac{1}{n}\ln P_\abb{err}\fleft(\en{E}^n\fright) \notag\\
	&\leq C^\flat\fleft(\rho_{[r]}\fright)+\frac{1}{\sqrt{n}}\left(\ln\Upsilon_{\max}\fleft(\rho_{[r]}\fright)\right)\left(\ln3+\frac{5}{\ln3}\ln\left(\frac{1}{p_{\min}}\right)\right), \label{pf:nonasymptotic-euclidean-1}
\end{align}
where
\begin{align}
	\Upsilon_{\max}\fleft(\rho_{[r]}\fright)&\coloneq1+\max_{x\in[r]}\left(\tr\fleft[\rho_x^\frac{1}{2}\fright]+\tr\fleft[\rho_x^{-\frac{1}{2}}\fright]\right), \label{pf:nonasymptotic-euclidean-2}
\end{align}
and $p_{\min}\equiv\min_{x\in[r]}p_x$.
\end{proposition}

\begin{proof}
If $\bigwedge_{x\in[r]}\rho_x^0=0$, then the right-hand side of \eqref{pf:nonasymptotic-euclidean-1} is equal to $\infty$.  Henceforth we assume that $\bigwedge_{x\in[r]}\rho_x^0\neq0$.  For every state $\tau\in\s{D}_A$ such that $\tau^0\leq\bigwedge_{x\in[r]}\rho_x^0$ and all $x\in[r]$, it follows from \eqref{pf:sandwiched-1} and the definition of the Petz--Rényi divergence (see \eqref{eq:petz}) that
\begin{align}
	\Upsilon\fleft(\tau\middle\|\rho_x\fright)&=1+\tr\fleft[\tau^\frac{1}{2}\rho_x^\frac{1}{2}\fright]+\tr\fleft[\tau^\frac{3}{2}\rho_x^{-\frac{1}{2}}\fright] \\
	&\leq1+\max_{x\in[r]}\left(\tr\fleft[\tau^\frac{1}{2}\rho_x^\frac{1}{2}\fright]+\tr\fleft[\tau^\frac{3}{2}\rho_x^{-\frac{1}{2}}\fright]\right) \\
	&\leq1+\max_{x\in[r]}\left(\tr\fleft[\rho_x^\frac{1}{2}\fright]+\tr\fleft[\rho_x^{-\frac{1}{2}}\fright]\right) \label{pf:nonasymptotic-euclidean-3}\\
	&=\Upsilon_{\max}\fleft(\rho_{[r]}\fright). \label{pf:nonasymptotic-euclidean-4}
\end{align}
Here \eqref{pf:nonasymptotic-euclidean-3} follows from the fact that $\tau^\alpha\leq\1$ for all $\alpha\in[0,\infty)$; Eq.~\eqref{pf:nonasymptotic-euclidean-4} follows from \eqref{pf:nonasymptotic-euclidean-2}.  Then for every positive integer $n$ and all $\alpha\in(1,1+\frac{\ln3}{4\ln\Upsilon_{\max}(\rho_{[r]})}]$, it follows from \eqref{pf:asymptotic-euclidean-3} that
\begin{align}
	&-\frac{1}{n}\ln P_\abb{err}\fleft(\en{E}^n\fright) \notag\\
	&\leq\sup_{s_{[r]}\in\s{P}_r}\inf_{\tau\in\aff\fleft(\s{D}_A\fright)}\sum_{x\in[r]}s_x\sw{D}_\alpha\fleft(\tau\middle\|\rho_x\fright)+\frac{\alpha}{n\left(\alpha-1\right)}\ln\left(\frac{1}{p_{\min}}\right) \\
	&\leq\sup_{s_{[r]}\in\s{P}_r}\inf_{\tau\in\s{D}_A}\sum_{x\in[r]}s_x\sw{D}_\alpha\fleft(\tau\middle\|\rho_x\fright)+\frac{\alpha}{n\left(\alpha-1\right)}\ln\left(\frac{1}{p_{\min}}\right) \label{pf:nonasymptotic-euclidean-5}\\
	&=\inf_{\tau\in\s{D}_A}\sup_{s_{[r]}\in\s{P}_r}\sum_{x\in[r]}s_x\sw{D}_\alpha\fleft(\tau\middle\|\rho_x\fright)+\frac{\alpha}{n\left(\alpha-1\right)}\ln\left(\frac{1}{p_{\min}}\right) \label{pf:nonasymptotic-euclidean-6}\\
	&=\inf_{\tau\in\s{D}_A'}\sup_{s_{[r]}\in\s{P}_r}\sum_{x\in[r]}s_x\sw{D}_\alpha\fleft(\tau\middle\|\rho_x\fright)+\frac{\alpha}{n\left(\alpha-1\right)}\ln\left(\frac{1}{p_{\min}}\right) \label{pf:nonasymptotic-euclidean-7}\\
	&\leq\inf_{\tau\in\s{D}_A'}\sup_{s_{[r]}\in\s{P}_r}\sum_{x\in[r]}s_x\left(D\fleft(\tau\middle\|\rho_x\fright)+4\left(\alpha-1\right)\left(\ln\Upsilon\fleft(\tau\middle\|\rho_x\fright)\right)^2\right) \notag\\
	&\qquad+\frac{\alpha}{n\left(\alpha-1\right)}\ln\left(\frac{1}{p_{\min}}\right) \label{pf:nonasymptotic-euclidean-8}\\
	&\leq\inf_{\tau\in\s{D}_A'}\sup_{s_{[r]}\in\s{P}_r}\sum_{x\in[r]}s_xD\fleft(\tau\middle\|\rho_x\fright)+4\left(\alpha-1\right)\left(\ln\Upsilon_{\max}\fleft(\rho_{[r]}\fright)\right)^2 \notag\\
	&\qquad+\frac{\alpha}{n\left(\alpha-1\right)}\ln\left(\frac{1}{p_{\min}}\right) \label{pf:nonasymptotic-euclidean-9}\\
	&=\inf_{\tau\in\s{D}_A}\sup_{s_{[r]}\in\s{P}_r}\sum_{x\in[r]}s_xD\fleft(\tau\middle\|\rho_x\fright)+4\left(\alpha-1\right)\left(\ln\Upsilon_{\max}\fleft(\rho_{[r]}\fright)\right)^2 \notag\\
	&\qquad+\frac{\alpha}{n\left(\alpha-1\right)}\ln\left(\frac{1}{p_{\min}}\right) \\
	&=C^\flat\fleft(\rho_{[r]}\fright)+4\left(\alpha-1\right)\left(\ln\Upsilon_{\max}\fleft(\rho_{[r]}\fright)\right)^2 \notag\\
	&\qquad+\frac{\alpha}{n\left(\alpha-1\right)}\ln\left(\frac{1}{p_{\min}}\right), \label{pf:nonasymptotic-euclidean-10}
\end{align}
where in \eqref{pf:nonasymptotic-euclidean-7}--\eqref{pf:nonasymptotic-euclidean-9} we denote $\s{D}_A'\equiv\{\tau\in\s{D}_A\colon\tau^0\leq\bigwedge_{x\in[r]}\rho_x^0\}$.  Here \eqref{pf:nonasymptotic-euclidean-6} follows from applying Lemma~\ref{lem:radius-alternative} to the left sandwiched Rényi radius and \eqref{eq:radius-2}; Eq.~\eqref{pf:nonasymptotic-euclidean-8} follows from Lemma~\ref{lem:sandwiched} and \eqref{pf:nonasymptotic-euclidean-4}; Eq.~\eqref{pf:nonasymptotic-euclidean-9} follows from \eqref{pf:nonasymptotic-euclidean-4}; Eq.~\eqref{pf:nonasymptotic-euclidean-10} follows from the equality between the left Umegaki radius and the multivariate log-Euclidean Chernoff divergence (see \eqref{eq:radius-2} and \eqref{eq:euclidean-radius}).  Inserting $\alpha\coloneq1+\frac{\ln3}{4\ln\Upsilon_{\max}(\rho_{[r]})\sqrt{n}}\in(1,1+\frac{\ln3}{4\ln\Upsilon_{\max}(\rho_{[r]})}]$ to \eqref{pf:nonasymptotic-euclidean-10}, we obtain that
\begin{align}
	&-\frac{1}{n}\ln P_\abb{err}\fleft(\en{E}^n\fright) \notag\\
	&\leq C^\flat\fleft(\rho_{[r]}\fright)+4\left(\alpha-1\right)\left(\ln\Upsilon_{\max}\fleft(\rho_{[r]}\fright)\right)^2+\frac{\alpha}{n\left(\alpha-1\right)}\ln\left(\frac{1}{p_{\min}}\right) \\
	&=C^\flat\fleft(\rho_{[r]}\fright)+\frac{1}{\sqrt{n}}\left(\ln3\ln\Upsilon_{\max}\fleft(\rho_{[r]}\fright)\vphantom{\left(\frac{4\ln\Upsilon_{\max}\fleft(\rho_{[r]}\fright)}{\ln3}+\frac{1}{\sqrt{n}}\right)\ln\left(\frac{1}{p_{\min}}\right)}\right. \notag\\
	&\qquad\left.\vphantom{}+\left(\frac{4\ln\Upsilon_{\max}\fleft(\rho_{[r]}\fright)}{\ln3}+\frac{1}{\sqrt{n}}\right)\ln\left(\frac{1}{p_{\min}}\right)\right) \\
	&\leq C^\flat\fleft(\rho_{[r]}\fright)+\frac{1}{\sqrt{n}}\ln\Upsilon_{\max}\fleft(\rho_{[r]}\fright)\left(\ln3+\frac{5}{\ln3}\ln\left(\frac{1}{p_{\min}}\right)\right). \label{pf:nonasymptotic-euclidean-11}
\end{align}
Here \eqref{pf:nonasymptotic-euclidean-11} follows from the fact that $\Upsilon_{\max}(\rho_{[r]})\geq3$ due to \eqref{pf:nonasymptotic-euclidean-4} and \eqref{pf:sandwiched-1}.
\end{proof}

\subsection{Channel exclusion}
\label{app:nonasymptotic-channel}

Recall that the geometric Rényi divergence for $\alpha\in(1,2]$ is defined in \eqref{eq:geometric}.  It can be defined for $\alpha\in(0,1)$ as well~\cite{matsumoto2018NewQuantumVersion,katariya2021GeometricDistinguishabilityMeasures}, and we still denote it by $\g{D}_\alpha(\rho\|\sigma)$ for this range.

\begin{lemma}[Connection between the channel $\g{D}_\alpha$ and $\g{D}$~{\cite[Lemma~34 and Eq.~(462)]{ding2023BoundingForwardClassical}}]
\label{lem:geometric}
Let $\ch{N}\in\s{C}_{A\to B}$ be a channel, and let $\ch{M}\in\s{CP}_{A\to B}$ be a completely positive map.  Then for all $\alpha\in(1,1+\frac{\ln3}{4\ln\g{\Upsilon}(\ch{N}\|\ch{M})}]$,
\begin{align}
	\g{D}_\alpha\fleft(\ch{N}\middle\|\ch{M}\fright)&\leq \g{D}\fleft(\ch{N}\middle\|\ch{M}\fright)+4\left(\alpha-1\right)\left(\ln\g{\Upsilon}\fleft(\ch{N}\middle\|\ch{M}\fright)\right)^2,
\end{align}
where
\begin{align}
	&\g{\Upsilon}\fleft(\ch{N}\middle\|\ch{M}\fright) \notag\\
	&\coloneq\sup_{\rho_{RA}}\left(1+\exp\left(-\frac{1}{2}\g{D}_\frac{1}{2}\fleft(\ch{N}_{A\to B}\fleft[\rho_{RA}\fright]\middle\|\ch{M}_{A\to B}\fleft[\rho_{RA}\fright]\fright)\right)\right. \notag\\
	&\qquad\left.\vphantom{}+\exp\left(\frac{1}{2}\g{D}_\frac{3}{2}\fleft(\ch{N}_{A\to B}\fleft[\rho_{RA}\fright]\middle\|\ch{M}_{A\to B}\fleft[\rho_{RA}\fright]\fright)\right)\right) \label{pf:geometric-1}\\
	&\geq3, \label{pf:geometric-2}
\end{align}
and $R$ is a system with $d_R$ allowed to be arbitrarily large.
\end{lemma}

\begin{proposition}[Barycentric converse bound on the nonasymptotic error probability for channels]
\label{prop:nonasymptotic-barycentric}
Let $\en{N}\equiv(p_{[r]},\ch{N}_{[r]})$ be an ensemble of channels with $p_{[r]}\in\itr(\s{P}_r)$ and $\ch{N}_{[r]}\in\s{C}_{A\to B}^{[r]}$.  Then for every positive integer $n$,
\begin{align}
	&-\frac{1}{n}\ln P_\abb{err}\fleft(n;\en{N}\fright) \notag\\
	&\leq R^\g{D}\fleft(\ch{N}_{[r]}\fright)+\frac{1}{\sqrt{n}}\left(\ln\g{\Upsilon}_{\max}\fleft(\ch{N}_{[r]}\fright)\right)\left(\ln3+\frac{5}{\ln3}\ln\left(\frac{1}{p_{\min}}\right)\right), \label{pf:nonasymptotic-barycentric-1}
\end{align}
where
\begin{align}
	\g{\Upsilon}_{\max}\fleft(\ch{N}_{[r]}\fright)&\coloneq2+d_A^\frac{3}{2}\max_{x\in[r]}\left\lVert\tr_B\fleft[J_{\ch{N}_x}^{-\frac{1}{2}}\fright]\right\rVert_\infty, \label{pf:nonasymptotic-barycentric-2}
\end{align}
and $p_{\min}\equiv\min_{x\in[r]}p_x$.
\end{proposition}

\begin{proof}
If $\bigwedge_{x\in[r]}J_{\ch{N}_x}^0=0$, then the right-hand side of \eqref{pf:nonasymptotic-barycentric-1} is equal to $\infty$.  Henceforth we assume that $\bigwedge_{x\in[r]}J_{\ch{N}_x}^0\neq0$.  For every channel $\ch{T}\in\s{C}_{A\to B}$ such that $J_\ch{T}^0\leq\bigwedge_{x\in[r]}J_{\ch{N}_x}^0$ and all $x\in[r]$, it follows from \eqref{pf:geometric-1} that
\begin{align}
	&\g{\Upsilon}\fleft(\ch{T}\middle\|\ch{N}_x\fright) \notag\\
	&=\sup_{\rho_{RA}}\left(1+\exp\left(-\frac{1}{2}\g{D}_\frac{1}{2}\fleft(\ch{T}_{A\to B}\fleft[\rho_{RA}\fright]\middle\|\ch{N}_{x,A\to B}\fleft[\rho_{RA}\fright]\fright)\right)\right. \notag\\
	&\qquad\left.\vphantom{}+\exp\left(\frac{1}{2}\g{D}_\frac{3}{2}\fleft(\ch{T}_{A\to B}\fleft[\rho_{RA}\fright]\middle\|\ch{N}_{x,A\to B}\fleft[\rho_{RA}\fright]\fright)\right)\right) \\
	&\leq\sup_{\rho_{RA}}\left(2+\exp\left(\frac{1}{2}\g{D}_\frac{3}{2}\fleft(\ch{T}_{A\to B}\fleft[\rho_{RA}\fright]\middle\|\ch{N}_{x,A\to B}\fleft[\rho_{RA}\fright]\fright)\right)\right) \label{pf:nonasymptotic-barycentric-3}\\
	&=2+\exp\left(\frac{1}{2}\g{D}_\frac{3}{2}\fleft(\ch{T}\middle\|\ch{N}_x\fright)\right) \label{pf:nonasymptotic-barycentric-4}\\
	&=2+\left\lVert\tr_B\fleft[J_{\ch{N}_x}^\frac{1}{2}\left(J_{\ch{N}_x}^{-\frac{1}{2}}J_\ch{T}J_{\ch{N}_x}^{-\frac{1}{2}}\right)^\frac{3}{2}J_{\ch{N}_x}^\frac{1}{2}\fright]\right\rVert_\infty \label{pf:nonasymptotic-barycentric-5}\\
	&\leq2+d_A^\frac{3}{2}\max_{x\in[r]}\left\lVert\tr_B\fleft[J_{\ch{N}_x}^{-\frac{1}{2}}\fright]\right\rVert_\infty \label{pf:nonasymptotic-barycentric-6}\\
	&=\g{\Upsilon}_{\max}\fleft(\rho_{[r]}\fright). \label{pf:nonasymptotic-barycentric-7}
\end{align}
Here \eqref{pf:nonasymptotic-barycentric-3} follows from the fact that $\g{D}_\frac{1}{2}(\rho\|\sigma)\geq0$ for all $\rho,\sigma\in\s{D}_{RB}$; Eq.~\eqref{pf:nonasymptotic-barycentric-4} follows from the definition of the geometric Rényi channel divergence (see \eqref{eq:geometric-channel}); Eq.~\eqref{pf:nonasymptotic-barycentric-5} follows from the closed-form expression of the geometric Rényi channel divergence in terms of Choi operators (see \eqref{eq:geometric-channel-choi}); Eq.~\eqref{pf:nonasymptotic-barycentric-6} follows from the fact that $J_\ch{T}\leq d_A\1$; Eq.~\eqref{pf:nonasymptotic-barycentric-7} follows from \eqref{pf:nonasymptotic-barycentric-2}.  Then for every positive integer $n$ and all $\alpha\in(1,1+\frac{\ln3}{4\ln\g{\Upsilon}_{\max}(\ch{N}_{[r]})}]$, it follows from Proposition~\ref{prop:nonasymptotic-channel} that
\begin{align}
	&-\frac{1}{n}\ln P_\abb{err}\fleft(n;\en{N}\fright) \notag\\
	&\leq\sup_{s_{[r]}\in\s{P}_r}\inf_{\ch{T}\in\s{C}_{A\to B}}\sum_{x\in[r]}s_x\g{D}_\alpha\fleft(\ch{T}\middle\|\ch{N}_x\fright)+\frac{\alpha}{n\left(\alpha-1\right)}\ln\left(\frac{1}{p_{\min}}\right) \label{pf:nonasymptotic-barycentric-8}\\
	&=\inf_{\ch{T}\in\s{C}_{A\to B}}\sup_{s_{[r]}\in\s{P}_r}\sum_{x\in[r]}s_x\g{D}_\alpha\fleft(\ch{T}\middle\|\ch{N}_x\fright)+\frac{\alpha}{n\left(\alpha-1\right)}\ln\left(\frac{1}{p_{\min}}\right) \label{pf:nonasymptotic-barycentric-9}\\
	&=\inf_{\ch{T}\in\s{C}_{A\to B}'}\sup_{s_{[r]}\in\s{P}_r}\sum_{x\in[r]}s_x\g{D}_\alpha\fleft(\ch{T}\middle\|\ch{N}_x\fright)+\frac{\alpha}{n\left(\alpha-1\right)}\ln\left(\frac{1}{p_{\min}}\right) \label{pf:nonasymptotic-barycentric-10}\\
	&\leq\inf_{\ch{T}\in\s{C}_{A\to B}'}\sup_{s_{[r]}\in\s{P}_r}\sum_{x\in[r]}s_x\left(\g{D}\fleft(\ch{T}\middle\|\ch{N}_x\fright)\vphantom{4\left(\alpha-1\right)\left(\ln\g{\Upsilon}\fleft(\ch{T}\middle\|\ch{N}_x\fright)\right)^2}\right. \notag\\
	&\qquad\left.\vphantom{}+4\left(\alpha-1\right)\left(\ln\g{\Upsilon}\fleft(\ch{T}\middle\|\ch{N}_x\fright)\right)^2\right)+\frac{\alpha}{n\left(\alpha-1\right)}\ln\left(\frac{1}{p_{\min}}\right) \label{pf:nonasymptotic-barycentric-11}\\
	&\leq\inf_{\ch{T}\in\s{C}_{A\to B}'}\sup_{s_{[r]}\in\s{P}_r}\sum_{x\in[r]}s_x\g{D}\fleft(\ch{T}\middle\|\ch{N}_x\fright) \notag\\
	&\qquad+4\left(\alpha-1\right)\left(\ln\g{\Upsilon}_{\max}\fleft(\ch{N}_{[r]}\fright)\right)^2+\frac{\alpha}{n\left(\alpha-1\right)}\ln\left(\frac{1}{p_{\min}}\right) \label{pf:nonasymptotic-barycentric-12}\\
	&=\inf_{\ch{T}\in\s{C}_{A\to B}}\sup_{s_{[r]}\in\s{P}_r}\sum_{x\in[r]}s_x\g{D}\fleft(\ch{T}\middle\|\ch{N}_x\fright) \notag\\
	&\qquad+4\left(\alpha-1\right)\left(\ln\g{\Upsilon}_{\max}\fleft(\ch{N}_{[r]}\fright)\right)^2+\frac{\alpha}{n\left(\alpha-1\right)}\ln\left(\frac{1}{p_{\min}}\right) \\
	&=R^\g{D}\fleft(\ch{N}_{[r]}\fright)+4\left(\alpha-1\right)\left(\ln\g{\Upsilon}_{\max}\fleft(\ch{N}_{[r]}\fright)\right)^2 \notag\\
	&\qquad+\frac{\alpha}{n\left(\alpha-1\right)}\ln\left(\frac{1}{p_{\min}}\right), \label{pf:nonasymptotic-barycentric-13}
\end{align}
where in \eqref{pf:nonasymptotic-barycentric-10}--\eqref{pf:nonasymptotic-barycentric-12} we denote $\s{C}_{A\to B}'\equiv\{\ch{T}\in\s{C}_{A\to B}\colon J_\ch{T}^0\leq\bigwedge_{x\in[r]}J_{\ch{N}_x}^0\}$.  Here \eqref{pf:nonasymptotic-barycentric-9} follows from applying Lemma~\ref{lem:radius-channel-alternative} to the left geometric Rényi channel radius and \eqref{eq:radius-channel-2}; Eq.~\eqref{pf:nonasymptotic-barycentric-11} follows from Lemma~\ref{lem:geometric} and \eqref{pf:nonasymptotic-barycentric-7}; Eq.~\eqref{pf:nonasymptotic-barycentric-12} follows from \eqref{pf:nonasymptotic-barycentric-4}; Eq.~\eqref{pf:nonasymptotic-barycentric-13} follows from \eqref{eq:radius-channel-2}.  Inserting $\alpha\coloneq1+\frac{\ln3}{4\ln\g{\Upsilon}_{\max}(\ch{N}_{[r]})\sqrt{n}}\in(1,1+\frac{\ln3}{4\ln\g{\Upsilon}_{\max}(\ch{N}_{[r]})}]$ to \eqref{pf:nonasymptotic-barycentric-8}, we obtain that
\begin{align}
	&-\frac{1}{n}\ln P_\abb{err}\fleft(n;\en{N}\fright) \notag\\
	&\leq R^\g{D}\fleft(\ch{N}_{[r]}\fright)+4\left(\alpha-1\right)\left(\ln\g{\Upsilon}_{\max}\fleft(\ch{N}_{[r]}\fright)\right)^2 \notag\\
	&\qquad+\frac{\alpha}{n\left(\alpha-1\right)}\ln\left(\frac{1}{p_{\min}}\right) \\
	&=R^\g{D}\fleft(\ch{N}_{[r]}\fright)+\frac{1}{\sqrt{n}}\left(\ln3\ln\g{\Upsilon}_{\max}\fleft(\ch{N}_{[r]}\fright)\vphantom{\left(\frac{4\ln\g{\Upsilon}_{\max}\fleft(\ch{N}_{[r]}\fright)}{\ln3}+\frac{1}{\sqrt{n}}\right)\ln\left(\frac{1}{p_{\min}}\right)}\right. \notag\\
	&\qquad\left.\vphantom{}+\left(\frac{4\ln\g{\Upsilon}_{\max}\fleft(\ch{N}_{[r]}\fright)}{\ln3}+\frac{1}{\sqrt{n}}\right)\ln\left(\frac{1}{p_{\min}}\right)\right) \\
	&\leq R^\g{D}\fleft(\ch{N}_{[r]}\fright)+\frac{1}{\sqrt{n}}\ln\g{\Upsilon}_{\max}\fleft(\ch{N}_{[r]}\fright)\left(\ln3+\frac{5}{\ln3}\ln\left(\frac{1}{p_{\min}}\right)\right). \label{pf:nonasymptotic-barycentric-14}
\end{align}
Here \eqref{pf:nonasymptotic-barycentric-14} follows from the fact that $\g{\Upsilon}_{\max}(\ch{N}_{[r]})\geq3$ due to \eqref{pf:nonasymptotic-barycentric-7} and \eqref{pf:geometric-2}.
\end{proof}

\section{Exact error exponent of state exclusion for multiple pure states}
\label{app:pure}

\begin{proposition}[Exact error exponent for pure states]
\label{prop:pure}
Let $\en{E}\equiv(p_{[r]},\rho_{[r]})$ be an ensemble of states with $p_{[r]}\in\itr(\s{P}_r)$ and $\rho_{[r]}\in\s{D}_A^{[r]}$.  If there exists $\{x_1,x_2,x_3\}\subseteq[r]$ such that $\rho_{x_1}$, $\rho_{x_2}$, and $\rho_{x_3}$ are three distinct pure states, then there exists a positive integer $n'$ such that, for every positive integer $n\geq n'$,
\begin{align}
	P_\abb{err}\fleft(\en{E}^n\fright)&=0,
\end{align}
and consequently,
\begin{align}
	\lim_{n\to\infty}-\frac{1}{n}\ln P_\abb{err}\fleft(\en{E}^n\fright)&=\infty.
\end{align}
\end{proposition}

\begin{proof}
Let $\rho_{x_1}\equiv\op{\psi_{x_1}}{\psi_{x_1}}$, $\rho_{x_2}\equiv\op{\psi_{x_2}}{\psi_{x_2}}$, and $\rho_{x_3}\equiv\op{\psi_{x_3}}{\psi_{x_3}}$ be three distinct pure states such that $\{x_1,x_2,x_3\}\subseteq[r]$.  Denote $a\equiv|\ip{\psi_{x_1}}{\psi_{x_2}}|^2$, $b\equiv|\ip{\psi_{x_2}}{\psi_{x_3}}|^2$, and $c\equiv|\ip{\psi_{x_3}}{\psi_{x_1}}|^2$.  It follows from Ref.~\cite[Section~V.B]{caves2002ConditionsCompatibilityQuantumstate} (also see Ref.~\cite[Eq.~(9)]{barrett2014NoPsepistemicModel}) that the tuple of three pure states $(\rho_{x_1},\rho_{x_2},\rho_{x_3})$ is perfectly antidistinguishable by an orthogonal rank-$1$ projective measurement if and only if
\begin{align}
	\label{pf:pure-1}
	&\begin{cases}
		a+b+c<1, \\
		\left(a+b+c-1\right)^2\geq4abc.
	\end{cases}
\end{align}
Define $m\coloneq\max\{a,b,c\}<1$.  Let $n'$ be a positive integer such that $m^{n'}\leq\frac{1}{4}$.  It follows that
\begin{align}
	a^{n'}+b^{n'}+c^{n'}&\leq 3m^{n'} \\
	&\leq\frac{3}{4} \\
	&<1, \label{pf:pure-2}
\end{align}
and
\begin{align}
	\left(a^{n'}+b^{n'}+c^{n'}-1\right)^2&\geq\left(1-3m^{n'}\right)^2 \\
	&\geq\frac{1}{16} \\
	&\geq4\left(m^{n'}\right)^3 \\
	&\geq4a^{n'}b^{n'}c^{n'}. \label{pf:pure-3}
\end{align}
According to the criterion in \eqref{pf:pure-1}, it follows from \eqref{pf:pure-2} and \eqref{pf:pure-3} that the $n'$-fold tuple of pure states $(\rho_{x_1}^{\otimes n'},\rho_{x_2}^{\otimes n'},\rho_{x_3}^{\otimes n'})$ is perfectly antidistinguishable.  Furthermore, the $n'$-fold tuple of the original $r$ states $\rho_{[r]}^{\otimes n'}$ is also perfectly antidistinguishable, since any exclusion strategy that incurs no error on $(\rho_{x_1}^{\otimes n'},\rho_{x_2}^{\otimes n'},\rho_{x_3}^{\otimes n'})$ also incurs no error on $\rho_{[r]}^{\otimes n'}$.  Therefore, there exists a positive integer $n'$ such that $P_\abb{err}(\en{E}^n)=0$ for all $n\geq n'$, leading to the desired statement.
\end{proof}

\begin{remark}[On the infinite error exponent for pure states]
\label{rem:pure}
While the error exponent of state exclusion for two distinct pure states is finite unless the states are orthogonal~\cite{audenaert2007DiscriminatingStatesQuantum,nussbaum2009ChernoffLowerBound}, the error exponent for three or more distinct pure states is always infinite, regardless of their orthogonality.  In addition, the error exponent being infinite does not mean that the tuple of states is perfectly antidistinguishable.  As implied by the proof of Proposition~\ref{prop:pure}, it is possible that the error probability reaches strictly zero in the nonasymptotic regime while having a nonzero value in the one-shot regime.  This poses a remarkable distinction between state exclusion and state discrimination.  For the latter, regardless of the states to be distinguished being pure or mixed, the error exponent is infinite if and only if the tuple of states is perfectly distinguishable if and only if the states of concern are orthogonal to each other.
\end{remark}

\section{Upper bound on the error exponent of channel exclusion under indefinite-causal-order strategies}
\label{app:indefinite}

Let $\en{N}\equiv(p_{[r]},\ch{N}_{[r]})$ be an ensemble of channels with $p_{[r]}\in\s{P}_r$ and $\ch{N}_{[r]}\in\s{C}_{A\to B}^{[r]}$.  Without assuming a definite causal order~\cite{chiribella2013QuantumComputationsDefinite,bavaresco2021StrictHierarchyParallel}, we consider a general strategy for channel exclusion with $n$ slots, represented by a tuple $\ch{I}^{(n)}\equiv(\Theta^{(n)},\Lambda_{[r]})$, where 
\begin{align}
	\Theta^{(n)}&\colon\underbrace{\s{C}_{A\to B}\times\dots\times\s{C}_{A\to B}}_\textnormal{Cartesian product of $n$ terms}\to\s{D}_C
\end{align}
is a transformation that maps $n$ channels to a state and $\Lambda_{[r]}\in\s{M}_{C,r}$ is a POVM.  We make no assumptions about the transformation $\Theta^{(n)}$ other than convex linearity in each of its $n$ arguments in the sense that
\begin{align}
	&\Theta^{(n)}\fleft(\ch{N}_1,\dots,\ch{N}_{i-1},p\ch{N}_i+(1-p)\ch{N}_i',\ch{N}_{i+1},\dots,\ch{N}_n\fright) \notag\\
	&=p\Theta^{(n)}\fleft(\ch{N}_1,\dots,\ch{N}_{i-1},\ch{N}_i,\ch{N}_{i+1},\dots,\ch{N}_n\fright) \notag\\
	&\qquad+(1-p)\Theta^{(n)}\fleft(\ch{N}_1,\dots,\ch{N}_{i-1},\ch{N}_i',\ch{N}_{i+1},\dots,\ch{N}_n\fright) \notag\\
	&\qquad\forall\ch{N}_{[n]}\in\s{C}_{A\to B}^{[n]},\;\ch{N}_i'\in\s{C}_{A\to B},\;p\in[0,1],
\end{align}
which is a minimal requirement for the strategy $\ch{I}^{(n)}$ to be physically legitimate.  We introduce the following shorthand for the post-transformation state:
\begin{align}
	\rho_x^{(n)}&\equiv\Theta^{(n)}\fleft(\ch{N}_x,\dots,\ch{N}_x\fright)\quad\forall x\in[r]. \label{pf:nonasymptotic-channel-indefinite-1}
\end{align}
The \emph{(nonasymptotic) error probability} of channel exclusion under indefinite-causal-order strategies with $n$ slots for the ensemble $\en{N}$ is thus bounded from below by
\begin{align}
	P_\abb{err}^\abb{ICO}\fleft(n;\en{N}\fright)\geq \underline{P}_\abb{err}\fleft(n;\en{N}\fright)&\coloneq\inf_{\ch{I}^{(n)}}\sum_{x\in[r]}p_x\tr\fleft[\Lambda_x\rho_x^{(n)}\fright] \\
	&=\inf_{\Theta^{(n)}}P_\abb{err}\fleft(\en{E}^{(n)}\fright), \label{eq:error-probability-channel-indefinite}
\end{align}
where $\en{E}^{(n)}\equiv(p_{[r]},\rho_{[r]}^{(n)})$ with $\rho_{[r]}^{(n)}\equiv(\rho_1^{(n)},\rho_2^{(n)},\dots,\rho_r^{(n)})$.  The \emph{(asymptotic) error exponent} of channel exclusion under indefinite-causal-order strategies for the ensemble $\en{N}$ is bounded from above by
\begin{align}
	\liminf_{n\to\infty}-\frac{1}{n}\ln P_\abb{err}^\abb{ICO}\fleft(n;\en{N}\fright)&\leq\liminf_{n\to\infty}-\frac{1}{n}\ln\underline{P}_\abb{err}\fleft(n;\en{N}\fright) \\
	&\leq\limsup_{n\to\infty}-\frac{1}{n}\ln\underline{P}_\abb{err}\fleft(n;\en{N}\fright). \label{eq:error-exponent-indefinite}
\end{align}

Let $\ch{N}\in\s{C}_{A\to B}$ be a channel, and let $\ch{M}\in\s{CP}_{A\to B}$ be a completely positive map.  The \emph{max-channel divergence} is defined as
\begin{align}
	D_{\max}\fleft(\ch{N}\middle\|\ch{M}\fright)&\coloneq\sup_{\rho\in\s{D}_{RA}}D_{\max}\fleft(\ch{N}_{A\to B}\fleft[\rho_{RA}\fright]\middle\|\ch{M}_{A\to B}\fleft[\rho_{RA}\fright]\fright) \\
	&=D_{\max}\fleft(J_\ch{N}\middle\|J_\ch{M}\fright), \label{eq:max-channel-choi}
\end{align}
where the equality between the max-channel divergence and the max-divergence of the channels' Choi operators in \eqref{eq:max-channel-choi} was established in Refs.~\cite{diaz2018UsingReusingCoherence,wilde2020AmortizedChannelDivergence}.  For a convex-linear transformation $\Theta^{(n)}\colon\s{C}_{A\to B}\times\dots\times\s{C}_{A\to B}\to\s{D}_C$ with $n$ arguments, as considered above in the general strategy for channel exclusion, the max-channel divergence has the following property~\cite[Theorem~13]{regula2021FundamentalLimitationsDistillation}:
\begin{align}
	&D_{\max}\fleft(\Theta^{(n)}\fleft(\ch{N}_1,\ch{N}_2,\dots,\ch{N}_n\fright)\middle\|\Theta^{(n)}\fleft(\ch{M}_1,\ch{M}_2,\dots,\ch{M}_n\fright)\fright) \notag\\
	&\leq\sum_{i\in[n]}D_{\max}\fleft(\ch{N}_i\middle\|\ch{M}_i\fright)\quad\forall\ch{N}_{[n]},\ch{M}_{[n]}\in\s{C}_{A\to B}^{[n]}. \label{eq:max-chain}
\end{align}

\begin{proposition}[Converse bound on the nonasymptotic error probability for channels under indefinite-causal-order strategies]
\label{prop:nonasymptotic-channel-indefinite}
Let $\en{N}\equiv(p_{[r]},\ch{N}_{[r]})$ be an ensemble of channels with $p_{[r]}\in\itr(\s{P}_r)$ and $\ch{N}_{[r]}\in\s{C}_{A\to B}^{[r]}$.  Then for every positive integer $n$,
\begin{align}
	\label{eq:nonasymptotic-channel-indefinite}
	&-\frac{1}{n}\ln\underline{P}_\abb{err}\fleft(n;\en{N}\fright) \notag\\
	&\leq\sup_{s_{[r]}\in\s{P}_r}\inf_{\ch{T}\in\s{C}_{A\to B}}\sum_{x\in[r]}s_xD_{\max}\fleft(\ch{T}\middle\|\ch{N}_x\fright)+\frac{1}{n}\ln\left(\frac{1}{p_{\min}}\right),
\end{align}
where $p_{\min}\equiv\min_{x\in[r]}p_x$.
\end{proposition}

\begin{proof}
Let $\ch{I}^{(n)}\equiv(\Theta^{(n)},\Lambda_{[r]})$ be a general strategy with $n$ slots, where $\Theta^{(n)}\colon\s{C}_{A\to B}\times\dots\times\s{C}_{A\to B}\to\s{D}_C$ is a convex-linear transformation and $\Lambda_{[r]}\in\s{M}_{C,r}$ is a POVM, and let $\ch{T}\in\s{C}_{A\to B}$ be a channel.  Applying Proposition~\ref{prop:oneshot-state} to the ensemble of states $\en{E}^{(n)}\equiv(p_{[r]},\rho_{[r]}^{(n)})$ and taking the supremum over $\alpha\in(1,\infty)$, we have that
\begin{align}
	&-\frac{1}{n}\ln P_\abb{err}\fleft(\en{E}^{(n)}\fright) \notag\\
	&\leq\sup_{\alpha\in(1,\infty)}\sup_{s_{[r]}\in\s{P}_r}\inf_{\tau\in\aff\fleft(\s{D}_C\fright)}\frac{1}{n}\sum_{x\in[r]}s_x\sw{D}_\alpha\fleft(\tau\middle\|\rho_x^{(n)}\fright) \notag\\
	&\qquad+\frac{1}{n}\ln\left(\frac{1}{p_{\min}}\right) \\
	&\leq\sup_{\alpha\in(1,\infty)}\sup_{s_{[r]}\in\s{P}_r}\inf_{\tau\in\s{D}_C}\frac{1}{n}\sum_{x\in[r]}s_x\sw{D}_\alpha\fleft(\tau\middle\|\rho_x^{(n)}\fright)+\frac{1}{n}\ln\left(\frac{1}{p_{\min}}\right) \\
	&=\sup_{s_{[r]}\in\s{P}_r}\inf_{\tau\in\s{D}_C}\sup_{\alpha\in(1,\infty)}\frac{1}{n}\sum_{x\in[r]}s_x\sw{D}_\alpha\fleft(\tau\middle\|\rho_x^{(n)}\fright)+\frac{1}{n}\ln\left(\frac{1}{p_{\min}}\right) \label{pf:nonasymptotic-channel-indefinite-2}\\
	&=\sup_{s_{[r]}\in\s{P}_r}\inf_{\tau\in\s{D}_C}\frac{1}{n}\sum_{x\in[r]}s_xD_{\max}\fleft(\tau\middle\|\rho_x^{(n)}\fright)+\frac{1}{n}\ln\left(\frac{1}{p_{\min}}\right) \label{pf:nonasymptotic-channel-indefinite-3}\\
	&=\sup_{s_{[r]}\in\s{P}_r}\inf_{\ch{T}\in\s{C}_{A\to B}}\frac{1}{n}\sum_{x\in[r]}s_x \notag\\
	&\qquad D_{\max}\fleft(\Theta^{(n)}\fleft(\ch{T},\dots,\ch{T}\fright)\middle\|\Theta^{(n)}\fleft(\ch{N}_x,\dots,\ch{N}_x\fright)\fright) \notag\\
	&\qquad+\frac{1}{n}\ln\left(\frac{1}{p_{\min}}\right) \label{pf:nonasymptotic-channel-indefinite-4}\\
	&\leq\sup_{s_{[r]}\in\s{P}_r}\inf_{\ch{T}\in\s{C}_{A\to B}}\sum_{x\in[r]}s_xD_{\max}\fleft(\ch{T}\middle\|\ch{N}_x\fright)+\frac{1}{n}\ln\left(\frac{1}{p_{\min}}\right). \label{pf:nonasymptotic-channel-indefinite-5}
\end{align}
Here \eqref{pf:nonasymptotic-channel-indefinite-2} follows from applying Lemma~\ref{lem:radius-alternative} to the left sandwiched Rényi radius and \eqref{eq:radius-2}; Eq.~\eqref{pf:nonasymptotic-channel-indefinite-3} uses the nondecreasing monotonicity of the sandwiched Rényi divergence in $\alpha$ and its limit as $\alpha\to\infty$, which is given by the max-divergence (see \eqref{eq:max-extended-limit}); Eq.~\eqref{pf:nonasymptotic-channel-indefinite-4} follows from \eqref{pf:nonasymptotic-channel-indefinite-1} and the fact that $\Theta^{(n)}(\ch{T},\dots,\ch{T})\in\aff(\s{D}_C)$ for all $\ch{T}\in\s{C}_{A\to B}$; Eq.~\eqref{pf:nonasymptotic-channel-indefinite-5} follows from \eqref{eq:max-chain} (also see Ref.~\cite[Eq.~(88)]{regula2024PostselectedQuantumHypothesis}).  Since \eqref{pf:nonasymptotic-channel-indefinite-5} holds for every general strategy $\ch{I}^{(n)}$, it follows from \eqref{eq:error-probability-channel-indefinite} that
\begin{align}
	&-\frac{1}{n}\ln\underline{P}_\abb{err}\fleft(n;\en{N}\fright) \notag\\
	&=\sup_{\Theta^{(n)}}-\frac{1}{n}\ln P_\abb{err}\fleft(\en{E}^{(n)}\fright) \\
	&\leq\sup_{s_{[r]}\in\s{P}_r}\inf_{\ch{T}\in\s{C}_{A\to B}}\sum_{x\in[r]}s_xD_{\max}\fleft(\ch{T}\middle\|\ch{N}_x\fright)+\frac{1}{n}\ln\left(\frac{1}{p_{\min}}\right),
\end{align}
thus concluding the proof.
\end{proof}

\begin{theorem}[Upper bound on the error exponent for channels under indefinite-causal-order strategies]
\label{thm:asymptotic-channel-indefinite}
\begin{align}
	&\limsup_{n\to\infty}-\frac{1}{n}\ln\underline{P}_\abb{err}\fleft(n;\en{N}\fright) \notag\\
	&\leq R^{D_{\max}}\fleft(\ch{N}_{[r]}\fright) \label{eq:asymptotic-channel-indefinite-1}\\
	&=\sup_{s_{[r]}\in\s{P}_r}\inf_{\ch{T}\in\s{C}_{A\to B}}\sum_{x\in[r]}s_xD_{\max}\fleft(\ch{T}\middle\|\ch{N}_x\fright). \label{eq:asymptotic-channel-indefinite-2}
\end{align}
\end{theorem}

\begin{proof}
Taking the limit superior as $n\to\infty$ on both sides of \eqref{eq:nonasymptotic-channel-indefinite} leads to the inequality between the left-hand side of \eqref{eq:asymptotic-channel-indefinite-1} and the right-hand side of \eqref{eq:asymptotic-channel-indefinite-2}.  Applying Lemma~\ref{lem:radius-alternative} to the left max-channel radius leads to the equality between the right-hand side of \eqref{eq:asymptotic-channel-indefinite-1} and the right-hand side of \eqref{eq:asymptotic-channel-indefinite-2}.
\end{proof}

\begin{remark}[Semidefinite representation of the channel $R^{D_{\max}}$]
\label{rem:SDP-indefinite}
It follows from the definition of the left max-channel radius (see \eqref{eq:radius-channel-1}) that
\begin{align}
	&R^{D_{\max}}\fleft(\ch{N}_{[r]}\fright) \notag\\
	&=\inf_{\ch{T}\in\s{C}_{A\to B}}\max_{x\in[r]}D_{\max}\fleft(\ch{T}\middle\|\ch{N}_x\fright) \\
	&=\inf_{\ch{T}\in\s{C}_{A\to B}}\inf_{\lambda\in(0,\infty)}\left\{\ln\lambda\colon\ln\lambda\geq D_{\max}\fleft(J_\ch{T}\middle\|J_{\ch{N}_x}\fright)\;\forall x\in[r]\right\} \\
	&=\inf_{\substack{J_\ch{T}\in\s{PSD}_{AB}, \\
		\lambda\in(0,\infty)
	}}\left\{\ln\lambda\colon J_\ch{T}\leq\lambda J_{\ch{N}_x},\;\tr_B\fleft[J_{\ch{T},AB}\fright]=\1_A\right\}. \label{pf:SDP-indefinite-1}
\end{align}
Here \eqref{pf:SDP-indefinite-1} follows from the equality between max-channel divergence and the max-divergence of the channel's Choi operators (see \eqref{eq:max-channel-choi}) and the definition of the max-divergence (see \eqref{eq:max-extended}).  This shows that the converse bound in Proposition~\ref{prop:nonasymptotic-channel-indefinite} and the upper bound in Theorem~\ref{thm:asymptotic-channel-indefinite} are both efficiently computable via an SDP.
\end{remark}


\bibliographystyle{IEEEtran}
\bibliography{Library}

\begin{thebibliography}{10}
\providecommand{\url}[1]{#1}
\csname url@samestyle\endcsname
\providecommand{\newblock}{\relax}
\providecommand{\bibinfo}[2]{#2}
\providecommand{\BIBentrySTDinterwordspacing}{\spaceskip=0pt\relax}
\providecommand{\BIBentryALTinterwordstretchfactor}{4}
\providecommand{\BIBentryALTinterwordspacing}{\spaceskip=\fontdimen2\font plus
\BIBentryALTinterwordstretchfactor\fontdimen3\font minus
  \fontdimen4\font\relax}
\providecommand{\BIBforeignlanguage}[2]{{%
\expandafter\ifx\csname l@#1\endcsname\relax
\typeout{** WARNING: IEEEtran.bst: No hyphenation pattern has been}%
\typeout{** loaded for the language `#1'. Using the pattern for}%
\typeout{** the default language instead.}%
\else
\language=\csname l@#1\endcsname
\fi
#2}}
\providecommand{\BIBdecl}{\relax}
\BIBdecl

\bibitem{wilde2017QuantumInformationTheory}
\BIBentryALTinterwordspacing
M.~M. Wilde, \emph{Quantum Information Theory}, 2nd~ed.\hskip 1em plus 0.5em
  minus 0.4em\relax Cambridge University Press, 2017. [Online]. Available:
  \url{https://www.cambridge.org/core/product/identifier/9781316809976/type/book}
\BIBentrySTDinterwordspacing

\bibitem{hayashi2017QuantumInformationTheory}
\BIBentryALTinterwordspacing
M.~Hayashi, \emph{Quantum Information Theory: {{Mathematical}} Foundation},
  ser. Graduate {{Texts}} in {{Physics}}.\hskip 1em plus 0.5em minus
  0.4em\relax Berlin, Heidelberg: Springer Berlin Heidelberg, 2017. [Online].
  Available: \url{http://link.springer.com/10.1007/978-3-662-49725-8}
\BIBentrySTDinterwordspacing

\bibitem{watrous2018TheoryQuantumInformation}
\BIBentryALTinterwordspacing
J.~Watrous, \emph{The Theory of Quantum Information}, 1st~ed.\hskip 1em plus
  0.5em minus 0.4em\relax Cambridge University Press, Apr. 2018. [Online].
  Available:
  \url{https://www.cambridge.org/core/product/identifier/9781316848142/type/book}
\BIBentrySTDinterwordspacing

\bibitem{khatri2024PrinciplesQuantumCommunication}
\BIBentryALTinterwordspacing
S.~Khatri and M.~M. Wilde, \emph{Principles of Quantum Communication Theory:
  {{A}} Modern Approach}.\hskip 1em plus 0.5em minus 0.4em\relax arXiv, Feb.
  2024. [Online]. Available: \url{http://arxiv.org/abs/2011.04672v2}
\BIBentrySTDinterwordspacing

\bibitem{bandyopadhyay2014ConclusiveExclusionQuantum}
\BIBentryALTinterwordspacing
S.~Bandyopadhyay, R.~Jain, J.~Oppenheim, and C.~Perry, ``Conclusive exclusion
  of quantum states,'' \emph{Physical Review A}, vol.~89, p. 022336, Feb. 2014.
  [Online]. Available:
  \url{https://link.aps.org/doi/10.1103/PhysRevA.89.022336}
\BIBentrySTDinterwordspacing

\bibitem{brun2002HowMuchState}
\BIBentryALTinterwordspacing
T.~A. Brun, J.~Finkelstein, and N.~D. Mermin, ``How much state assignments can
  differ,'' \emph{Physical Review A}, vol.~65, p. 032315, Feb. 2002. [Online].
  Available: \url{https://link.aps.org/doi/10.1103/PhysRevA.65.032315}
\BIBentrySTDinterwordspacing

\bibitem{caves2002ConditionsCompatibilityQuantumstate}
\BIBentryALTinterwordspacing
C.~M. Caves, C.~A. Fuchs, and R.~Schack, ``Conditions for compatibility of
  quantum-state assignments,'' \emph{Physical Review A}, vol.~66, p. 062111,
  Dec. 2002. [Online]. Available:
  \url{https://link.aps.org/doi/10.1103/PhysRevA.66.062111}
\BIBentrySTDinterwordspacing

\bibitem{pusey2012RealityQuantumState}
\BIBentryALTinterwordspacing
M.~F. Pusey, J.~Barrett, and T.~Rudolph, ``On the reality of the quantum
  state,'' \emph{Nature Physics}, vol.~8, pp. 475--478, Jun. 2012. [Online].
  Available: \url{https://www.nature.com/articles/nphys2309}
\BIBentrySTDinterwordspacing

\bibitem{leifer2014QuantumStateReal}
\BIBentryALTinterwordspacing
M.~S. Leifer, ``Is the quantum state real? {{An}} extended review of
  {$\psi$}-ontology theorems,'' \emph{Quanta}, vol.~3, p.~67, Nov. 2014.
  [Online]. Available:
  \url{http://quanta.ws/ojs/index.php/quanta/article/view/22}
\BIBentrySTDinterwordspacing

\bibitem{barrett2014NoPsepistemicModel}
\BIBentryALTinterwordspacing
J.~Barrett, E.~G. Cavalcanti, R.~Lal, and O.~J.~E. Maroney, ``No
  {$\psi$}-epistemic model can fully explain the indistinguishability of
  quantum states,'' \emph{Physical Review Letters}, vol. 112, p. 250403, Jun.
  2014. [Online]. Available:
  \url{https://link.aps.org/doi/10.1103/PhysRevLett.112.250403}
\BIBentrySTDinterwordspacing

\bibitem{heinosaari2018AntidistinguishabilityPureQuantum}
\BIBentryALTinterwordspacing
T.~Heinosaari and O.~Kerppo, ``Antidistinguishability of pure quantum states,''
  \emph{Journal of Physics A: Mathematical and Theoretical}, vol.~51, p.
  365303, Sep. 2018. [Online]. Available:
  \url{https://iopscience.iop.org/article/10.1088/1751-8121/aad1fc}
\BIBentrySTDinterwordspacing

\bibitem{crickmore2020UnambiguousQuantumState}
\BIBentryALTinterwordspacing
J.~Crickmore, I.~V. Puthoor, B.~Ricketti, S.~Croke, M.~Hillery, and
  E.~Andersson, ``Unambiguous quantum state elimination for qubit sequences,''
  \emph{Physical Review Research}, vol.~2, p. 013256, Mar. 2020. [Online].
  Available: \url{https://link.aps.org/doi/10.1103/PhysRevResearch.2.013256}
\BIBentrySTDinterwordspacing

\bibitem{russo2023InnerProductsPure}
\BIBentryALTinterwordspacing
V.~Russo and J.~Sikora, ``Inner products of pure states and their
  antidistinguishability,'' \emph{Physical Review A}, vol. 107, p. L030202,
  Mar. 2023. [Online]. Available:
  \url{https://link.aps.org/doi/10.1103/PhysRevA.107.L030202}
\BIBentrySTDinterwordspacing

\bibitem{johnston2025TightBoundsAntidistinguishability}
\BIBentryALTinterwordspacing
N.~Johnston, V.~Russo, and J.~Sikora, ``Tight bounds for antidistinguishability
  and circulant sets of pure quantum states,'' \emph{Quantum}, vol.~9, p. 1622,
  Feb. 2025. [Online]. Available:
  \url{https://doi.org/10.22331/q-2025-02-04-1622}
\BIBentrySTDinterwordspacing

\bibitem{leifer2020NoncontextualityInequalitiesAntidistinguishability}
\BIBentryALTinterwordspacing
M.~Leifer and C.~Duarte, ``Noncontextuality inequalities from
  antidistinguishability,'' \emph{Physical Review A}, vol. 101, p. 062113, Jun.
  2020. [Online]. Available:
  \url{https://link.aps.org/doi/10.1103/PhysRevA.101.062113}
\BIBentrySTDinterwordspacing

\bibitem{havlicek2020SimpleCommunicationComplexity}
\BIBentryALTinterwordspacing
V.~Havl{\'i}{\v c}ek and J.~Barrett, ``Simple communication complexity
  separation from quantum state antidistinguishability,'' \emph{Physical Review
  Research}, vol.~2, p. 013326, Mar. 2020. [Online]. Available:
  \url{https://link.aps.org/doi/10.1103/PhysRevResearch.2.013326}
\BIBentrySTDinterwordspacing

\bibitem{collins2014RealizationQuantumDigital}
\BIBentryALTinterwordspacing
R.~J. Collins, R.~J. Donaldson, V.~Dunjko, P.~Wallden, P.~J. Clarke,
  E.~Andersson, J.~Jeffers, and G.~S. Buller, ``Realization of quantum digital
  signatures without the requirement of quantum memory,'' \emph{Physical Review
  Letters}, vol. 113, p. 040502, Jul. 2014. [Online]. Available:
  \url{https://link.aps.org/doi/10.1103/PhysRevLett.113.040502}
\BIBentrySTDinterwordspacing

\bibitem{amiri2021Imperfect1outof2Quantum}
\BIBentryALTinterwordspacing
R.~Amiri, R.~St{\'a}rek, D.~Reichmuth, I.~V. Puthoor, M.~Mi{\v c}uda, L.~Mi{\v
  s}ta, Jr., M.~Du{\v s}ek, P.~Wallden, and E.~Andersson, ``Imperfect
  1-out-of-2 quantum oblivious transfer: {{Bounds}}, a protocol, and its
  experimental implementation,'' \emph{PRX Quantum}, vol.~2, p. 010335, Mar.
  2021. [Online]. Available:
  \url{https://link.aps.org/doi/10.1103/PRXQuantum.2.010335}
\BIBentrySTDinterwordspacing

\bibitem{ducuara2020OperationalInterpretationWeightbased}
\BIBentryALTinterwordspacing
A.~F. Ducuara and P.~Skrzypczyk, ``Operational interpretation of weight-based
  resource quantifiers in convex quantum resource theories,'' \emph{Physical
  Review Letters}, vol. 125, p. 110401, Sep. 2020. [Online]. Available:
  \url{https://link.aps.org/doi/10.1103/PhysRevLett.125.110401}
\BIBentrySTDinterwordspacing

\bibitem{uola2020AllQuantumResources}
\BIBentryALTinterwordspacing
R.~Uola, T.~Bullock, T.~Kraft, J.-P. Pellonp{\"a}{\"a}, and N.~Brunner, ``All
  quantum resources provide an advantage in exclusion tasks,'' \emph{Physical
  Review Letters}, vol. 125, p. 110402, Sep. 2020. [Online]. Available:
  \url{https://link.aps.org/doi/10.1103/PhysRevLett.125.110402}
\BIBentrySTDinterwordspacing

\bibitem{mishra2024OptimalErrorExponents}
\BIBentryALTinterwordspacing
H.~K. Mishra, M.~Nussbaum, and M.~M. Wilde, ``On the optimal error exponents
  for classical and quantum antidistinguishability,'' \emph{Letters in
  Mathematical Physics}, vol. 114, p.~76, Jun. 2024. [Online]. Available:
  \url{https://link.springer.com/10.1007/s11005-024-01821-z}
\BIBentrySTDinterwordspacing

\bibitem{mosonyi2024GeometricRelativeEntropies}
\BIBentryALTinterwordspacing
M.~Mosonyi, G.~Bunth, and P.~Vrana, ``Geometric relative entropies and
  barycentric {{R{\'e}nyi}} divergences,'' \emph{Linear Algebra and its
  Applications}, vol. 699, pp. 159--276, Oct. 2024. [Online]. Available:
  \url{https://linkinghub.elsevier.com/retrieve/pii/S0024379524002490}
\BIBentrySTDinterwordspacing

\bibitem{harrow2010AdaptiveNonadaptiveStrategies}
\BIBentryALTinterwordspacing
A.~W. Harrow, A.~Hassidim, D.~W. Leung, and J.~Watrous, ``Adaptive versus
  nonadaptive strategies for quantum channel discrimination,'' \emph{Physical
  Review A}, vol.~81, p. 032339, Mar. 2010. [Online]. Available:
  \url{https://link.aps.org/doi/10.1103/PhysRevA.81.032339}
\BIBentrySTDinterwordspacing

\bibitem{wilde2020AmortizedChannelDivergence}
\BIBentryALTinterwordspacing
M.~M. Wilde, M.~Berta, C.~Hirche, and E.~Kaur, ``Amortized channel divergence
  for asymptotic quantum channel discrimination,'' \emph{Letters in
  Mathematical Physics}, vol. 110, pp. 2277--2336, Aug. 2020. [Online].
  Available: \url{https://link.springer.com/10.1007/s11005-020-01297-7}
\BIBentrySTDinterwordspacing

\bibitem{hayashi2009DiscriminationTwoChannels}
\BIBentryALTinterwordspacing
M.~Hayashi, ``Discrimination of two channels by adaptive methods and its
  application to quantum system,'' \emph{IEEE Transactions on Information
  Theory}, vol.~55, pp. 3807--3820, Aug. 2009. [Online]. Available:
  \url{http://ieeexplore.ieee.org/document/5165184/}
\BIBentrySTDinterwordspacing

\bibitem{wang2020AlogarithmicNegativity}
\BIBentryALTinterwordspacing
X.~Wang and M.~M. Wilde, ``{$\alpha$}-logarithmic negativity,'' \emph{Physical
  Review A}, vol. 102, p. 032416, Sep. 2020. [Online]. Available:
  \url{https://link.aps.org/doi/10.1103/PhysRevA.102.032416}
\BIBentrySTDinterwordspacing

\bibitem{polyanskiy2010ArimotoChannelCoding}
\BIBentryALTinterwordspacing
Y.~Polyanskiy and S.~Verd{\'u}, ``Arimoto channel coding converse and
  {{R{\'e}nyi}} divergence,'' in \emph{2010 48th {{Annual Allerton Conference}}
  on {{Communication}}, {{Control}}, and {{Computing}} ({{Allerton}})}.\hskip
  1em plus 0.5em minus 0.4em\relax Monticello, IL, USA: IEEE, Sep. 2010, pp.
  1327--1333. [Online]. Available:
  \url{http://ieeexplore.ieee.org/document/5707067/}
\BIBentrySTDinterwordspacing

\bibitem{sharma2012StrongConversesQuantum}
\BIBentryALTinterwordspacing
N.~Sharma and N.~A. Warsi, ``On the strong converses for the quantum channel
  capacity theorems,'' Jun. 2012. [Online]. Available:
  \url{https://arxiv.org/abs/1205.1712}
\BIBentrySTDinterwordspacing

\bibitem{wilde2014StrongConverseClassical}
\BIBentryALTinterwordspacing
M.~M. Wilde, A.~Winter, and D.~Yang, ``Strong converse for the classical
  capacity of entanglement-breaking and {{Hadamard}} channels via a sandwiched
  {{R{\'e}nyi}} relative entropy,'' \emph{Communications in Mathematical
  Physics}, vol. 331, pp. 593--622, Oct. 2014. [Online]. Available:
  \url{http://link.springer.com/10.1007/s00220-014-2122-x}
\BIBentrySTDinterwordspacing

\bibitem{muller-lennert2013QuantumRenyiEntropies}
\BIBentryALTinterwordspacing
M.~{M{\"u}ller-Lennert}, F.~Dupuis, O.~Szehr, S.~Fehr, and M.~Tomamichel, ``On
  quantum {{R{\'e}nyi}} entropies: {{A}} new generalization and some
  properties,'' \emph{Journal of Mathematical Physics}, vol.~54, p. 122203,
  Dec. 2013. [Online]. Available:
  \url{https://pubs.aip.org/jmp/article/54/12/122203/233328/On-quantum-Renyi-entropies-A-new-generalization}
\BIBentrySTDinterwordspacing

\bibitem{beigi2013SandwichedRenyiDivergence}
\BIBentryALTinterwordspacing
S.~Beigi, ``Sandwiched {{R{\'e}nyi}} divergence satisfies data processing
  inequality,'' \emph{Journal of Mathematical Physics}, vol.~54, p. 122202,
  Dec. 2013. [Online]. Available:
  \url{https://pubs.aip.org/jmp/article/54/12/122202/233308/Sandwiched-Renyi-divergence-satisfies-data}
\BIBentrySTDinterwordspacing

\bibitem{frank2013MonotonicityRelativeRenyi}
\BIBentryALTinterwordspacing
R.~L. Frank and E.~H. Lieb, ``Monotonicity of a relative {{R{\'e}nyi}}
  entropy,'' \emph{Journal of Mathematical Physics}, vol.~54, p. 122201, Dec.
  2013. [Online]. Available:
  \url{https://pubs.aip.org/jmp/article/54/12/122201/233338/Monotonicity-of-a-relative-Renyi-entropy}
\BIBentrySTDinterwordspacing

\bibitem{mosonyi2015QuantumHypothesisTesting}
\BIBentryALTinterwordspacing
M.~Mosonyi and T.~Ogawa, ``Quantum hypothesis testing and the operational
  interpretation of the quantum {{R{\'e}nyi}} relative entropies,''
  \emph{Communications in Mathematical Physics}, vol. 334, pp. 1617--1648, Mar.
  2015. [Online]. Available:
  \url{http://link.springer.com/10.1007/s00220-014-2248-x}
\BIBentrySTDinterwordspacing

\bibitem{tomamichel2016QuantumInformationProcessing}
\BIBentryALTinterwordspacing
M.~Tomamichel, \emph{Quantum Information Processing with Finite Resources},
  ser. {{SpringerBriefs}} in {{Mathematical Physics}}.\hskip 1em plus 0.5em
  minus 0.4em\relax Cham: Springer International Publishing, 2016, vol.~5.
  [Online]. Available: \url{http://link.springer.com/10.1007/978-3-319-21891-5}
\BIBentrySTDinterwordspacing

\bibitem{umegaki1962ConditionalExpectationOperator}
\BIBentryALTinterwordspacing
H.~Umegaki, ``Conditional expectation in an operator algebra, {{IV}} (entropy
  and information),'' \emph{Kodai Mathematical Journal}, vol.~14, Jan. 1962.
  [Online]. Available:
  \url{https://projecteuclid.org/journals/kodai-mathematical-journal/volume-14/issue-2/Conditional-expectation-in-an-operator-algebra-IV-Entropy-and-information/10.2996/kmj/1138844604.full}
\BIBentrySTDinterwordspacing

\bibitem{matsumoto2018NewQuantumVersion}
\BIBentryALTinterwordspacing
K.~Matsumoto, ``A new quantum version of {$f$}-divergence,'' in \emph{Reality
  and {{Measurement}} in {{Algebraic Quantum Theory}}}, M.~Ozawa,
  J.~Butterfield, H.~Halvorson, M.~R{\'e}dei, Y.~Kitajima, and F.~Buscemi,
  Eds., vol. 261.\hskip 1em plus 0.5em minus 0.4em\relax Singapore: Springer
  Singapore, 2018, pp. 229--273. [Online]. Available:
  \url{http://link.springer.com/10.1007/978-981-13-2487-1_10}
\BIBentrySTDinterwordspacing

\bibitem{petz1998ContractionGeneralizedRelative}
\BIBentryALTinterwordspacing
D.~Petz and M.~B. Ruskai, ``Contraction of generalized relative entropy under
  stochastic mappings on matrices,'' \emph{Infinite Dimensional Analysis,
  Quantum Probability and Related Topics}, vol.~01, pp. 83--89, Jan. 1998.
  [Online]. Available:
  \url{https://www.worldscientific.com/doi/abs/10.1142/S0219025798000077}
\BIBentrySTDinterwordspacing

\bibitem{fang2021GeometricRenyiDivergence}
\BIBentryALTinterwordspacing
K.~Fang and H.~Fawzi, ``Geometric {{R{\'e}nyi}} divergence and its applications
  in quantum channel capacities,'' \emph{Communications in Mathematical
  Physics}, vol. 384, pp. 1615--1677, Jun. 2021. [Online]. Available:
  \url{https://link.springer.com/10.1007/s00220-021-04064-4}
\BIBentrySTDinterwordspacing

\bibitem{katariya2021GeometricDistinguishabilityMeasures}
\BIBentryALTinterwordspacing
V.~Katariya and M.~M. Wilde, ``Geometric distinguishability measures limit
  quantum channel estimation and discrimination,'' \emph{Quantum Information
  Processing}, vol.~20, p.~78, Feb. 2021. [Online]. Available:
  \url{https://link.springer.com/10.1007/s11128-021-02992-7}
\BIBentrySTDinterwordspacing

\bibitem{belavkin1982$C^$algebraicGeneralizationRelative}
\BIBentryALTinterwordspacing
V.~P. Belavkin and P.~Staszewski, ``{$C^*$}-algebraic generalization of
  relative entropy and entropy,'' \emph{Annales de l'I.H.P. Physique
  th{\'e}orique}, vol.~37, pp. 51--58, 1982. [Online]. Available:
  \url{http://eudml.org/doc/76163}
\BIBentrySTDinterwordspacing

\bibitem{leditzky2018ApproachesApproximateAdditivity}
\BIBentryALTinterwordspacing
F.~Leditzky, E.~Kaur, N.~Datta, and M.~M. Wilde, ``Approaches for approximate
  additivity of the {{Holevo}} information of quantum channels,''
  \emph{Physical Review A}, vol.~97, p. 012332, Jan. 2018. [Online]. Available:
  \url{https://link.aps.org/doi/10.1103/PhysRevA.97.012332}
\BIBentrySTDinterwordspacing

\bibitem{ding2023BoundingForwardClassical}
\BIBentryALTinterwordspacing
D.~Ding, S.~Khatri, Y.~Quek, P.~W. Shor, X.~Wang, and M.~M. Wilde, ``Bounding
  the forward classical capacity of bipartite quantum channels,'' \emph{IEEE
  Transactions on Information Theory}, vol.~69, pp. 3034--3061, May 2023.
  [Online]. Available: \url{https://ieeexplore.ieee.org/document/10005080/}
\BIBentrySTDinterwordspacing

\bibitem{mosonyi2021DivergenceRadiiStrong}
\BIBentryALTinterwordspacing
M.~Mosonyi and T.~Ogawa, ``Divergence radii and the strong converse exponent of
  classical-quantum channel coding with constant compositions,'' \emph{IEEE
  Transactions on Information Theory}, vol.~67, pp. 1668--1698, Mar. 2021.
  [Online]. Available: \url{https://ieeexplore.ieee.org/document/9276455/}
\BIBentrySTDinterwordspacing

\bibitem{audenaert2015A$z$RenyiRelativeEntropies}
\BIBentryALTinterwordspacing
K.~M.~R. Audenaert and N.~Datta, ``{$\alpha$}-{$z$}-{{R{\'e}nyi}} relative
  entropies,'' \emph{Journal of Mathematical Physics}, vol.~56, p. 022202, Feb.
  2015. [Online]. Available:
  \url{https://pubs.aip.org/jmp/article/56/2/022202/97443/z-Renyi-relative-entropies}
\BIBentrySTDinterwordspacing

\bibitem{mosonyi2017StrongConverseExponent}
\BIBentryALTinterwordspacing
M.~Mosonyi and T.~Ogawa, ``Strong converse exponent for classical-quantum
  channel coding,'' \emph{Communications in Mathematical Physics}, vol. 355,
  pp. 373--426, Oct. 2017. [Online]. Available:
  \url{http://link.springer.com/10.1007/s00220-017-2928-4}
\BIBentrySTDinterwordspacing

\bibitem{nuradha2025MultivariateFidelities}
\BIBentryALTinterwordspacing
T.~Nuradha, H.~K. Mishra, F.~Leditzky, and M.~M. Wilde, ``Multivariate
  fidelities,'' \emph{Journal of Physics A: Mathematical and Theoretical},
  vol.~58, p. 165304, Apr. 2025. [Online]. Available:
  \url{https://iopscience.iop.org/article/10.1088/1751-8121/adc645}
\BIBentrySTDinterwordspacing

\bibitem{kosaki1984ApplicationsComplexInterpolation}
\BIBentryALTinterwordspacing
H.~Kosaki, ``Applications of the complex interpolation method to a von
  {{Neumann}} algebra: {{Non-commutative}} {$L^p$}-spaces,'' \emph{Journal of
  Functional Analysis}, vol.~56, pp. 29--78, Mar. 1984. [Online]. Available:
  \url{https://linkinghub.elsevier.com/retrieve/pii/0022123684900259}
\BIBentrySTDinterwordspacing

\bibitem{datta2009MinMaxrelativeEntropies}
\BIBentryALTinterwordspacing
N.~Datta, ``Min- and max-relative entropies and a new entanglement monotone,''
  \emph{IEEE Transactions on Information Theory}, vol.~55, pp. 2816--2826, Jun.
  2009. [Online]. Available: \url{http://ieeexplore.ieee.org/document/4957651/}
\BIBentrySTDinterwordspacing

\bibitem{buscemi2010QuantumCapacityChannels}
\BIBentryALTinterwordspacing
F.~Buscemi and N.~Datta, ``The quantum capacity of channels with arbitrarily
  correlated noise,'' \emph{IEEE Transactions on Information Theory}, vol.~56,
  pp. 1447--1460, Mar. 2010. [Online]. Available:
  \url{http://ieeexplore.ieee.org/document/5429118/}
\BIBentrySTDinterwordspacing

\bibitem{wang2012OneshotClassicalquantumCapacity}
\BIBentryALTinterwordspacing
L.~Wang and R.~Renner, ``One-shot classical-quantum capacity and hypothesis
  testing,'' \emph{Physical Review Letters}, vol. 108, p. 200501, May 2012.
  [Online]. Available:
  \url{https://link.aps.org/doi/10.1103/PhysRevLett.108.200501}
\BIBentrySTDinterwordspacing

\bibitem{cooney2016StrongConverseExponents}
\BIBentryALTinterwordspacing
T.~Cooney, M.~Mosonyi, and M.~M. Wilde, ``Strong converse exponents for a
  quantum channel discrimination problem and quantum-feedback-assisted
  communication,'' \emph{Communications in Mathematical Physics}, vol. 344, pp.
  797--829, Jun. 2016. [Online]. Available:
  \url{http://link.springer.com/10.1007/s00220-016-2645-4}
\BIBentrySTDinterwordspacing

\bibitem{audenaert2007DiscriminatingStatesQuantum}
\BIBentryALTinterwordspacing
K.~M.~R. Audenaert, J.~Calsamiglia, R.~{Mu{\~n}oz-Tapia}, E.~Bagan, {\relax
  Ll}.~Masanes, A.~Acin, and F.~Verstraete, ``Discriminating states: {{The}}
  quantum {{Chernoff}} bound,'' \emph{Physical Review Letters}, vol.~98, p.
  160501, Apr. 2007. [Online]. Available:
  \url{https://link.aps.org/doi/10.1103/PhysRevLett.98.160501}
\BIBentrySTDinterwordspacing

\bibitem{nussbaum2009ChernoffLowerBound}
\BIBentryALTinterwordspacing
M.~Nussbaum and A.~Szko{\l}a, ``The {{Chernoff}} lower bound for symmetric
  quantum hypothesis testing,'' \emph{The Annals of Statistics}, vol.~37, Apr.
  2009. [Online]. Available:
  \url{https://projecteuclid.org/journals/annals-of-statistics/volume-37/issue-2/The-Chernoff-lower-bound-for-symmetric-quantum-hypothesis-testing/10.1214/08-AOS593.full}
\BIBentrySTDinterwordspacing

\bibitem{yuen1975OptimumTestingMultiple}
\BIBentryALTinterwordspacing
H.~Yuen, R.~Kennedy, and M.~Lax, ``Optimum testing of multiple hypotheses in
  quantum detection theory,'' \emph{IEEE Transactions on Information Theory},
  vol.~21, pp. 125--134, Mar. 1975. [Online]. Available:
  \url{http://ieeexplore.ieee.org/document/1055351/}
\BIBentrySTDinterwordspacing

\bibitem{vazquez-vilar2016MultipleQuantumHypothesis}
\BIBentryALTinterwordspacing
G.~{Vazquez-Vilar}, ``Multiple quantum hypothesis testing expressions and
  classical-quantum channel converse bounds,'' in \emph{2016 {{IEEE
  International Symposium}} on {{Information Theory}} ({{ISIT}})}.\hskip 1em
  plus 0.5em minus 0.4em\relax Barcelona, Spain: IEEE, Jul. 2016, pp.
  2854--2857. [Online]. Available:
  \url{http://ieeexplore.ieee.org/document/7541820/}
\BIBentrySTDinterwordspacing

\bibitem{konig2009OperationalMeaningMin}
\BIBentryALTinterwordspacing
R.~K{\"o}nig, R.~Renner, and C.~Schaffner, ``The operational meaning of min-
  and max-entropy,'' \emph{IEEE Transactions on Information Theory}, vol.~55,
  pp. 4337--4347, Sep. 2009. [Online]. Available:
  \url{http://ieeexplore.ieee.org/document/5208530/}
\BIBentrySTDinterwordspacing

\bibitem{mosonyi2009GeneralizedRelativeEntropies}
\BIBentryALTinterwordspacing
M.~Mosonyi and N.~Datta, ``Generalized relative entropies and the capacity of
  classical-quantum channels,'' \emph{Journal of Mathematical Physics},
  vol.~50, p. 072104, Jul. 2009. [Online]. Available:
  \url{https://pubs.aip.org/jmp/article/50/7/072104/231912/Generalized-relative-entropies-and-the-capacity-of}
\BIBentrySTDinterwordspacing

\bibitem{sion1958GeneralMinimaxTheorems}
\BIBentryALTinterwordspacing
M.~Sion, ``On general minimax theorems,'' \emph{Pacific Journal of
  Mathematics}, vol.~8, pp. 171--176, Mar. 1958. [Online]. Available:
  \url{http://msp.org/pjm/1958/8-1/p14.xhtml}
\BIBentrySTDinterwordspacing

\bibitem{petz1985QuasientropiesStatesNeumann}
\BIBentryALTinterwordspacing
D.~Petz, ``Quasi-entropies for states of a von {{Neumann}} algebra,''
  \emph{Publications of the Research Institute for Mathematical Sciences},
  vol.~21, pp. 787--800, Aug. 1985. [Online]. Available:
  \url{https://ems.press/doi/10.2977/prims/1195178929}
\BIBentrySTDinterwordspacing

\bibitem{petz1986QuasientropiesFiniteQuantum}
\BIBentryALTinterwordspacing
------, ``Quasi-entropies for finite quantum systems,'' \emph{Reports on
  Mathematical Physics}, vol.~23, pp. 57--65, Feb. 1986. [Online]. Available:
  \url{https://linkinghub.elsevier.com/retrieve/pii/0034487786900674}
\BIBentrySTDinterwordspacing

\bibitem{mosonyi2011QuantumRenyiRelative}
\BIBentryALTinterwordspacing
M.~Mosonyi and F.~Hiai, ``On the quantum {{R{\'e}nyi}} relative entropies and
  related capacity formulas,'' \emph{IEEE Transactions on Information Theory},
  vol.~57, pp. 2474--2487, Apr. 2011. [Online]. Available:
  \url{http://ieeexplore.ieee.org/document/5730573/}
\BIBentrySTDinterwordspacing

\bibitem{boyd2004ConvexOptimization}
\BIBentryALTinterwordspacing
S.~Boyd and L.~Vandenberghe, \emph{Convex Optimization}, 1st~ed.\hskip 1em plus
  0.5em minus 0.4em\relax Cambridge University Press, Mar. 2004. [Online].
  Available:
  \url{https://www.cambridge.org/core/product/identifier/9780511804441/type/book}
\BIBentrySTDinterwordspacing

\bibitem{hiai1994EqualityCasesMatrix}
\BIBentryALTinterwordspacing
F.~Hiai, ``Equality cases in matrix norm inequalities of golden-thompson
  type,'' \emph{Linear and Multilinear Algebra}, vol.~36, pp. 239--249, Apr.
  1994. [Online]. Available:
  \url{http://www.tandfonline.com/doi/abs/10.1080/03081089408818297}
\BIBentrySTDinterwordspacing

\bibitem{huang2024ExactQuantumSensing}
\BIBentryALTinterwordspacing
Z.~Huang, L.~Lami, and M.~M. Wilde, ``Exact quantum sensing limits for
  {{Bosonic}} dephasing channels,'' \emph{PRX Quantum}, vol.~5, p. 020354, Jun.
  2024. [Online]. Available:
  \url{https://link.aps.org/doi/10.1103/PRXQuantum.5.020354}
\BIBentrySTDinterwordspacing

\bibitem{gutoski2007GeneralTheoryQuantum}
\BIBentryALTinterwordspacing
G.~Gutoski and J.~Watrous, ``Toward a general theory of quantum games,'' in
  \emph{Proceedings of the Thirty-Ninth Annual {{ACM}} Symposium on {{Theory}}
  of Computing}.\hskip 1em plus 0.5em minus 0.4em\relax San Diego California
  USA: ACM, Jun. 2007, pp. 565--574. [Online]. Available:
  \url{https://dl.acm.org/doi/10.1145/1250790.1250873}
\BIBentrySTDinterwordspacing

\bibitem{chiribella2008MemoryEffectsQuantum}
\BIBentryALTinterwordspacing
G.~Chiribella, G.~M. D'Ariano, and P.~Perinotti, ``Memory effects in quantum
  channel discrimination,'' \emph{Physical Review Letters}, vol. 101, p.
  180501, Oct. 2008. [Online]. Available:
  \url{https://link.aps.org/doi/10.1103/PhysRevLett.101.180501}
\BIBentrySTDinterwordspacing

\bibitem{chiribella2008QuantumCircuitArchitecture}
\BIBentryALTinterwordspacing
------, ``Quantum circuit architecture,'' \emph{Physical Review Letters}, vol.
  101, p. 060401, Aug. 2008. [Online]. Available:
  \url{https://link.aps.org/doi/10.1103/PhysRevLett.101.060401}
\BIBentrySTDinterwordspacing

\bibitem{chiribella2013QuantumComputationsDefinite}
\BIBentryALTinterwordspacing
G.~Chiribella, G.~M. D'Ariano, P.~Perinotti, and B.~Valiron, ``Quantum
  computations without definite causal structure,'' \emph{Physical Review A},
  vol.~88, p. 022318, Aug. 2013. [Online]. Available:
  \url{https://link.aps.org/doi/10.1103/PhysRevA.88.022318}
\BIBentrySTDinterwordspacing

\bibitem{bavaresco2021StrictHierarchyParallel}
\BIBentryALTinterwordspacing
J.~Bavaresco, M.~Murao, and M.~T. Quintino, ``Strict hierarchy between
  parallel, sequential, and indefinite-causal-order strategies for channel
  discrimination,'' \emph{Physical Review Letters}, vol. 127, p. 200504, Nov.
  2021. [Online]. Available:
  \url{https://link.aps.org/doi/10.1103/PhysRevLett.127.200504}
\BIBentrySTDinterwordspacing

\bibitem{fawzi2021DefiningQuantumDivergences}
\BIBentryALTinterwordspacing
H.~Fawzi and O.~Fawzi, ``Defining quantum divergences via convex
  optimization,'' \emph{Quantum}, vol.~5, p. 387, Jan. 2021. [Online].
  Available: \url{https://doi.org/10.22331/q-2021-01-26-387}
\BIBentrySTDinterwordspacing

\bibitem{fawzi2017LiebsConcavityTheorem}
\BIBentryALTinterwordspacing
H.~Fawzi and J.~Saunderson, ``Lieb's concavity theorem, matrix geometric means,
  and semidefinite optimization,'' \emph{Linear Algebra and its Applications},
  vol. 513, pp. 240--263, Jan. 2017. [Online]. Available:
  \url{https://linkinghub.elsevier.com/retrieve/pii/S0024379516304852}
\BIBentrySTDinterwordspacing

\bibitem{bjelakovic2005QuantumVersionSanovs}
\BIBentryALTinterwordspacing
I.~Bjelakovi{\'c}, J.-D. Deuschel, T.~Kr{\"u}ger, R.~Seiler,
  R.~{Siegmund-Schultze}, and A.~Szko{\l}a, ``A quantum version of {{Sanov}}'s
  theorem,'' \emph{Communications in Mathematical Physics}, vol. 260, pp.
  659--671, Dec. 2005. [Online]. Available:
  \url{http://link.springer.com/10.1007/s00220-005-1426-2}
\BIBentrySTDinterwordspacing

\bibitem{brandao2010GeneralizationQuantumSteins}
\BIBentryALTinterwordspacing
F.~G. S.~L. Brand{\~a}o and M.~B. Plenio, ``A generalization of quantum
  {{Stein}}'s lemma,'' \emph{Communications in Mathematical Physics}, vol. 295,
  pp. 791--828, May 2010. [Online]. Available:
  \url{http://link.springer.com/10.1007/s00220-010-1005-z}
\BIBentrySTDinterwordspacing

\bibitem{li2016DiscriminatingQuantumStates}
\BIBentryALTinterwordspacing
K.~Li, ``Discriminating quantum states: {{The}} multiple {{Chernoff}}
  distance,'' \emph{The Annals of Statistics}, vol.~44, Aug. 2016. [Online].
  Available:
  \url{https://projecteuclid.org/journals/annals-of-statistics/volume-44/issue-4/Discriminating-quantum-states-The-multiple-Chernoff-distance/10.1214/16-AOS1436.full}
\BIBentrySTDinterwordspacing

\bibitem{brandao2020AdversarialHypothesisTesting}
\BIBentryALTinterwordspacing
F.~G. S.~L. Brandao, A.~W. Harrow, J.~R. Lee, and Y.~Peres, ``Adversarial
  hypothesis testing and a quantum {{Stein}}'s lemma for restricted
  measurements,'' \emph{IEEE Transactions on Information Theory}, vol.~66, pp.
  5037--5054, Aug. 2020. [Online]. Available:
  \url{https://ieeexplore.ieee.org/document/9031743/}
\BIBentrySTDinterwordspacing

\bibitem{bunth2023EquivariantRelativeSubmajorization}
\BIBentryALTinterwordspacing
G.~Bunth and P.~Vrana, ``Equivariant relative submajorization,'' \emph{IEEE
  Transactions on Information Theory}, vol.~69, pp. 1057--1073, Feb. 2023.
  [Online]. Available: \url{https://ieeexplore.ieee.org/document/9917549/}
\BIBentrySTDinterwordspacing

\bibitem{hayashi2024GeneralizedQuantumSteins}
\BIBentryALTinterwordspacing
M.~Hayashi and H.~Yamasaki, ``Generalized quantum {{Stein}}'s lemma and second
  law of quantum resource theories,'' Oct. 2024. [Online]. Available:
  \url{http://arxiv.org/abs/2408.02722}
\BIBentrySTDinterwordspacing

\bibitem{lami2025SolutionGeneralizedQuantum}
\BIBentryALTinterwordspacing
L.~Lami, ``A solution of the generalized quantum {{Stein}}'s lemma,''
  \emph{IEEE Transactions on Information Theory}, vol.~71, pp. 4454--4484, Jun.
  2025. [Online]. Available:
  \url{https://ieeexplore.ieee.org/document/10898013/}
\BIBentrySTDinterwordspacing

\bibitem{furuya2023MonotonicMultistateQuantum}
\BIBentryALTinterwordspacing
K.~Furuya, N.~Lashkari, and S.~Ouseph, ``Monotonic multi-state quantum
  {$f$}-divergences,'' \emph{Journal of Mathematical Physics}, vol.~64, p.
  042203, Apr. 2023. [Online]. Available:
  \url{https://pubs.aip.org/jmp/article/64/4/042203/2878628/Monotonic-multi-state-quantum-f-divergencesFuruya}
\BIBentrySTDinterwordspacing

\bibitem{audenaert2014UpperBoundsError}
\BIBentryALTinterwordspacing
K.~M.~R. Audenaert and M.~Mosonyi, ``Upper bounds on the error probabilities
  and asymptotic error exponents in quantum multiple state discrimination,''
  \emph{Journal of Mathematical Physics}, vol.~55, p. 102201, Oct. 2014.
  [Online]. Available:
  \url{https://pubs.aip.org/jmp/article/55/10/102201/232991/Upper-bounds-on-the-error-probabilities-and}
\BIBentrySTDinterwordspacing

\bibitem{tomamichel2009FullyQuantumAsymptotic}
\BIBentryALTinterwordspacing
M.~Tomamichel, R.~Colbeck, and R.~Renner, ``A fully quantum asymptotic
  equipartition property,'' \emph{IEEE Transactions on Information Theory},
  vol.~55, pp. 5840--5847, Dec. 2009. [Online]. Available:
  \url{http://ieeexplore.ieee.org/document/5319753/}
\BIBentrySTDinterwordspacing

\bibitem{palomar2008LautumInformation}
\BIBentryALTinterwordspacing
D.~P. Palomar and S.~Verd{\'u}, ``Lautum information,'' \emph{IEEE Transactions
  on Information Theory}, vol.~54, pp. 964--975, Mar. 2008. [Online].
  Available: \url{http://ieeexplore.ieee.org/document/4455754/}
\BIBentrySTDinterwordspacing

\bibitem{hayashi2016CorrelationDetectionOperational}
\BIBentryALTinterwordspacing
M.~Hayashi and M.~Tomamichel, ``Correlation detection and an operational
  interpretation of the {{R{\'e}nyi}} mutual information,'' \emph{Journal of
  Mathematical Physics}, vol.~57, p. 102201, Oct. 2016. [Online]. Available:
  \url{https://pubs.aip.org/jmp/article/57/10/102201/904248/Correlation-detection-and-an-operational}
\BIBentrySTDinterwordspacing

\bibitem{sharma2013FundamentalBoundReliability}
\BIBentryALTinterwordspacing
N.~Sharma and N.~A. Warsi, ``Fundamental bound on the reliability of quantum
  information transmission,'' \emph{Physical Review Letters}, vol. 110, p.
  080501, Feb. 2013. [Online]. Available:
  \url{https://link.aps.org/doi/10.1103/PhysRevLett.110.080501}
\BIBentrySTDinterwordspacing

\bibitem{diaz2018UsingReusingCoherence}
\BIBentryALTinterwordspacing
M.~G. D{\'i}az, K.~Fang, X.~Wang, M.~Rosati, M.~Skotiniotis, J.~Calsamiglia,
  and A.~Winter, ``Using and reusing coherence to realize quantum processes,''
  \emph{Quantum}, vol.~2, p. 100, Oct. 2018. [Online]. Available:
  \url{https://doi.org/10.22331/q-2018-10-19-100}
\BIBentrySTDinterwordspacing

\bibitem{regula2021FundamentalLimitationsDistillation}
\BIBentryALTinterwordspacing
B.~Regula and R.~Takagi, ``Fundamental limitations on distillation of quantum
  channel resources,'' \emph{Nature Communications}, vol.~12, p. 4411, Jul.
  2021. [Online]. Available:
  \url{https://www.nature.com/articles/s41467-021-24699-0}
\BIBentrySTDinterwordspacing

\bibitem{regula2024PostselectedQuantumHypothesis}
\BIBentryALTinterwordspacing
B.~Regula, L.~Lami, and M.~M. Wilde, ``Postselected quantum hypothesis
  testing,'' \emph{IEEE Transactions on Information Theory}, vol.~70, pp.
  3453--3469, May 2024. [Online]. Available:
  \url{https://ieeexplore.ieee.org/document/10196477/}
\BIBentrySTDinterwordspacing

\end{thebibliography}

\end{document}